\documentclass[aoas]{imsart}

\RequirePackage{amsthm,amsmath,amsfonts,amssymb}
\RequirePackage[authoryear]{natbib}
\RequirePackage[colorlinks,citecolor=blue,urlcolor=blue]{hyperref}
\RequirePackage{graphicx}

\usepackage{float}

\usepackage[utf8]{inputenc}

\usepackage{tikz}
\usepackage{tikz-network}
\usetikzlibrary{cd}
\usetikzlibrary{backgrounds}
\usetikzlibrary{arrows.meta}
\pgfdeclarelayer{bg}    
\pgfsetlayers{bg,main}

\usetikzlibrary{fit,positioning,arrows.meta,bayesnet,shapes.geometric,chains,matrix,scopes,calc,decorations.markings}

\tikzstyle{param}=[circle, minimum size = 0.7cm, thick, draw=black!100, fill = gray!10, node distance = 0.5cm]
\tikzstyle{data}=[rectangle, minimum size = 0.7cm, thick, draw =black!100, node distance = 0.5cm]
\tikzstyle{model}=[rectangle, minimum size = 1cm, thick, draw=black!100, node distance = 0.5cm]

\definecolor{orange}{rgb}{0.8,0.6,0.0}
\definecolor{purple}{rgb}{0.4,0,0.8}
\definecolor{brightgreen}{rgb}{1.0,0.8,1.0}
\definecolor{indigo}{rgb}{0.1,0,1.0}
\definecolor{Orange}{rgb}{1.0,1.0,0.5}

\startlocaldefs
\theoremstyle{plain}

\theoremstyle{remark}

\newtheorem{proposition}{Proposition}

\def\H{\mathcal{H}}
\def\L{\mathcal{L}}
\def\P{\mathcal{P}}
\def\R{\mathbb{R}}
\def\Sm{D}
\def\Sb{S}
\def\B{\mathcal{B}}
\def\K{{\mathcal{K}}}
\def\e#1#2{\langle #1,#2\rangle} 
\def\I{{\mathcal{I}}}
\def\M{{\mathcal{M}}}
\def\H{{\mathcal{H}}}
\def\A{{\mathcal{A}}}
\def\T{{\mathcal{T}}}
\def\V{{\mathcal{V}}}
\def\X{{\mathcal{X}}}
\def\Z{{\mathcal{Z}}}
\def\elpd{\mbox{elpd}}
\def\eq{Eqn}
\def\eqs{Eqns}
\def\fig{Fig}
\def\figs{Figs}
\def\sec{Sec.}
\def\secs{Secs.}
\def\app{Supp}
\def\apps{Supp}
\usepackage{stackengine} 
\newcommand\op{\stackMath\mathbin{\stackinset{c}{0ex}{c}{0ex}{{}_P}{\bigcirc}}}
\newcommand\os{\stackMath\mathbin{\stackinset{c}{0ex}{c}{0ex}{{}_S}{\bigcirc}}}

\endlocaldefs

\begin{document}

\begin{frontmatter}
\title{Bayesian inference for partial orders from random linear extensions: power relations from 12th Century Royal Acta}
\runtitle{Inference for partial orders from random linear extensions}

\begin{aug}
\author[A]{\fnms{Geoff K. Nicholls}\ead[label=e1]{nicholls@stats.ox.ac.uk}}
\author[B]{\fnms{Jeong Eun Lee}\ead[label=e2]{kate.lee@auckland.ac.nz}}
\author[C]{\fnms{Nicholas Karn}\ead[label=e3]{N.E.Karn@soton.ac.uk}}
\author[D]{\fnms{David Johnson} \ead[label=e4]{david.johnson@spc.ox.ac.uk}}
\author[E]{\fnms{Rukuang Huang} \ead[label=e5]{rukuang.huang@jesus.ox.ac.uk}}
\and
\author[F]{\fnms{Alexis Muir-Watt} \ead[label=e6]{alexis.muirwatt@gmail.com}}
\address[A]{Department of Statistics, The University of Oxford,
UK.\printead[presep={,\ }]{e1}}

\address[B]{Department of Statistics, The University of Auckland,
New Zealand \printead[presep={,\ }]{e2}}
\address[C]{Faculty of Arts and Humanities, University of Southampton,
UK\printead[presep={,\ }]{e3}}
\address[D]{St Peter's college, The University of Oxford,
UK \printead[presep={,\ }]{e4}}
\address[E]{Department of Psychiatry, University of Oxford,
UK \printead[presep={,\ }]{e5}}
\address[F]{Private researcher, London, UK \printead[presep={,\ }]{e6}}
\end{aug}

\begin{abstract}
In the eleventh and twelfth centuries in England, Wales and Normandy, Royal Acta were legal documents in which witnesses were listed in order of social status. Any bishops present were listed as a group. For our purposes, each witness-list is an ordered permutation of bishop names with a known date or date-range. Changes over time in the order bishops are listed may reflect changes in their authority. Historians would like to detect and quantify these changes. 
There is no reason to assume that the underlying social order which constrains bishop-order within lists is a complete order. We therefore model the evolving social order as an evolving partial ordered set or {\it poset}. 

We construct a Hidden Markov Model for these data. The hidden state is an evolving poset (the evolving social hierarchy) and the emitted data are random total orders (dated lists) respecting the poset present at the time the order was observed. This generalises existing models for rank-order data such as Mallows and Plackett-Luce. We account for noise via a random ``queue-jumping'' process. Our latent-variable prior for the random process of posets is marginally consistent. A parameter controls poset depth and actor-covariates inform the position of actors in the hierarchy. 
We fit the model, estimate posets and find evidence for changes in status over time. We interpret our results in terms of court politics. Simpler models, based on Bucket Orders and vertex-series-parallel orders, are rejected. We compare our results with a time-series extension of the Plackett-Luce model. Our software is publicly available.
\end{abstract}


\end{frontmatter}

\section{Introduction}

In rank-order data we are presented with a collection of lists ranking a common set of items from best to worst or first to last. A list might order items according to the preferences of an assessor, or the outcome of a multiplayer game, and may rank all elements in the set or just some subset presented to an assessor. 

In this paper we analyse a time series of $371$ lists recording the order in which $67$ different bishops are named as witnesses to legal documents called {\it Royal Acta}. The data are available online \citep{sharpe14}. These documents date from the eleventh and twelfth century (see \app~\ref{app:data-registration} for example lists). Just a small subset of the bishops are named in any given list 
but each bishop is present in many lists.
In a list the bishops' names were written down by a clerk in an order that is known to reflect status (henceforth status in the context of witnessing, which might differ from status in other social contexts). The status of a bishop was partly determined by seniority and diocese. The first canon of the Council of London in 1075 concerns ecclesiastical precedence: ``...each man shall sit according to his date of ordination, except for those who have more honourable seats by ancient custom or by the privileges of their churches'' \citep{clover79}. However, political standing may have contributed to status, and if it did, then the position of a bishop in the status-hierarchy would not be fully explained by time-in-office and diocese. \cite{russell37} writes ``The names of Eustace, bishop of Ely, and John, bishop of Norwich, frequently appear at the head of the list of bishops, before the names of the bishops
of London and Winchester and of bishops who were consecrated before them... [however] 
...they were very close friends of the king during most of his reign and were frequently at court''. The dioceses of London and Winchester were nominally above Ely and Norwich so here is a case where some political element seems to count for more than seniority or ``honourable seats''. 
This analysis was contested by \cite{haskins38} wrote ``the precedence of witnesses to private grants is too erratic to serve as an index to their respective station''. The question is still open. See \app~\ref{app:lit-rev-history} for a review of recent literature written by historians on this topic.





The set of precedence relations we reconstruct determine a \emph{social hierarchy}. This follows the definition given in \cite{vanWietmarschenNous22}, as the ``socially expected behavior'' for the clerk is to ``value'' by status. 
Precedence relations are transitive inequalities. This holds on social and historical grounds. Precedence is a social dominance relation and transitivity is fundamental to human and animal understanding of dominance \citep{gazes17,vasconcelos08}. \cite{shizuka12} find evidence for transitivity in 84 of 101 published animal dominance data tables, and it is the norm in models for social and organisational hierarchies \citep{friedell67,roberts90}, in early models for general preference relations \citep{bogart73b} and in all analyses which assume an underlying complete order. 

Gathering these observations, we represent the social hierarchy of bishops using a Directed Acyclic Graph (DAG).  Nodes correspond to bishops and a directed edge indicates precedence. The DAG is transitively closed, so it defines a {\it partially ordered set} or ``{\it poset}'' (see \sec~\ref{sec:partial-order-intro}, the terms are interchangeable). This is more general than assuming the social hierarchy is a total order, as would be the case if we fit a Mallows model. Some pairs of bishops may simply be unordered. The poset will be our parameter, the thing we want to estimate.
\fig~\ref{fig:intro-po-example} gives an example of a poset extracted from a larger set of relations reconstructed in \sec~\ref{sec:poHB2aRS6}. 
\begin{figure}
    \centering
    \includegraphics[width=1in, trim=1.2in 1in 0.8in 0.9in]{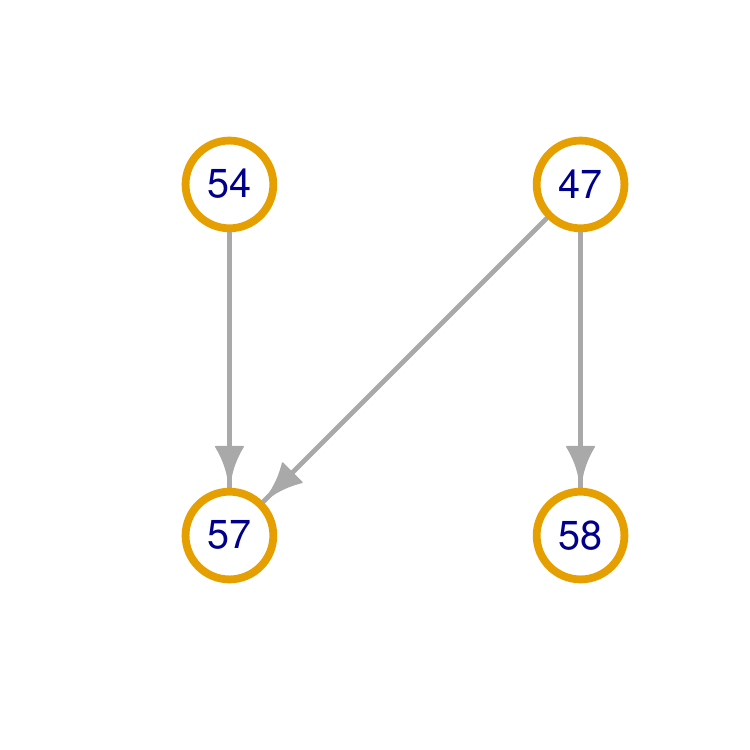}
    \caption{A poset represented as a transitively closed DAG. In \sec~\ref{sec:vsp-bucket-all} we find evidence that this suborder contains all precedence relations between (54) Algar, Bishop of Coutances, (57) Adelulf, Bishop of Carlisle, (47) Simon, Bishop of Worcester and (58) Richard, de Beaufeu, Bishop of Avranches in 1136 (numbering as \fig~\ref{fig:bishop-names}).}
    \label{fig:intro-po-example}
\end{figure}
The witness lists are our data. The poset constrains the data: bishop B shouldn't appear before bishop A in a witness list if B was ordered below A in the poset. In \sec~\ref{sec:queue-jumping-lkd-suborders} we define an observation model which allows the lists to contradict the order relations in the reconstructed poset. However, we view this as as noise. The underlying structure of the social hierarchy of bishops in the context of witnessing is assumed to be a poset.

The first statistical work to model social networks using posets is \cite{mogapi09} who writes ``the application of partial orders in Social Network data has not been studied in the past''. \cite{martin02} used DAGs (which need not be transitively closed) in a similar setting without statistical inference.  \cite{friedell67} used semilattices, a subclass of posets in which each pair of actors has a unique upper bound, to model organisational hierarchies. Their example, the "Cornerville S \& A Club" hierarchy was a poset which only became a semilattice on dropping selected actors. \cite{bogart73b} uses partial orders to express preference relations. They outline a model for evolving posets and remark ``it would be extremely worthwhile to develop a theory of statistical inference for partial orderings''. \cite{fishburn75} mention social dominance relations as a potential application in their work on posets. However, there seems to have been no statistical development of posets as models for social relations prior to \cite{mogapi09} and no other precedent for posets to be used in the way we do. Some of our methods (excluding time series, covariates, and the theory in \sec~\ref{sec:priors-properties}) were outlined in \citep{nicholls11} by two of the present authors and \cite{muirwatt15} gives a continuous-time analysis without covariates. However, despite careful design of the particle filtering Monte Carlo, this approach does not seem promising from modelling and computational perspectives. Recent non-parametric work by \cite{jiang24} shows how the dimension of the latent space parameterisation of posets may be estimated. 

\subsection{Alternative Approaches} \label{sec:intro-lit-alternative-CE-models}
Analysis of rank-order data (ie, our lists) often seeks a total order or actor-ranking which is ``central'' to the lists in the data, so that many lists are summarised by one complete actor ranking. Mallows models \citep{mallows1957non} have a location parameter which is a ranking of the actor labels.
In our setting this parameter would be interpreted as the unknown true bishop order.
A dispersion parameter controls the distribution of the distance between realised lists and this centre-order. In Generalised Mallows models \citep{fligner86} the dispersion parameter can vary across rank positions. There is freedom in the choice of distance measure between ranking lists. \cite{diaconis88} points to Kendall's-tau as having many good properties. \cite{vitelli18} adopt the foot-rule distance on modeling grounds and they and \cite{irurozki19} give methods and software \citep{irurozki16,vitelli20sofware} for computing the dispersion-dependent likelihood normalisation.
Bayesian methods which allow for variability in the quality of the assessors providing the rankings \citep{deng14} have been given. Mixture-model analysis \citep{meila2012dirichlet,tkachenko2016plackett,vitelli18,lu2014effective} can be used when there is a latent group structure in the population from which the lists are drawn, \cite{meila16} gives Bayesian methods for analysing non-parametric mixtures of generalised Mallows models and \cite{vitelli18} gives unsupervised clustering and treats incomplete lists. \cite{asfaw17} give time-series Mallows models. 

The actor skill-vector in Plackett-Luce models \citep{luce1959possible,plackett1975analysis}, which we discuss further in \apps~\ref{app:PL-appendix-section}, determines a rank-order parameter which plays a similar role to the Mallows location parameter. Dispersion in the observation space of lists is controlled by the scale of the skill scores. \cite{hunter04} gave an EM algorithm (see \cite{Caron12mcmc}) for parameter estimation and much useful background theory. Bayesian methods \cite{guiver2009bayesian}, mixture models \citep{Mollica14,Mollica20}, time-series models \citep{caron12time,glickman2015stochastic} and non-parametric Dirichlet process mixture models for clustering \citep{caron12time,caron2014bayesian} have been developed. \cite{Caron12mcmc} give efficient Monte Carlo methods exploiting data augmentation and conjugate priors. The Contextual Repeated Selection (CRS) model \citep{seshadri20} generalises Plackett-Luce to handle non-transitive relations. This generality comes at a price. If we had ten lists of length two in which bishop A precedes B and ten in which B precedes C then A probably precedes C when they appear in the same list in our poset-based analysis. In the simplest CRS model at least, the twenty lists would provide no information about the ordering of A and C in a list containing only A and C.

Our poset-based parameterisation has strengths and weaknesses when compared to these models. Among the weaknesses, our likelihood evaluation does not scale to data with ordered lists of actors longer than about 20 (though scaling with the number of lists is linear). This is discussed further in \sec~\ref{sec:noise-free-lkd}. In response, \cite{jiang23} restrict fitted posets to be Vertex-Series-Parallel posets (VSP,  \cite{valdes78}). These orders and Bucket Orders (in which unordered groups of actors are arranged in a total order) admit likelihood evaluation at a cost which is linear in the list length. Our model comparisons in \sec~\ref{sec:vsp-bucket-all} generally favor posets over VSPs and Bucket Orders. In other respects, our approach is closest to Placket-Luce as our latent variables play a similar role its skill scores.
Our poset-model inherits some of the strengths of Placket Luce: while some Mallows and all Placket-Luce (and our poset) models handle top-$k$ data straightforwardly (where the assessor just ranks their top $k$ preferences, see \sec~\ref{sec:queue-jumping-lkd-suborders}), fitting subset-data (where assessors are presented with different subsets for ranking) is more challenging for Mallows models, as the missing ranks have to be treated as missing data \citep{vitelli18}. This is not necessary in poset or Plackett-Luce models. On the other hand, Placket-Luce models are ``context independent'': the probability for actor A to be listed above actor B is independent of the presence or absence of any other actor. Mallows, CRS and our poset-based models are in general context dependent.
This property of the observation model defined in \sec~\ref{sec:noise-free-lkd} is discussed further in \app~\ref{app:lkd-context-dependence}.

Why is a qualitatively new ranking model needed, given the extensive range of models in the literature? First, we saw above that we have reason to believe the underlying social hierarchy is a partial order; we cannot estimate it reliably without fitting a model in which the parameter is a partial order. 
The alternative models described above impose a total order structure on the hierarchy; this is not justified and not needed.
Secondly, we divide rank-order analyses into two classes: those which aim to reconstruct an underlying true or ``physical'' order and those in which the fitted order is understood as a heuristic summary of the lists. We work in the former setting. However, the models cited in \sec~\ref{sec:intro-lit-alternative-CE-models} can be adapted to make heuristic models for our data. In \app~\ref{app:PL-time-series-model} we specify and fit a Plackett-Luce time-series model with covariates. \cite{glickman2015stochastic,glickman24} define a Plackett-Luce model with many of the same features. Some conclusions from our poset-analysis can be obtained by fitting this relatively simpler model. However, our poset-model is simply a better model for our data, in the sense of goodness of fit.
In \app~\ref{app:pl-mix-elpd} we estimate the Expected Pointwise Log Posterior Predictive (ELPD) model-selection measure \citep{Vehtari17} using Leave-One-Out Cross Validation (LOOCV) for our model and a Plackett-Luce mixture model \citep{Mollica14}. Our model is preferred.

\subsection{Statistical work with Partial Orders}

The first statistical methods inferring partial orders from list data were given in \citet{mannila00} and \citet{mannila06}. They treat problems of seriation in archaeology and biochronology in palaeontology and work with the VSP and Bucket Order sub-classes of posets for rapid evaluation of the ``noise free'' likelihood (see \sec~\ref{sec:noise-free-lkd}). \citet{mannila08} gives a Bayesian analysis for Bucket Orders. 
In other important early work  \citet{beerenwinkel2007conjunctive} define maximum likelihood posets and give a Bayesian analysis in \citet{beerenwinkel12}. They fit a probabilistic graphical model in which genetic mutations accumulate in a total order constrained by a poset. In related work \cite{froelich07} model signaling pathways for gene expression and fit their models using simulated annealing. 

In these archaeological and genetic settings there is an unknown true underlying poset and the data are total orders respecting that poset so we have the same data type and similar inferential goals. Our new contributions are as follows. In \sec~\ref{sec:queue-jumping-lkd-suborders} we give a generative model for lists which allows for noise in the realised lists. We idealise the list-observation model as a snapshot of a ``queue''. In \sec~\ref{sec:prior-main}, building on work by \citet{winkler1985random}, we give marginally consistent priors with a hyper-parameter controlling the ``depth'' of a random poset. Depth is a quantity of historical interest, so our prior should be non-informative with respect to depth. This rules out the uniform prior over posets, taken by \cite{beerenwinkel12}, as it is strongly informative of depth (see \app~\ref{app:prior-simulation}). In \sec~\ref{sec:prior-prob-dbns} we bring covariates into our model. We have a ``linear predictor'' which determines an actor's position in the poset. Finally, our list-data are a time series, so the generative model in \sec~\ref{sec:time-series-posterior} is a Hidden Markov Model (HMM) with a latent process of posets and ``emitted data'' which are lists of actors respecting the poset at the time each list was formed.

Posets appear in a range of data-analytic settings. In \citet{mogapi09} the data are edges in a directed graph. The edges are noisy observations of the relations in an underlying poset representing information flow in a company, and a prior controls the number of relations in the order, like our focus on prior depth. This prior is not marginally consistent. \citet{mannila06} encode list data as a precedence matrix giving the proportion of times any pair of items appear in a given order. A poset has a corresponding precedence matrix, estimated using random {\it linear extensions} (lists which are total orders respecting the poset). The ``distance'' between lists and a poset is the distance between their precedence matrices. The estimated poset is a Bucket Order minimising this distance. \citet{arcagni22} has poset data and a wider range of otherwise similar loss functions. They fit both posets and Bucket Orders.


In \cite{rising21} the poset is a summary statistic, displaying order relations between parameter estimates. Posets are also used for structure discovery in Bayesian Networks \citep{koivisto16a,koivisto16b}, where Bucket Orders support evaluation of marginal likelihoods. The likelihood is written as a sum over total orders respecting a Bucket Order and samples are reweighted by the order-count of the poset in an importance-sampling setup. The same count appears in our likelihood and we evaluate it using the same {\it lecount()} package \citep{koivisto19}. However, in our setting the poset is a parameter of interest, not a supporting structure in the computation.

\subsection{Contributions and plan}

Our main contribution is our analysis in \sec~\ref{sec:poHB2aRS6} of the bishop-list data. Our reconstruction of the evolving social hierarchy in \sec~\ref{sec:main-results} is the first statistical analysis of this kind of data. We answer some longstanding questions. How important was seniority in determining precedent? Did court politics play a role? Our poset models are also new (extending \cite{mannila08} and \cite{beerenwinkel12} as detailed above). Although our methodology was motivated by one particular data set, we nevertheless propose our poset-based ranking models as potentially useful for the analysis of rank data more broadly, to stand alongside Mallows, Plackett-Luce and other models for rank data.  

In \sec~\ref{sec:data} we describe the list data, their associated dates and a seniority covariate on the bishops which informs their position in the hierarchy. \sec~\ref{sec:models-and-inference} begins by setting out notation and defining partial orders and how they constrain the lists we actually observe. We motivate and define the observation model for lists in \sec~\ref{sec:parameters-obs-model-lkd} and then in \sec~\ref{sec:prior-main} give the prior. This is where we define ``status'' and how status is mapped to preference. The generative model and posterior are given \sec~\ref{sec:time-series-posterior}.
In \sec~\ref{sec:priors-properties} we show that our priors are marginally consistent and have support on every poset. We give a brief outline of our MCMC and define some useful summary statistics in \sec~\ref{sec:comp-methods}, relegating the detail to \app~\ref{app:mcmc}.
In \sec~\ref{sec:main-results} we present the results of our Bayesian-MCMC analysis. We begin in \sec~\ref{sec:poHB1aRS6} with a sanity-check: we drop an order constraint on the seniority covariate-effects which historians expect to hold and show the order is recovered. In \sec~\ref{sec:poHB2aRS6} we present our main results with the constraint now imposed. Results are discussed from a historical perspective in \sec~\ref{sec:discussion-of-results}. In \sec~\ref{sec:vsp-bucket-all} we make model comparisons with other methods (fitting VSP and Bucket-Order models). Further comparison with variants of Plackett-Luce are given in \app~\ref{app:PL-appendix-section}.
These favor our poset-model, though VSP and Bucket-Order models do quite well. We summarise our contribution in \sec~\ref{sec:conclusions} and point to future work. A supplement discusses data registration, properties of the observation model and prior, MCMC, further results, model comparisons and results on synthetic data.
 
\section{Introduction to the data}\label{sec:data}

\subsection{Context}\label{sec:data-background}
This study draws on an accumulated dataset, accessed through the database made for ‘The Charters of William II and Henry I’ project by the late Professor Richard Sharpe and Dr Nicholas Karn \citep{sharpe14}.
Some historical background on the data is given in \app~\ref{app:data-sources}. Each witness list in the data is an ordered list of names of individual witnesses taken
from a single legal document or ``act'' (collectively ``acta''). A typical example (with List id 2364) is given in \app~\ref{app:data-registration}. 

We have 1610 witness lists dated between 1066 CE and about 1166 CE involving 1760 individuals. A witness list is a ``snapshot'' created at a single event on a single day. We assume distinct lists are generated independently (for example, we see no evidence for a pair of Acta created at the same time with identical lists). Acta are witnessed in order of social rank from the king or queen, archbishop, bishops (as a group), earls (as a group, may precede bishops) and so on down through society. 
Historians ask if the order in which bishops appear within the their sub-list reflects their evolving personal authority. As we focus on the bishop-hierarchy we extract from the data the sub-lists of bishops. Many of the resulting lists contain less than two bishops and these are discarded as not informing relations between bishops. 
Data processing is set out in detail in \app~\ref{app:data-registration}. 

We take time as discrete by year as the data gives dates rounded to the year. Further coarsening would mask recoverable structural change. 
Outside the range $[B=1080,E=1155]$ CE the lists are sparse so we focus on this interval, covering the reigns of William II, Henry I and Stephen and $T=E-B+1=76$ years. The period is long enough for us to witness changes in the status of individual bishops, but short enough for there to be some hope of temporal homogeneity in the social conventions mapping status to witness list.

Data registration is detailed in \app~\ref{app:data-registration-discuss}. It leaves us with $N=371$ lists, dated between $1080$ and $1155$, and containing two or more bishops. We refer to this as the ``full data'' (we look at shorter time intervals in our goodness-of-fit). Each of the $M=67$ bishops in at least two lists is assigned a numerical index from $1$ to $M$ in \fig~\ref{fig:bishop-names} in \app~\ref{app:poHB2aRS6}. 

\subsection{Data notation}\label{sec:notation-data}

Let $\I=\{1,...,N\}$ and $\M=\{1,...,M\}$ be the sets of list and bishop labels respectively. Each list $y_i=(y_{i,1},...,y_{i,n_i}),\ i\in \I$ is an ordered list of $n_i$ bishops, so that $y_i\subset \M$, with bishop-$y_{i,1}$ first in the list, $y_{i,2}$ second and so on. Not all bishops appear in every list so the list is conditioned on ``attendance''. Let $o_i=\{y_{i,1},...,y_{i,n_i}\}$ be the {\it unordered} list of bishops in list $i$. Let $y=\{y_1,...,y_N\}$ and $o=\{o_1,...,o_N\}$.

We suppose the lists were generated by a process running over an interval of time $[B,E]=\{B,B+1,\dots,E\}$. For $i\in \I$ let $\tau_i\in [B,E]$ give the time at which list $i$ was created. This is sometimes uncertain. However, bounds $\tau^-_i\le \tau_i\le \tau^+_i$ are available. We estimate the missing list dates, taking a uniform prior on $\tau_i\in [\tau^-_i,\tau^+_i]$.
Let $\tau=(\tau_1,...,\tau_N)$, $\tau^\pm_i=(\tau^-_i,\tau^+_i)$ and $\tau^{\pm}=(\tau^{\pm}_1,...,\tau^{\pm}_N)$. 

Bishops enter the social hierarchy when consecrated and leave when they die. For $j\in \M$ we have \citep{FastiEcclesiaeAnglicanae} dates of consecration $b_j<E$ and death $e_j>B$ for each bishop.
The distribution of these intervals can be seen in \app~\ref{app:data} in \fig~\ref{fig:lists-and-bishops} at left. The intervals match the list date-ranges $(\tau^-_i,\tau^+_i)$ so that no bishop appears in a list when not in post. At any given time some dioceses may be empty. \fig~\ref{fig:diocese-activity} shows the presence and absence of bishops by diocese.
For $t\in [B,E]$ let $\M_{t}=\{j\in\M: b_j\le t\le e_j\}$ give the set of bishops active at time $t$, let $m_t=|\M_t|$ give the number of active bishops at time $t$ and let
\begin{equation}\label{eq:S}
    \Sm=\max_{t\in [B,E]} m_t
\end{equation}
give the greatest number active at any time ($D=22$, in 1133). This quantity plays a role in bounding the required dimension of the model parameter space.



\subsection{Witness list data}\label{sec:data-witness-lists}

Bishops from thirty one dioceses appear in the data (including one, Tusculum, from Italy, dropped from the data). They are listed in \app~\ref{app:data-dioceses} and can be seen at the left side of \fig~\ref{fig:diocese-activity}. 
\begin{figure}
    \centering
    \hspace*{-0.2in}\includegraphics[width=5.5in, trim=0.3in 0.5in 0.1in 0.85in, clip]{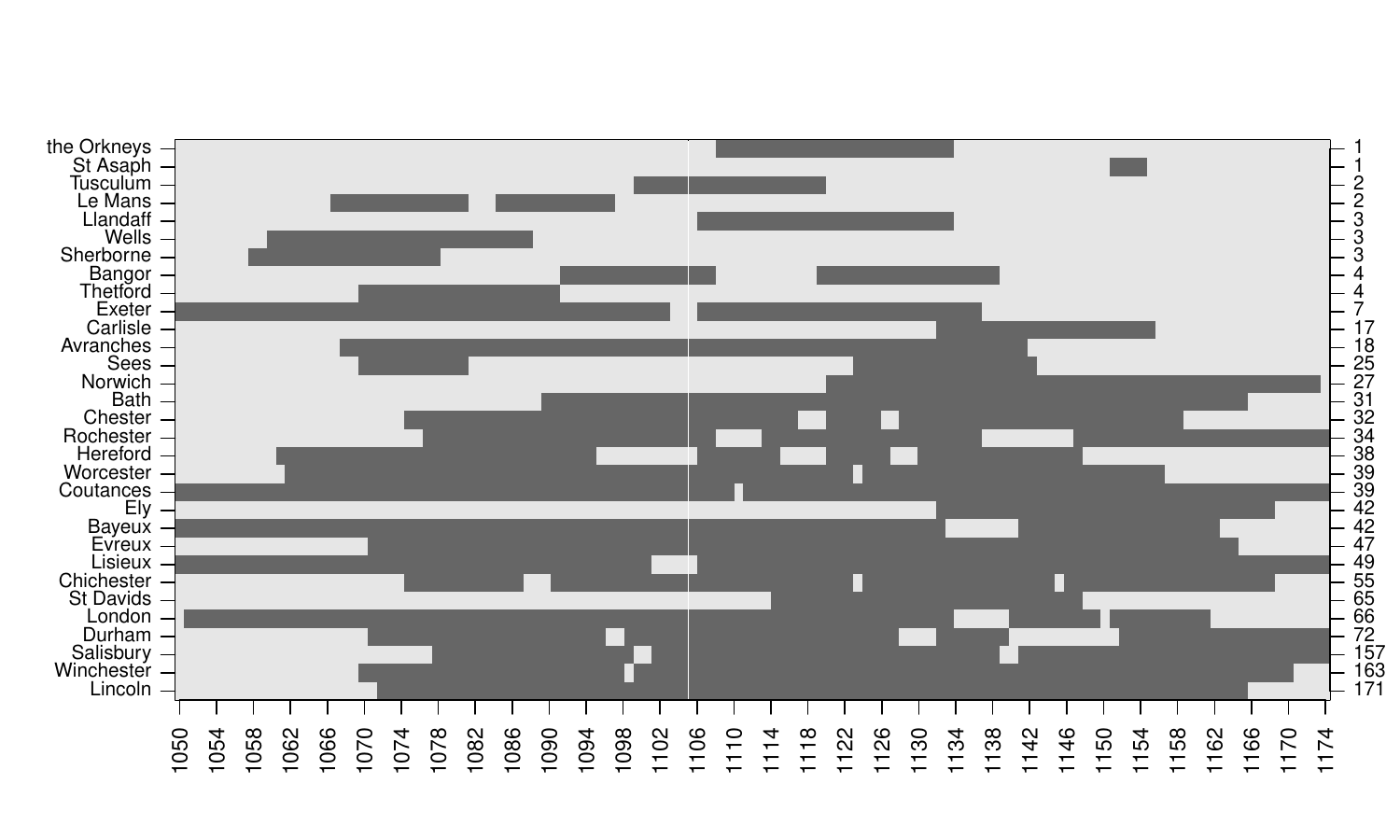}
    \caption{Left axis gives dioceses. Right axis gives the number of lists in which each diocese appears. The $x$-axis gives the date in years. In each year a diocese may have a bishop in-post (dark cell) or be unoccupied (light). }
    \label{fig:diocese-activity}
\end{figure}
The dates of $|\{i\in \I:\ \tau^+_i-\tau_i^->1\}|=212$ lists are uncertain (the mean interval length is 4 years, and 90\% span less than 10 years).
Date intervals $[\tau^-_i,\tau^+_i]$ are plotted in \app~\ref{app:data} in \fig~\ref{fig:lists-and-bishops} at right. 
\fig~\ref{fig:list-dates-lengths} plots lists and their lengths againt their dates (using the midpoint of $[\tau^-_i,\tau^+_i],\ i\in \I$).
\begin{figure}[b]
    \centering
    \hspace*{0in}\includegraphics[width=5.5in, trim=0in 0.2in 0.2in 0.75in, clip]{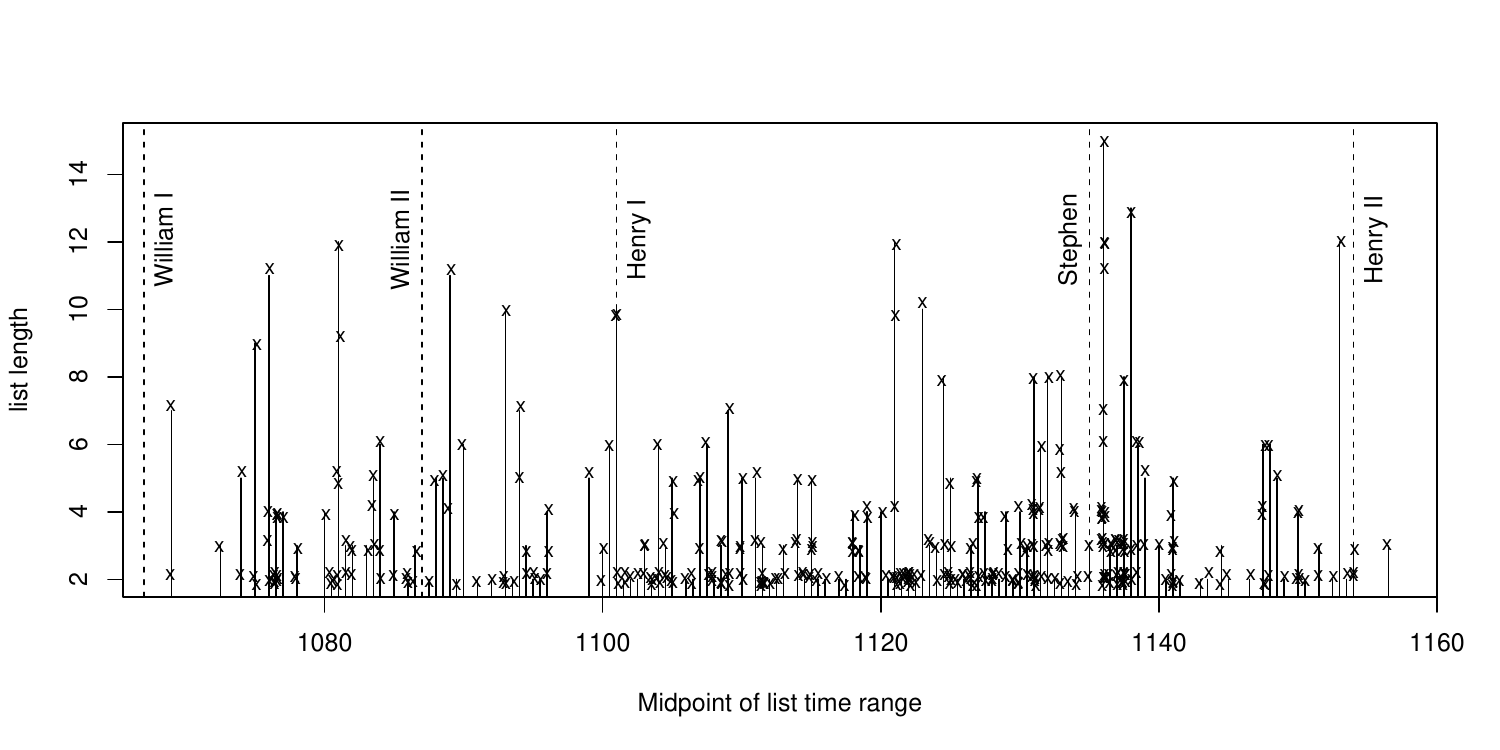}
    \caption{List lengths and dates. Dashed vertical lines are coronation dates of Kings, bar heights are longest lists at that date, with a (jittered) cross for each list plotted at (date,length).}
    \label{fig:list-dates-lengths}
\end{figure}
We include a list if at least half its interval falls within the 76-year interval $[B,E]$; most of the lists in our analysis fall entirely within it.

Our information about a bishop's status is limited by the number of lists in which a bishop appears. Longer lists are more informative as they inform relations between many pairs of bishops. \fig~\ref{fig:list-per-bishop-per-list} shows the distribution of the number of lists a bishop appears in and the distribution of list lengths. 
\begin{figure}
    \centering
    \hspace*{-0.15in}\includegraphics[width=5.5in, trim=0in 0.2in 0.5in 0.75in]{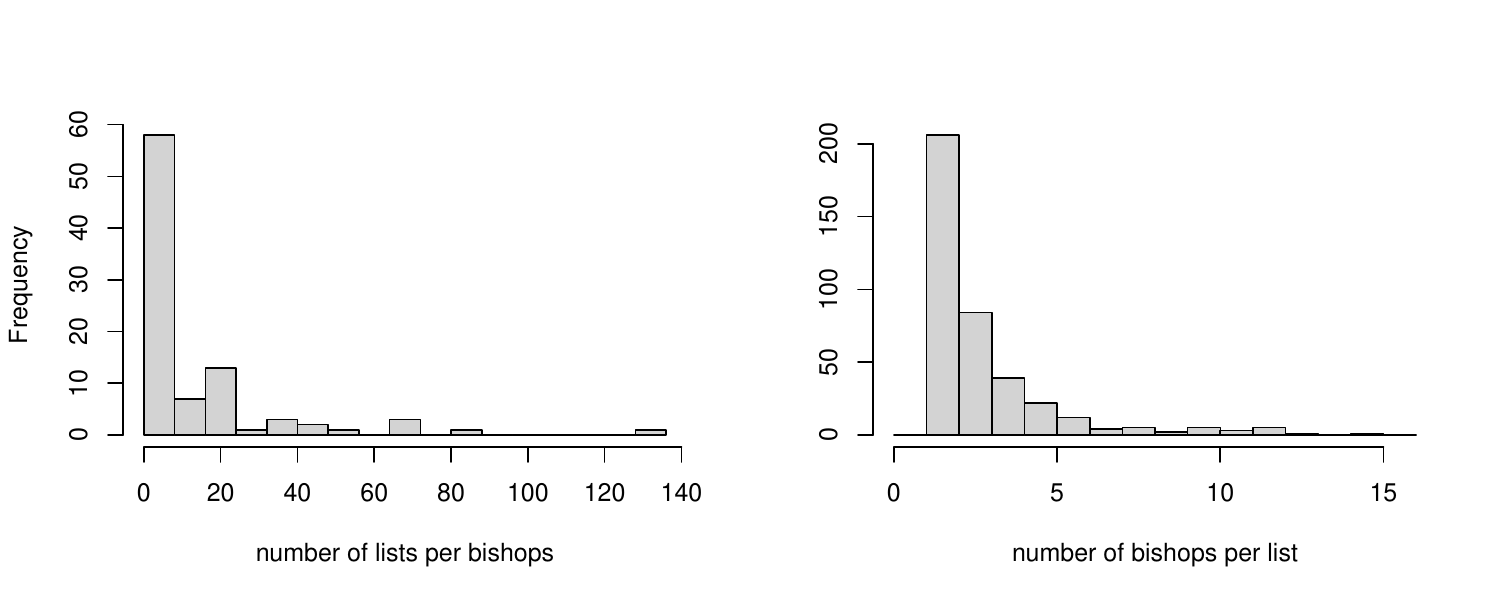}
    \caption{Frequency of lists per bishop, a histogram of the counts $\sum_{i\in\I}\mathbb{I}_{j\in o_i}$, is plotted at left. Distribution of list lengths, a histogram of the counts $n_i=|o_i|$, is plotted at right.}
    \label{fig:list-per-bishop-per-list}
\end{figure}
Most bishops appear in a small number of lists and most lists are relatively short. However, lists ``link together''. If two bishops $j_1, j_2$ do not appear in any list together, but $j_1$ comes before $j_3$ in some list and $j_3$ before $j_2$ in another, then this is evidence for $j_1$ having higher status than $j_2$. This evidence accumulates over lists.

\subsection{Seniority covariates}\label{sec:data-covariates-seniority}
A bishop's ``seniority'' may have contributed to their overall status \citep{clover79} so it is a covariate in our model for the hierarchy of bishops. In year $t$ the longest serving bishop has seniority-rank one and the last appointed bishop has rank at most $m_t$.
Denote by $s_{t,j}\in \{1,...,m_t\}$ the seniority-rank of bishop $j\in\M_t$ in year $t\in \{b_j,b_j+1,...,e_j\}$. We define seniority-rank as
\begin{equation}\label{eq:seniority-s-covariate-definition}
    s_{t,j}=\sum_{k\in \M_t} \mathbb{I}_{b_j\ge b_k}.
\end{equation}
Bishops have equal seniority if there are ties in the start dates $b_j,\ j\in \M$.  
The greatest seniority observed, $\Sb=\max_{t,j} s_{t,j}$, is less than or equal $\Sm$, the most bishops active. This is a second quantity which informs the dimension of the model we fit.

\fig~\ref{fig:seniority-rank-stuff} shows (at left) seniority-rank traces for each bishop from their first to last year in post. Bishops progress in rank by about one place every one or two years, more rapidly at first, as there are more bishops ahead of them. 
Our estimates of the effect of possessing seniority-rank $r\in\{1,...,\Sb\}$ depends for precision on a bishop with seniority $r$ appearing in a reasonable number of lists, so we plot (\fig~\ref{fig:seniority-rank-stuff}, right) the occurrence frequency $f_r=\sum_{i\in \I}\sum_{j\in o_i}\mathbb{I}_{s_{\tau_i,j}=r}$ against $r$ to see which levels of the covariate are well represented in the data. 
\begin{figure}
    \centering
    \hspace*{-0.15in}\includegraphics[width=5.5in, trim=0in 0.2in 0.5in 0.75in, clip]{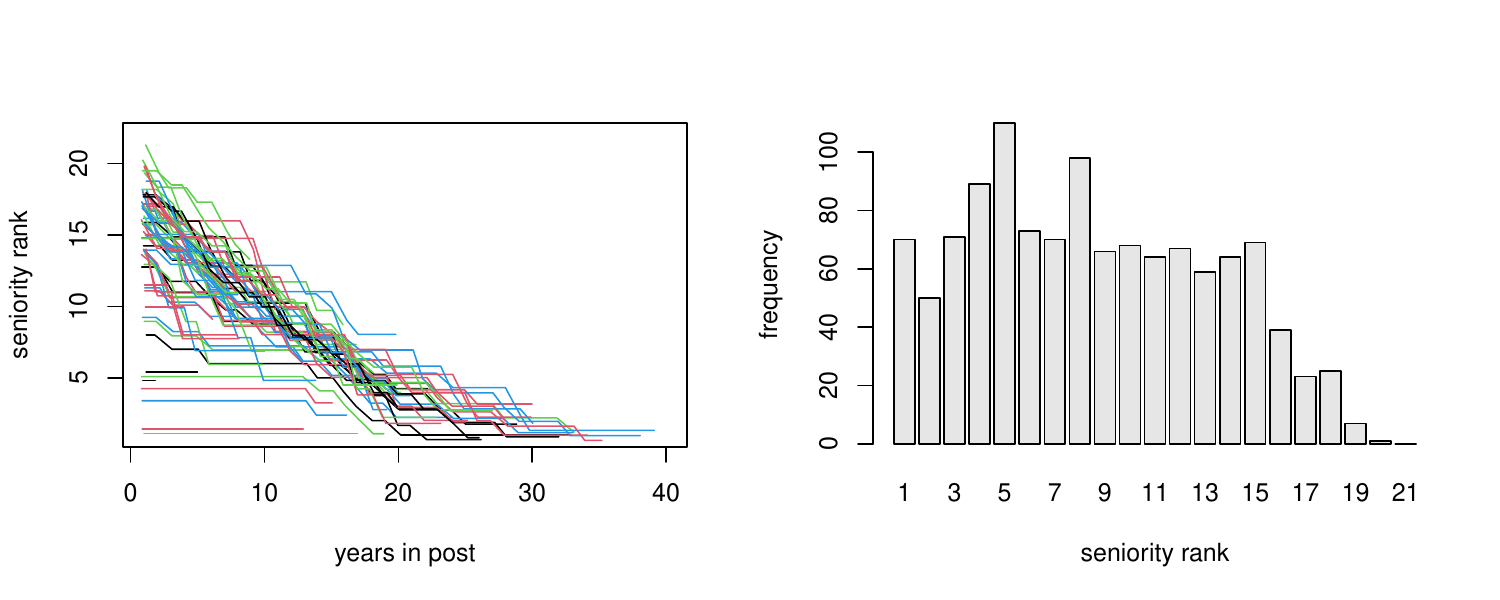}
    \caption{Seniority rank covariate $s$ defined in \eq~\ref{eq:seniority-s-covariate-definition}. Rank $s_{t,j}$ is plotted against ``years in post'', $t-b_j$, for each bishop $j\in\M$ (Left). The frequency of seniority-rank $r$ in lists is plotted against $r$ from 1 to $\Sb$ (Right).}
    \label{fig:seniority-rank-stuff}
\end{figure}
Some dioceses were more peaceful and wealthy than others, so we considered taking diocese label as a second covariate for status. However, diocese would be colinear with bishop label, as each bishop only occupies one diocese in the period of study. An effect due to diocese would not be identifiable with the effects due to the bishops in that diocese.

\subsection{Key questions}

The key question is whether the status of bishop $j\in\M$ in year $t\in [B,E]$ was determined by their seniority $s_{t,j}$ in that year and their diocese or whether the personal authority or position in court of bishop $j$  played a role. In terms of the list data $y$, are the positions of bishops in a list determined by seniority and diocese, with any variation away from this order just unstructured noise, or are the variations away from precedence rules structured by some kind of latent authority? The status of a bishop determines their place in the hierarchy, so what did the hierarchy look like in any given year? We answer these questions by building a statistical model for the bishop-list data.

\section{Models and Inference}\label{sec:models-and-inference}

Our model, in which lists gathered in year $t$ respect a social hierarchy which is known and respected by all but subject to occasional change, expresses the evolving social hierarchy. We present our model as a description of relations between actors, in the usual terminology of social network analysis. However, it can be applied to the analysis of any ranking list data, with or without time-series structure: \cite{jiang23} shows that a related fixed-time model is a good fit for Formula~1 race outcomes. 

\subsection{Parameters and observation model}
\label{sec:parameters-obs-model-lkd}

\subsubsection{Partial orders and linear extensions}\label{sec:partial-order-intro}
In this section we define precedence relations using posets. We drop the time dependence and consider a single generic observation. Suppose we have $m$ actors with labels in $[m]=\{1,...m\}$. 
We represent the unknown true order relations between actors as a poset on $[m]$. \cite{brightwell93} gives an overview of models for random posets and is the source for much of what follows. A strong partial order $\succ_H$ on the ground set $[m]$ is a set of acyclic, transitively closed relations $i\succ_H j$ on the elements of $i,j\in [m]$. Ties $i\sim_H j$ are excluded.
The relations in $H$ are transitively closed if $i\succ_H j$ and $j\succ_H k$ implies $i\succ_H k$. The order is only partial as some elements are not ordered.  
A poset is a total order if $i\succ_H j$ or $j\succ_H i$ for every pair $i,j\in [m]$.

Partial orders on $[m]$ are one to one with transitively closed directed acyclic graphs (DAGs) with vertex labels $1,...,m$, one vertex for each of the $m$ actors, so $\succ_H$ is represented by a DAG $(H,[m])$ with edge set 
\[
H=\{\e{i}{j}\in [m]\times [m]: i\succ_H j\}.
\]
See the example in \fig~\ref{fig:po-example}. We refer to transitively closed DAGs as if they were posets, as they correspond one to one. We can identify a poset by its edge set $H$ as the edge set will be random while the vertex labels $[m]$ which define the ground set are always fixed.
Since posets are edge sets we can take intersections of 
posets. This gives the poset with all relations shared by the intersected posets. The \emph{dimension} of a poset $H\in\H_{[m]}$ is the smallest number of total orders which intersect to give $H$. 

Let $\H_{[m]}$ be the set of all transitively closed DAGs on $[m]$ and let $H\in \H_{[m]}$ be a generic poset. For plotting purposes the transitive reduction is convenient. This is the unique DAG obtained from $H$ by removing all edges implied by transitivity. The depth of a social hierarchy is of interest in many applications. The depth $d(H)$ of $\succ_H$ is the length of the longest path on the DAG $H$, so $d:\H_{[m]}\rightarrow [m]$. The poset in \fig~\ref{fig:po-example} has $d(H)=4$.  

\begin{figure}
    \centering
    \hspace*{-0.1in}{\includegraphics[width=4in, trim=0.1in 2.2in 0.3in 2.25in, clip]{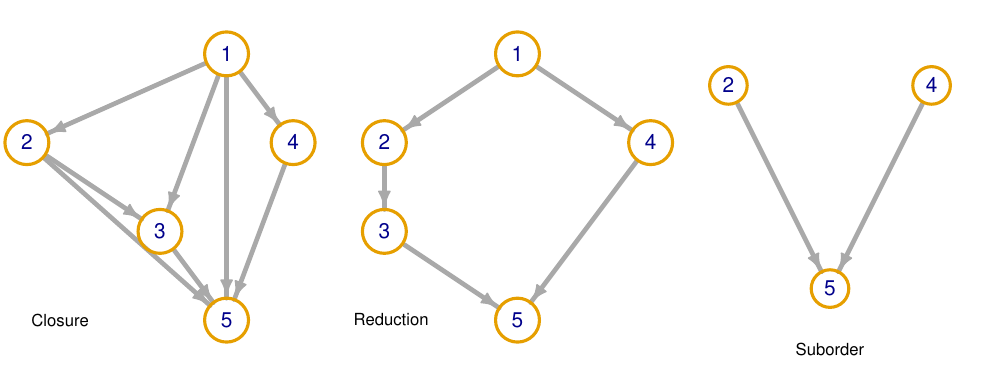}}
    \caption{Poset $H=\{\e12,\e13,\e14,\e15,\e23,\e25,\e35,\e45\},\ H\in\H_{[m]}$ on $[m]=(1,2,3,4,5)$ represented by its transitively closed directed acyclic graph (left) and transitive reduction (centre). Suborder $H[O]$ for $O=\{2,4,5\}$ (right). The longest path on $H$ is $1\to 2\to 3\to 5$ so the depth is $d(H)=4$.}
    \label{fig:po-example}
\end{figure}

Let $\P_{[m]}$ be the set of all permutations of $[m]$. A {\it linear extension} of $H$ is any list $\ell=(\ell_1,...\ell_m),\ \ell\in\P_{[m]}$ in which lesser entries come after greater entries, so $\ell_j\succ_H \ell_k$ is not allowed if $k<j$. For example, if $H$ is the poset displayed in \fig~\ref{fig:po-example} then $H$ has three linear extensions, $(1,2,3,4,5), (1,2,4,3,5)$ and $(1,4,2,3,5)$. Denote by 
\[
\L[H]=\{\ell\in \P_{[m]}: \e{\ell_j}{\ell_k}\not\in H\ \mbox{for all}\ 1\le k<j\le m\}
\]
the set of all linear extensions of $H\in \H_{[m]}$.

If we start with a social hierarchy $H\in\H_{[m]}$ over all actors then the hierarchy constraining any given subset $O\subseteq [m]$ of the actors is the  {\it suborder}
\begin{equation}\label{eq:suborder-defn}
    H[O]=\{\e{j_1}{j_2}\in H: \{j_1,j_2\}\subseteq O\}
\end{equation}
obtained by retaining edges between vertices in $O$. If $H$ is a poset then so is $H[O]$. For example, in \fig~\ref{fig:po-example}, if $O=\{2,4,5\}$ then suborder $H[O]$ is the three-vertex DAG at right.

\subsubsection{Lists as randomly ordered queues}\label{sec:noise-free-lkd}

A single generic list $Y=(Y_1,...,Y_m)$ is modeled as a random linear extension of $H$. 
Let $C(H)=|\L[H]|$ be the number of linear extensions of poset $H\in \H_{[m]}$. The ``noise free'' likelihood for $H$ is simply
\begin{equation}\label{eq:lkd-noise-free}
p(Y|H)=C(H)^{-1}\mathbb{I}_{Y\in \L[H]}.
\end{equation}
All lists which respect the social hierarchy are equally likely. This ``context dependent'' observation model (see \app~\ref{app:lkd-context-dependence}) is motivated by thinking of each list as a realisation of a random queue process in which actors not ordered by $H$ randomly swap places.
The equilibrium of this process is the uniform distribution on linear extensions of $H$ \citep{Karzanov91,jiang24}, so if $Y$ is a snapshot of this queue at equilibrium then $Y\sim \mbox{Unif}(\L[H])$.
Witness lists were written down by a royal scribe with (we assume) perfect knowledge of the hierarchy $H$, so the queue model is an idealisation. However, we arrive at the same model if we assume the clerks regarded all orders not conflicting the hierarchy as equally likely.

Computation of $C(H)$ is \#P-complete \citep{brightwell1991counting} so no polynomial time algorithm for computing $C(H)$ is available, or is likely to exist, and evaluating $p(Y|H)$ is prohibitive at large $m$. However, in our data $m$ is small enough to allow likelihood evaluation in reasonable time. For $j\in [m]$ let $\L_j[H]=\{\ell\in \L[H]: \ell_1=j\}$ be the set of linear extensions with $j$ first in the list and let $C_j(H)=|\L_j[H]|$. Partitioning on the first entry,
\[
C(H)=\sum_{j=1}^m C_j(H),
\]
which \cite{knuth1974structured} compute using the suborder recursion in \app~\ref{app:counting-linear-extensions}.

\subsubsection{Queue-jumping observation model for suborders}\label{sec:queue-jumping-lkd-suborders}
In this section we define the observation model we use in our analyses. We modify the likelihood in \eq~\ref{eq:lkd-noise-free} in three ways: lists may be ``noisy''; just a subset of actors are in any given list; lists are observed over time.

We allow for noise in the observation model by allowing individuals to ``jump the queue''. A list $Y$ is formed by taking individuals from the top of a queue which continues to mix rapidly, constrained by the suborder on those remaining. Before the $j$'th actor is chosen,
there are $m-j+1$ individuals (with labels $Y_{j:m}$) yet to be placed. With probability $p$ the next actor (ie $Y_j$) is chosen at random, ignoring any order constraints. Otherwise, $Y_j$ is chosen as the first actor in a random linear extension of the suborder $H[Y_{j:m}]$ for those remaining. The fraction of lists headed by $Y_j$ is $C_{Y_j}(H[Y_{j:m}])/C(H[Y_{j:m}])$. Working from the top down,
\begin{align}
  p_{(D)}(Y|H,p)&=\prod_{j=1}^{m} p_{(D)}(Y_j|H[Y_{j:m}],p)\nonumber \\
  &=\mathbb{I}_{Y\in \P_{[m]}}\prod_{j=1}^{m}\left( \frac{p}{m-j+1}+(1-p)\frac{C_{Y_j}(H[Y_{j:m}])}{C(H[Y_{j:m}])}\right).\label{eq:lkd-noisy-down}
\end{align}
Noise allows any list to appear with non-zero probability. This is a ``repeated selection'' model \citep{seshadri20} in which the next actor is chosen sequentially from those that remain.
It follows that the correct likelihood for top-$k$ data (where an assessor sees all $m$ actors but just ranks their top $k<m$ actors) simply stops the product in \eq~\ref{eq:lkd-noisy-down} at $j=k$.




We take $p\sim \mbox{Beta}(1,\delta)$ as our family of priors for the queue-jumping probability $p$. The prior hyperparameter $\delta\ge 1$ is fixed (for example, in \sec~\ref{sec:main-results} we take $\delta=9$, so the prior probability for a queue-jumping event is about ten percent), expressing the belief that if order relations are present then they are respected. 

In \eq~\ref{eq:lkd-noisy-down} individuals are promoted up the queue. We can also model random ``demotion''. In this case the list is filled from the bottom up: with probability $p$ the next actor is chosen at random, ignoring any order constraints; otherwise, they are the last entry in a random linear extension of the suborder for the remaining individuals. The likelihood becomes
\begin{equation}\label{eq:lkd-noisy-up}
  p_{(U)}(Y|H,p)=\mathbb{I}_{Y\in \P_{[m]}}\prod_{j=1}^{m}\left( \frac{p}{j}+(1-p)\frac{\tilde C_{Y_j}(H[Y_{1:j}])}{C(H[Y_{1:j}])}\right),
\end{equation}
where $\tilde C_{Y_j}(H)$ is the number of linear extensions of $H$ which end with $Y_j$.
\cite{jiang23} extend these models to allow random promotion and demotion in a single realised list. The likelihood is tractable, but evaluation is time consuming so we do not fit that model here.

The noise free case is obtained from both \eqs~\ref{eq:lkd-noisy-down} and \ref{eq:lkd-noisy-up} at $p=0$. One consequence is that the noise-free model is also a repeated selection model. See \app~\ref{app:noise-free-special-case} for the proof. 
\begin{proposition}\label{prop:noise-free-special-case}
   The noise-free and noisy models (\eqs~\ref{eq:lkd-noise-free}, \ref{eq:lkd-noisy-down} and \ref{eq:lkd-noisy-up}) all coincide at $p=0$. 
\end{proposition}

We now consider what happens when just a subset of actors are present. Relations are given by their suborder and the observation model applies for lists realised on suborders. Suppose that, when $Y$ was realised, a subset $O=\{O_1,...,O_n\},\ O\subseteq [m]$ of actors entered the queue. Since they were constrained by the suborder $H[O]$, the noise free observation model is $Y\sim \mbox{Unif}(\L[H[O]])$: the list is a random draw from the linear extensions of the suborder. For example, for $H$ in \fig~\ref{fig:po-example}, if $O=\{2,4,5\}$, the linear extensions are $\L[H[O]]=\{(2,4,5),(4,2,5)\}$ and $Y$ is chosen at random from this set. The queue-jumping likelihoods $p_{(D)},p_{(U)}$ are obtained by replacing $H\to H[O],\ [m]\to O$ and $m\to n$ in \eqs~\ref{eq:lkd-noisy-down}~and~\ref{eq:lkd-noisy-up}.

\subsubsection{Time series of lists}\label{sec:time-series-likelihood}
Finally, we restore time and give the full likelihood. This observation model is illustrated in \fig~\ref{fig:hmm-obs-only}. 
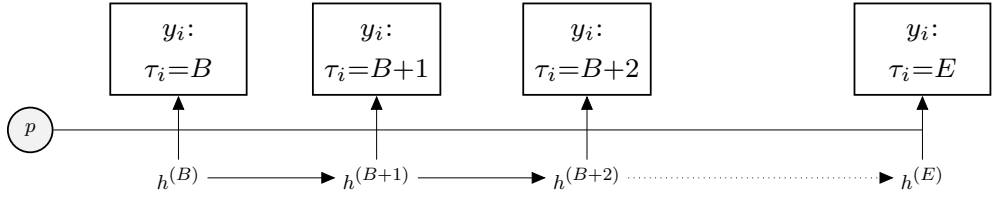
\begin{figure}
    \centering
    \[\scalebox{0.85}{\begin{tikzpicture}

      \node (h1) {$h^{(B)}$};
      \node (h2) [right=of h1, xshift=1cm] {$h^{(B+1)}$};
      \node (h3) [right=of h2, xshift=1cm] {$h^{(B+2)}$};
      \node (h4) [right=of h3, xshift=0.3cm] {};
      \node (h5) [right=of h4, xshift=0.1cm] {};
      \node (hT) [right=of h5, xshift=0.3cm] {$h^{(E)}$};

      \node (y1) [data, above=of h1, yshift=0.5cm] {\scalebox{1.5}{$\strut y_i:\atop \strut\ \ \  \tau_i=B\ \ \ $}};
      \node (y2) [data, above=of h2, yshift=0.5cm] {\scalebox{1.5}{$\strut y_i:\atop \strut\tau_i=B+1$}};
      \node (y3) [data, above=of h3, yshift=0.5cm] {\scalebox{1.5}{$\strut y_i:\atop \strut\tau_i=B+2$}};
      \node (yT) [data, above=of hT, yshift=0.5cm] {\scalebox{1.5}{$\strut y_i:\atop \strut\ \ \ \tau_i=E\ \ \ $}};
      
      \node (p) [param, left=of h1, xshift=-1cm, yshift=0.75cm] {$p$};
      
      \node[shape=coordinate] (hTp) [right=of p, xshift=12.59cm] {};
      \edge [-] {p} {hTp};

      \edge [->] {h1} {h2};
      \edge [->] {h2} {h3};
      \edge [dotted,->] {h3} {hT};
      
      \edge [->] {h1} {y1};
      \edge [->] {h2} {y2};
      \edge [->] {h3} {y3};
      \edge [->] {hT} {yT};
    \end{tikzpicture} }\]
    \caption{The observation model. The distribution of lists $\{y_i: i\in\I,\ \tau_i=t\}$ observed at time $t$ is parameterised by a poset $h^{(t)}$ and noise parameter $p$. A stochastic process realising $h^{(t)},\ t\in [B,E]$ is given in \sec~\ref{sec:prior-main}.} 
    \label{fig:hmm-obs-only}
\end{figure} 
At each time $t\in [B,E]$ the actors indexed in $\M_{t}$ were active and had precedence relations given by some poset $h^{(t)}\in \H_{\M_{t}}$. 
The sequence of partial orders from $B$ to $E$ is
\begin{align}
h&=(h^{(B)},h^{(B+1)},...,h^{(E)}),\label{eq:h-time-series-defn}\\
\intertext{where $h\in \H^{(B,E)}$ with}
\H^{(B,E)}&=\H_{\M_B}\times \H_{\M_{B+1}}\times ... \times \H_{\M_E}.\label{eq:h-sequence-space-defn}
\end{align}
The model for $h$ in \sec~\ref{sec:prior-main} will define a poset process with poset time-series realisations.  

For $i\in\I$, list $y_i$ was formed under constraints imposed by the suborder $h^{(\tau_i)}[o_i]$, so in the noise-free model $y_i\in \L[h^{(\tau_i)}[o_i]]$. Allowing for noise, the likelihood is
\begin{equation}
    p(y|h,\tau,p)=\prod_{i=1}^N p(y_i|h^{(\tau_i)}[o_i],p),\label{eq:lkd-timeseries}
\end{equation}
where $p(y_i|h^{(\tau_i)}[o_i],p)$ is given by $p_{(D)}$ or $p_{(U)}$ in \eqs~\ref{eq:lkd-noisy-down}~or~\ref{eq:lkd-noisy-up} (depending on our choice of model) with the replacements $Y\to y_i$, $m\to n_i$ and $H\to h^{(\tau_i)}[o_i]$. 

\subsection{Latent variables and covariates in a prior for partial orders}\label{sec:prior-main}
We derive prior models for posets from $k$-dimensional random orders \citep{winkler1985random}. 
We modify this setup as we would like to have some control over the prior depth distribution. The depth of the social hierarchy of bishops is meaningful to historians so a prior which is non-informative with respect to depth will be useful.
The uniform prior on posets certainly wouldn't be appropriate  as it is strongly informative of depth: it concentrates on posets of depth three as $m\to \infty$ \citep{kleitman1975asymptotic}. The effect is illustrated in \app~\ref{app:prior-simulation}.

\subsubsection{Latent variable parameterisation} \label{sec:latent-variable-param}
In this section we define status and how it is mapped to precendence order relations. We use feature vectors, one for each actor, to determine actor-placing in the social hierarchy. See \fig~\ref{fig:po-example-UZ} for illustration. These ``status-features'' do not correspond to any identifiable physical attributes. Following \citet{winkler1985random}, we associate with each actor $j\in \M_t$, at a time $t\in [b_j,e_j]$, a $1\times K$ latent vector $Z^{(t)}_{j}\in \R^K,\ Z^{(t)}_{j}=( Z^{(t)}_{j,1},..., Z^{(t)}_{j,K})$ of $K\ge 1$ status-features. Let $Z^{(t)}=( Z^{(t)}_{j} )_{j\in \M_t}$ be an $m_t\times K$ status matrix, with one row for each actor active at time $t$ and one column for each status feature $1,...,K$.

The partial-order $h^{(t)}$ at time $t$ is a function of $Z^{(t)}$. At time $t$ actor $j\in \M_t$ is above actor $j'\in \M_t$ if {\it all} status variables of $j$ are greater than those of $j'$, that is, $h^{(t)}=h(Z^{(t)})$ with
 \begin{equation}\label{eq:z-to-h-mapping}
 h(Z^{(t)})=\{\e{i}{j}\in \M_t\times\M_t: Z^{(t)}_{i,k}>Z^{(t)}_{j,k}\ \mbox{for all $k=1,...,K$}\}.
  \end{equation}
The setup is illustrated in \fig~\ref{fig:po-example-UZ} at centre, where the $Z$-matrix has $m_t=5$ rows and $K=4$ columns. The rows of $Z^{(t)}$ are ``paths'' in the space $[K]\times \R$: the relation $\e{j}{j'}\in  h^{(t)}$ holds when the path $(k,Z^{(t)}_{j,k})_{k=1}^K$ lies above the path through $(k,Z^{(t)}_{j',k})_{k=1}^K$; if the paths cross then the actors are unordered. In \fig~\ref{fig:po-example-UZ}, the $Z$-path for actor 4 intersects the paths for 2 and 3 so 4 is unordered with respect to 2 and 3. The other paths do not intersect so 1, 2, 3, 5 have a total order. This is a latent feature representation of the poset at right. In \cite{winkler1985random} the columns of $Z^{(t)}$ are independent and paths are likely to cross. By contrast, the priors we give in \sec~\ref{sec:prior-prob-dbns} correlate columns and give us some control over the prior depth distribution.
\begin{figure}
    \centering
    {\includegraphics[width=5in, trim=0in 0in 1in 0in, clip]{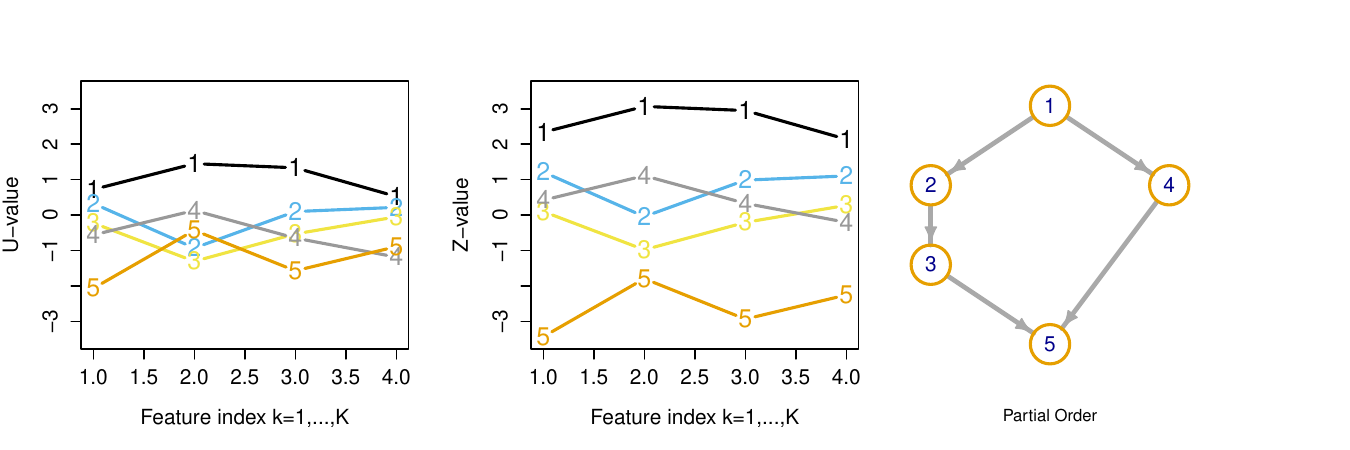}}
    \caption{Latent variable representation of the poset $H$ in \fig~\ref{fig:po-example}: (left) $U$-matrix defined in \sec~\ref{sec:covariate-effects}; (centre) $Z$-matrix and (right) poset $h(Z)$ defined in \sec~\ref{sec:covariate-effects}. The paths in $Z$ are shifted by the effects $\beta$.}
    \label{fig:po-example-UZ}
\end{figure}



Our mapping from $Z^{(t)}$ to $h^{(t)}$ is equivalent to intersecting the $K$ total orders given by ranking the entries in each column of $Z^{(t)}$. See \app~\ref{app:basis-Z-for-h} for a discussion of this point.
As we take a fixed value of $K$, the dimension of $h^{(t)}$ is at most $K$. However, using results from \citep{Hiraguchi51,Bogart73a}, it may be shown that if $K\ge \lfloor m_t/2\rfloor$ then any poset $h^{(t)}\in \M_t$ can be represented by some $m_t\times K$ matrix $Z^{(t)}$. This is discussed in a modelling context in \citet{muirwatt15} and proven in \sec~\ref{sec:prior-properties-Krho}. Since $h^{(t)},\ t\in [B,E]$ has at most $\Sm$ vertices (see \eq~\ref{eq:S}), a model with $K\ge \lfloor \Sm/2\rfloor$ can represent any partial order process $h\in \H^{(B,E)}$. 

Taking $K$ as large as $\lfloor \Sm/2\rfloor$ is conservative as the true $h$-process may live in a lower dimensional space. In recent work \cite{jiang24} estimates $K$, using a non-parametric prior for the $U$-process and sampling $K$ using reversible jump MCMC. They considered smaller data sets with fewer lists all gathered at a single fixed time. In that setting the best fitting values of $K$ were often much smaller than $\lfloor \Sm/2\rfloor$. In earlier work \cite{durante17rejoinder} used a similar approach to estimate the dimension of a latent space parameterisation of a population of networks. In both cases graph structures are modeled using latent continuous variables. Our approach is more like \cite{durante17} and \cite{gwee23}: the dimension of the latent space is chosen so that some kind of representation property holds (like our Proposition~\ref{prop:basis-Z-for-h}); shrinkage \citep{rousseau11} effectively removes unwanted components. In our setting $\rho$ controls this shrinkage. When it is larger the columns of $Z$ tend to have the same order, so the dimension of the poset is smaller.

\subsubsection{Covariate effects for partial orders} \label{sec:covariate-effects} .

In \sec~\ref{sec:notation-data} we introduced an actor-specific covariate informing status relations among bishops. In the following $s_{t,j}\in \{1,2,\ldots,\Sb\}$ is a single categorical or ordinal variable with levels from 1 to $\Sb$. More general covariates are easily accommodated. 
Let $\beta\in \R^{\Sb}$ be the vector of level effects. We split the \emph{status}-vector $Z^{(t)}_j$ of actor $j\in \M_t$ into an additive effect $\beta_{s_{t,j}}$ due to $s_{t,j}$ and an \emph{authority}-vector $U^{(t)}_j\in \R^K$, which captures all aspects of status which are not attributable to the covariate.
Our additive model is, for $j\in \M_t$ and $t\in [B,E]$,
\begin{equation}\label{eq:additive-model-for-Z}
Z^{(t)}_j= U^{(t)}_j+1_{K}\beta_{s_{t,j}}
\end{equation}
where $1_K$ is a row vector of $K$ ones. Higher $\beta_{s_{t,j}}$-values lift all components of $U^{(t)}_j$ by a constant, raising the status features in $Z^{(t)}_j$. This moves the path $(k, Z^{(t)}_{j,k})_{k=1}^K$ above other paths and gives a higher position for actor $j$ in the poset $h^{(t)}$. 
See \fig~\ref{fig:po-example-UZ} for an example. 

In our application the covariate is ordinal with a greater effect expected for lower values of $s$ so we will be interested in testing for $\beta_1>\beta_2>...>\beta_{\Sb}$. Let $\B_0=\R^{\Sb}$ and
\begin{equation}\label{eq:beta-constrained-space}
    \B_{\Sb}=\{\beta\in \B_0: \beta_1>\beta_2>...>\beta_{\Sb}\}.
\end{equation}
We carry out analyses under models with $\beta\in \B_0$ (to check our prior expectation that $\beta\in \B_{\Sb}$) and then again with $\beta\in\B_{\Sb}$ (for best estimation with a well supported subjective prior). 

Let $Z=( Z^{(t)})_{t=B}^E$ and $U=( U^{(t)})_{t=B}^E$ and write $Z=Z(U,\beta;s)$ for the function defined in \eq~\ref{eq:additive-model-for-Z}. The parameters $U$ and $\beta$ replace $h$ in the likelihood via $h=h(Z(U,\beta;s))$. 

\subsubsection{Prior probability distributions}\label{sec:prior-prob-dbns}

We model the $K$-dimensional {authority}-process $U_j=(U^{(t)}_j)_{t=b_j}^{e_j}$ for actor $j\in \M$ as a vector autoregression of order one with times-series correlation $\theta$ and covariance $\Sigma^{(\rho)}$. In our model latent {authority}-features are correlated from one time step to another with a drift towards zero. The setup is illustrated in \fig~\ref{fig:hmm}.

Our prior for the process $U_j$ is independent over $j\in\M$ with correlation parameters $0\le \theta\le 1$ and $0\le \rho\le 1$. Let $\Sigma^{(\rho)}$ be a $K\times K$ covariance matrix with diagonal elements $\Sigma^{(\rho)}_{k,k}=1$ and off diagonal $\Sigma^{(\rho)}_{k,k'}=\rho$ for $k,k'\in [K]$. 
Let $0_K$ be a vector of $K$ zeros. For $j\in \M$ let
\begin{align}
 U^{(b_j)}_j &\sim N\left(0_K,\frac{\Sigma^{(\rho)}}{(1-\theta)^2}\right)\nonumber\\  
\intertext{and for $\epsilon^{(t)}_j$ iid for $t\in [b_j+1,e_j]$ and each $j$,}
 U^{(t)}_j &=\theta U^{(t-1)}_j + \epsilon^{(t)}_j\,, \hspace{.3cm} \epsilon^{(t)}_j \sim N(0,\Sigma^{(\rho)}) \, .  \label{eq:U-VAR-model}
\end{align}
Write $U_j\sim \mbox{VAR}^{(b_j,e_j)}_{K,\rho,\theta}(1),\ j\in\M$ for the vector auto-regression with density
\begin{equation}\label{eq:U-VAR-prior-densty}
    \pi(U_j|\rho,\theta)= N\left(U^{(b_j)}_j;0_K,\frac{\Sigma^{(\rho)}}{(1-\theta^2)}\right) \prod_{t=b_j+1}^{e_j} N\left(U^{(t)}_j;\theta U^{(t-1)}_j,\Sigma^{(\rho)}\right)\, .
\end{equation}

The parameter $\rho$ controls the prior depth-distribution. When $\rho\simeq 1$ paths $(k,U^{(t)}_{j,k})_{k=1}^K$ are relatively flat as entries are strongly correlated. Flat paths don't intersect, so there are more order relations and $d(h^{(t)})$ is larger. When $\rho$ is close to zero the paths are more jagged so there are few order relations and $d(h^{(t)})$ is smaller. 
We take as our prior $\rho\sim \mbox{Beta}(\gamma_1,\gamma_2,\gamma_3)$ with non-centrality parameter $\gamma_3$ and $\gamma=(\gamma_1,\gamma_2,\gamma_3)$ fixed. Prior simulation (\app~\ref{app:prior-simulation}, \fig~\ref{fig:prior-depth-dbns-curves-NF9-18}, left) showed $K=\lfloor \Sm/2\rfloor$ and $\gamma=(1,1/3,8)$ gave a prior on posets which is acceptably uniform on depth. Our prior on $\theta$ is uniform, $\theta\sim \mbox{Unif}(0,1)$. 
 
Our prior density for $\beta$ is $\pi_\beta(\beta)=N(\beta; 0,I_{\Sb})$. The variation between levels of a covariate equals the variation in {authority} over one time step: $\beta_j-\beta_{j'}$ has the same variance as the components of $U^{t}_j-U^{(t-1)}_{j}$ in \eq~\ref{eq:U-VAR-model}. We set $\Sigma^{(\rho)}_{k,k}=1$ as we can scale $\rho$ and the variance of $\beta$ to get the same distribution for $h$. It is necessary to take proper priors for $U$ and $\beta$. We discuss these priors further in \sec~\ref{sec:post-sum} in relation to identifiability.

%
%

\subsection{Prior summary and Posterior distribution}\label{sec:time-series-posterior}
We now summarise our generative model for the data. The model is depicted in \figs~\ref{fig:hmm-obs-only} and \ref{fig:hmm}.
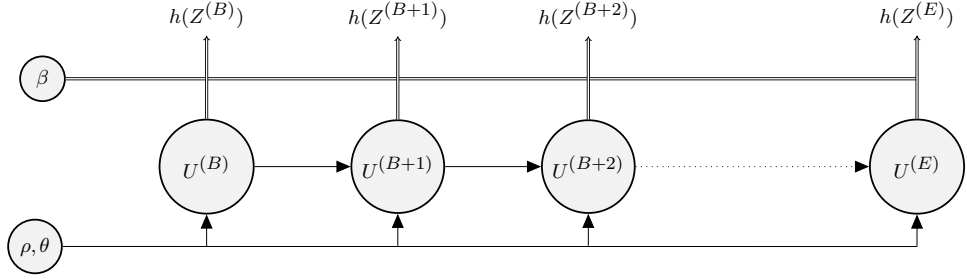
\begin{figure}
    \centering
    \[\scalebox{0.85}{\begin{tikzpicture}
      \node (t1) [param] {$\, \ \ U^{(B)}\ \ \,$};
      \node (t2) [param, right=of t1, xshift=1cm] {$U^{(B+1)}$};
      \node (t3) [param, right=of t2, xshift=1cm] {$U^{(B+2)}$};
      \node (t4) [right=of t3, xshift=0.3cm] {};
      \node (t5) [right=of t4, xshift=0.1cm] {};
      \node (T) [param, right=of t5, xshift=0.3cm] {$\, \ \ U^{(E)}\ \ \, $};

      \node (h1) [above=of t1, yshift=0.3cm] {$h(Z^{(B)})$};
      \node (h2) [above=of t2, yshift=0.3cm] {$h(Z^{(B+1)})$};
      \node (h3) [above=of t3, yshift=0.3cm] {$h(Z^{(B+2)})$};
      \node (hT) [above=of T, yshift=0.3cm] {$h(Z^{(E)})$};

      \node (b) [param, left=of h1, xshift=-1cm, yshift=-1cm] {$\beta$};
      \node[shape=coordinate] (hTb) [right=of b, xshift=12.3cm] {};
      \edge [double,-] {b} {hTb};
      
      \node (rt) [param, left=of t1, xshift=-1cm, yshift=-1.25cm] {$\rho,\theta$};
      \node[shape=coordinate] (t1rt) [below=of t1, yshift=0.5cm] {};
      \node[shape=coordinate] (t2rt) [below=of t2, yshift=0.5cm] {};
      \node[shape=coordinate] (t3rt) [below=of t3, yshift=0.5cm] {};
      \node[shape=coordinate] (tTrt) [below=of T, yshift=0.5cm] {};
      \edge [-] {rt} {tTrt};
      \edge [->] {t1rt} {t1};
      \edge [->] {t2rt} {t2};
      \edge [->] {t3rt} {t3};
      \edge [->] {tTrt} {T};
      
      \edge [->] {t1} {t2};
      \edge [->] {t2} {t3};
      \edge [dotted,->] {t3} {T};
      \edge [double,-{implies}] {t1} {h1};
      \edge [double,-{implies}] {t2} {h2};
      \edge [double,-{implies}] {t3} {h3};
      \edge [double,-{implies}] {T} {hT};
    \end{tikzpicture} }\]
    \caption{Hidden layer of HMM: a latent process $U^{(t)}$ determines a sequence of posets $h^{(t)}=h(Z(U^{(t)},\beta,s)$ for $t\in [B,E]$. The $U$-prior has hyperparameters $\rho$ and $\theta$, and double lines indicate implication $\Rightarrow$. This figure connects with \fig~\ref{fig:hmm-obs-only} at the top row of nodes $h^{(t)}=h(Z^{(t)})$ to give the generative model for the data.} 
    \label{fig:hmm}
\end{figure} 
The data are the lists $y$. We condition on knowledge of the uncertainty ranges $\tau^{\pm}$, the covariate data $s$, the actor activities $\M_t,\ t\in [B,E]$ and the prior hyper-parameters $\gamma,\delta>0$ and $K\ge 1$. The generative model is,
\begin{align}
    \rho&\sim \mbox{Beta}(\gamma), 
    && \quad\mbox{with $\gamma=(1,1/3,8)$ unless stated,}\nonumber\\
    \theta&\sim \mbox{Unif}(0,1), 
    && \nonumber\\
    U_j&\sim \mbox{VAR}^{(b_j,e_j)}_{K,\rho,\theta}(1),
    &&\quad \mbox{with $K=\lfloor \Sm/2\rfloor$ unless stated, and $U$ in \eq~\ref{eq:U-VAR-model},}\nonumber\\
    &&&\quad \mbox{$U^{(t)}_j$ iid for}\ j\in\M\ \mbox{and defined for $t: j\in\M_t$ and} \nonumber\\
    \beta&\sim N(0,I_{\Sb}),
    &&\quad \mbox{either $\beta\in\B_0$ or constrained $\beta\in\B_{\Sb}$ per \eq~\ref{eq:beta-constrained-space}}. \nonumber\\
\intertext{These collectively determine the partial-order prior via}
    Z&=Z(U,\beta;s),
    &&\quad \mbox{from \eq~\ref{eq:additive-model-for-Z}, giving $Z=(Z^{(t)})_{t=B}^E$, and}\nonumber\\
    h&=h(Z(U,\beta;s)),
    &&\quad \mbox{from \eq~\ref{eq:z-to-h-mapping}, giving $h=(h^{(t)})_{t=B}^E$}; \label{eq:prior-for-h}
    \intertext{Priors for the remaining observation model parameters are}
    \tau_i&\sim U\{\tau^-_i,\tau^+_i\},
    &&\quad \mbox{independently for}\ i=1,...,N,\ \mbox{and} \nonumber\\
    p&\sim \mbox{Beta}(1,\delta); 
    &&\quad \mbox{with $\delta\ge 1$ and $\delta=9$ by default.}\nonumber\\
    \intertext{Finally the data are realised}
    y_i&\sim p(\cdot|h^{(\tau_i)}[o_i],p),
    &&\quad \mbox{independently for $i\in \I$},\label{eq:observation-model}
\end{align}
 using the distribution for $y_i$ given in \eq~\ref{eq:lkd-noisy-down} or \ref{eq:lkd-noisy-up}.
The joint posterior distribution is 
\begin{align}\label{eq:posterior}
\pi(\rho,\theta,U,\beta,\tau,p|y) \propto \pi(\rho,\theta,\beta,\tau,p)  \pi(U|\rho,\theta)p(y|U,\beta,\tau,p) \, , 
\end{align} 
where $p(y|U,\beta,\tau,p)=p(y|h(Z(U,\beta;s)),\tau,p)$ in \eq~\ref{eq:lkd-timeseries} and
    $\pi(U|\rho,\theta)=\prod_{j\in\M} \pi(U_j|\rho,\theta)$ 
with $\pi(U_j|\rho,\theta)$ given in \eq~\ref{eq:U-VAR-prior-densty}. 
The model without covariates or time is set out in \app~\ref{app:fixed-time-posterior}. 

\section{Properties of priors}\label{sec:priors-properties}

Our partial order prior is marginally consistent and expresses any partial-order time series in $\H^{(B,E)}$.
\app~\ref{app:prior-simulation} explores the prior using simulation.

\subsection{Marginal consistency}\label{sec:marginal-consistent-prior} 

Marginal consistency is a relationship between members of a family of distributions. Dropping time, suppose that for each non-empty $O\subseteq [m]$ we write down a prior $\pi_{\H_O}(G)$ on $G\in \H_O$. These priors are marginally consistent if for each $O$ and $H\sim \pi_{\H_{[m]}}$ we have $H[O]\sim \pi_{\H_O}$ for the distribution of the suborder, that is,
\begin{equation}\label{eq:PO-prior-marg-con}
\pi_{\H_O}(G)=\sum_{H\in \H_{[m]}} \pi_{\H_{[m]}}(H)\mathbb{I}_{G=H[O]}.
\end{equation}
The point here is that we define $\pi_{\H_O}$ for each $O\subseteq [m]$ and then verify \eq~\ref{eq:PO-prior-marg-con}. It can fail to hold if we write down a distribution $\pi_{\H_O}$ for each $O$ without care. For example, the uniform distributions $\mbox{Unif}(\H_{O}),\ O\subseteq [m]$ are not marginally consistent. There are three partial orders for $m=2$ and nineteen for $m=3$, so if $H\sim \mbox{Unif}(\H_{[3]})$ then we won't have $H[(1,2)]\sim \mbox{Unif}(\H_{[2]})$ as we can't group nineteen posets into three equal-sized groups.

\cite{winkler1985random} shows marginal consistency when the columns of $Z^{(t)}$ are independent. We extend this to our more general setting.
We take $\beta=0_{\Sb}$, so no covariates, as we cannot expect marginal consistency when we have covariate information which explicitly breaks it. Let $\pi_{\H^{(B,E)}}(h|\beta=0_{\Sb})$ be the marginal prior for $h\in\H^{(B,E)}$ when $\beta=0_{\Sb}$ (see \eq~\ref{eq:h-marginal-discrete-prior-beta0} in \app~\ref{app:marginal-consistency}). Let $\H^{(B,E)}_{-j}$ be the set of all partial order time-series with $j\in\M$ removed and for $h\in \H^{(B,E)}$ let $h_{-j}$ be the suborder obtained by removing $j$.

\begin{proposition}\label{prop:marginal-consistent}
Let $g\in \H^{(B,E)}_{-j}$ be given. Our priors are marginally consistent, that is
\[
\pi_{\H^{(B,E)}_{-j}}(g|\beta=0_{\Sb})=\sum_{h\in \H^{(B,E)}} \mathbb{I}_{g=h_{-j}} \pi_{\H^{(B,E)}}(h|\beta=0_{\Sb}).
\]
for each $j\in \M$. [See \app~\ref{app:marginal-consistency} for proof]
\end{proposition}

Proposition~\ref{prop:marginal-consistent} establishes marginal consistency for removing one element of $\M$. Consistency for more general marginals follows by removing elements one at a time. See \app~\ref{app:marginal-consistency} for remarks on settings where marginal consistency holds with covariate effects retained.

\subsection{Universal representation}\label{sec:prior-properties-Krho}

We claimed in \sec~\ref{sec:latent-variable-param} that, if $K\ge \lfloor \Sm/2\rfloor$ with $\Sm$ defined in \eq~\ref{eq:S} then any poset $h^{(t)}\in \H_{\M_t}$ can be represented by some $m_t\times K$ matrix $Z^{(t)}$.
We restore $\beta\in \B_0$ or $\B_{\Sb}$ and covariate effects.
Let
\begin{equation}\label{eq:h-marginal-discrete-prior}
    \pi_{\H^{(B,E)}}(h)=\int_{\R^{\Sb}} \pi_{\H^{(B,E)}}(h|\beta)\pi(\beta)\, d\beta
\end{equation}
be the full marginal prior with variable $\beta$, extended from \eq~\ref{eq:h-marginal-discrete-prior-beta0}.
\begin{proposition}\label{prop:basis-Z-for-h}
Suppose $\min_{t\in [B,E]} m_t \ge 4$. The probability $\pi_{\H^{(B,E)}}(h)$ in \eq~\ref{eq:h-marginal-discrete-prior}, given by the generative model \eq~\ref{eq:prior-for-h} with $K\ge \lfloor \Sm/2\rfloor$, assigns a positive probability mass $\pi_{\H^{(B,E)}}(h)>0$ to every time-series $h\in\H^{(B,E)}$. [See \app~\ref{app:basis-Z-for-h} for proof.]
\end{proposition}


\section{Computational methods}\label{sec:comp-methods}


\subsection{Markov Chain Monte Carlo}
We implemented an MCMC algorithm targeting $\pi(\rho,\theta,U,\beta,\tau,p|y)$ in \eq~\ref{eq:posterior}. Each update is a simple Metropolis-Hastings MCMC step. The updates are summarised in \app~\ref{app:mcmc}. We tested the software evaluating the likelihood by simulating synthetic data and checking list proportions matched their probability in the likelihood. We also recover the true parameters of synthetic data (see \fig~\ref{app:synth-poHB22aRS4}).
We run the MCMC producing $L$ samples (after burn-in and thinning) $\rho^{(l)},\theta^{(l)}$, $U^{(t,l)}=(U^{(t,l)}_j)_{j\in \M_t}$,
$\beta^{(l)},\tau^{(l)}$ and $p^{(l)}$ for $l=1,...,L$.
This determines samples for
$Z^{(t,l)}=(Z^{(t,l)}_j)_{j\in \M_t}$,
with 
\[
Z^{(t,l)}_j= U^{(t,l)}_j+1_{K}\beta^{(l)}_{s_{t,j}}
\]
and similarly $h^{(t,l)}=h(Z^{(t,l)})$, for $t=B,...,E$ and $l=1,...,L$.

The likelihood is not differentiable in $U$ and $\beta$ ruling out Hamiltonian MCMC. We tried updating sections of the time series in parallel (exploiting the Markov structure), halving the runtime at best. The bottleneck is ultimately the likelihood evaluation which is $\#P$-complete in list length. When lists are short the method scales well in the number of actors and the number of lists. 
The latent-parameter dimension $\dim(U)=K\times \sum_t m_t$ becomes the limiting factor (in \sec~\ref{sec:poHB2aRS6}, $\dim(U)=13,453$).
These challenges motivated work on VSP-models in \citep{jiang23} which scale to lists with hundreds of actors.

\subsection{Posterior summaries}\label{sec:post-sum}

Besides plotting marginals for individual parameters $\rho,\theta,p$ and $\beta$, we report selected summary statistics computed on the MCMC output. These are the consensus poset (which displays relations with posterior support greater than one half, see \app~\ref{app:posterior-summary-stats}) and the Bayes factor for the first $S'$ of the $\Sb$ covariate effects to be ordered, $\B_{S'}=\{\beta\in \B_0: \beta_1>\beta_2>...>\beta_{S'}\}$ and given by
$
B_{S',0}={p(y|\beta\in \B_{S'})}/{p(y|\beta\in \B_0)}.
$
Formulae for estimating these quantities are given in \app~\ref{app:posterior-summary-stats}.

\subsubsection{Non-identifiability of authority and seniority effects}
\label{sec:non-identifiable-UbetaPO} 
We are interested in separating the relative {authority} $U^{(t)}_j$ of a actor from the status $Z^{(t)}_j$. There are two sources of non-identifiability. The latent variables $U$ have a label swapping symmetry: the posterior is invariant under permutation of the columns of $U^{(t)}$ if the same permutation is applied at each $t\in [B,E]$. One simple summary which is invariant under column permutation is the row-average. This gives the estimated average posterior authority for actor $j$ at time $t$,
\begin{equation}\label{eq:post-mean-U-process}
\bar U^{(t)}_j=\frac{1}{L}\sum_{l=1}^L \frac{1}{K}\sum_{k=1}^K U^{(t,l)}_{j,k}.
\end{equation}
A plot of $\bar U^{(t)}_j$ against $t$ shows $j$'s changing authority. Average status $\bar Z^{(t)}_j$ is defined similarly.

The second source of non-identifiability is shift invariance of $h^{(t)}$ under
\begin{align}\label{eq:non-identifiable-shift}
    U^{(t)}_j&\to U^{(t)}_j+1_Ku^{(t)},\ t\in [B,E],\ j\in \M_t,\\
    \beta_r&\to \beta_r+c,\ r\in [\Sb],
\end{align}
where $u^{(t)}\in \R$ is a common shift applied to each actor, which may vary over $t$, and $c\in \R$ is common shift applied to all effects. The proper $U$ and $\beta$ priors shrink these shifts towards zero. We project these degrees of freedom out by subtracting the averages, $\bar U^{(t)}_j\to \bar U^{(t)}_j-\M_t^{-1}\sum_j \bar U^{(t)}_j$ and $\beta_r\to \beta_r-\Sb^{-1}\sum_{r'}\beta_{r'}$, before computing the summary statistics and plotting. A similar issue arises in the Plackett-Luce time-series model in \app~\ref{app:PL-time-series-model}.

\section{Results}\label{sec:main-results}

\subsection{Fitted models}\label{sec:fitted-models-in-runs}

Prior distributions are summarised in \sec~\ref{sec:time-series-posterior}. 
Unless otherwise indicated we present results for the likelihood $p_{(D)}$ in \eq~\ref{eq:lkd-noisy-down} as clerks wrote lists from top down. Results are near-identical with $p_{(U)}$ in \eq~\ref{eq:lkd-noisy-up}. Experiments showed fractionally lower estimated noise probabilities $p$ for $p_{(U)}$ than $p_{(D)}$ (see \fig~\ref{fig:poHB1-14-16-rho-theta-p}). This suggests the $p_{(U)}$ model is a slightly better fit, but there is little difference. 
The greatest number of active bishops is $\Sm=22$ so we take $K=11$ (see end of \sec~\ref{sec:latent-variable-param}) for the dimension of the latent status feature vectors $Z^{(t)}_{j},\ j\in \M, t\in [B,E]$ in our main analyses in \secs~\ref{sec:poHB1aRS6} and \ref{sec:poHB2aRS6}. We check robustness by taking $K=2$ and $K=22$ in \sec~\ref{app:K-equals-18-posterior} of \app~\ref{app:further-results}. The $\beta$-dimension is $\Sb=21$ (less than $\Sm$ as there is a tie at seniority-rank 21 in 1133, the year with the most active bishops).

\subsection{Analysis with unconstrained seniority effects}\label{sec:poHB1aRS6}

We begin by presenting our results for the full data set defined in \sec~\ref{sec:data-witness-lists}. We first check that we see declining seniority effect at lower seniority, so we do not constrain the seniority effects to be ordered and take $\beta\in \B_0$. 

Traces in \fig~\ref{fig:mcmc-poHB1aRS6} in \app~\ref{app:poHB1aRS6} for MCMC targeting the posterior in \sec~\ref{sec:time-series-posterior} show convergence. Marginal posterior densities from two independent runs are shown in \fig~\ref{fig:rho-th-p-poHB1aRS6} and are near-identical. The correlation $\rho$ of features in $U^{(t)}_j,\ j\in\M$ at each fixed time $t\in [b_j,e_j]$ is close to one, supporting relatively deeper posets. The time-series correlation parameter $\theta$ is close to one, indicating strong serial correlation between $U^{(t)}_j$ and $U^{(t+1)}_j$, and therefore also $h^{(t)}$ and $h^{(t+1)}$. Finally, the error probability $p$, which controls the extent to which lists $y_i$ may depart from the linear extensions $\L[h^{(\tau_i)}]$, is small, as we would expect if the poset model captures the variation in lists. The prior for $p$ has $\delta=9$, favoring small $p$. We checked robustness to this choice: \fig~\ref{fig:rho-th-p-poHB1aRS6} superimposes the posterior density when $\delta=1$ (uniform, estimated by importance sampling); the shift from the $\delta=9$ posterior is barely discernible. 
\begin{figure}
    \centering
    \hspace*{-0.3in}\includegraphics[width=5.5in,trim=0.2in 0.1in 0.15in 0in,clip]{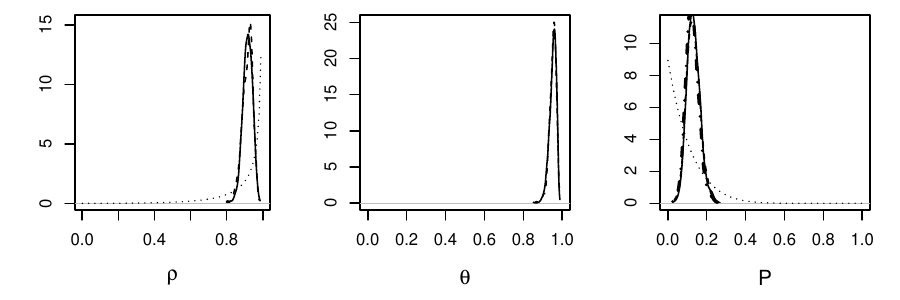}
    \caption{Posterior densities for $\rho,\theta$ and $p$ from the unconstrained seniority effects analysis in \sec~\ref{sec:poHB1aRS6}. Two independent MCMC runs are shown (solid and dashed). The dotted line in the $\rho$ and $p$ graphs is their prior. The prior for $\theta$ is uniform. The thick dash-dot curve in the $p$-density is the posterior density for a uniform prior on $p$.}
    \label{fig:rho-th-p-poHB1aRS6}
\end{figure}

\begin{figure}
    \centering
    \hspace*{-0.2in}\begin{tabular}{cc}
    {\includegraphics[width=2.5in,trim=1.2in 3in 1.5in 3.6in,clip]{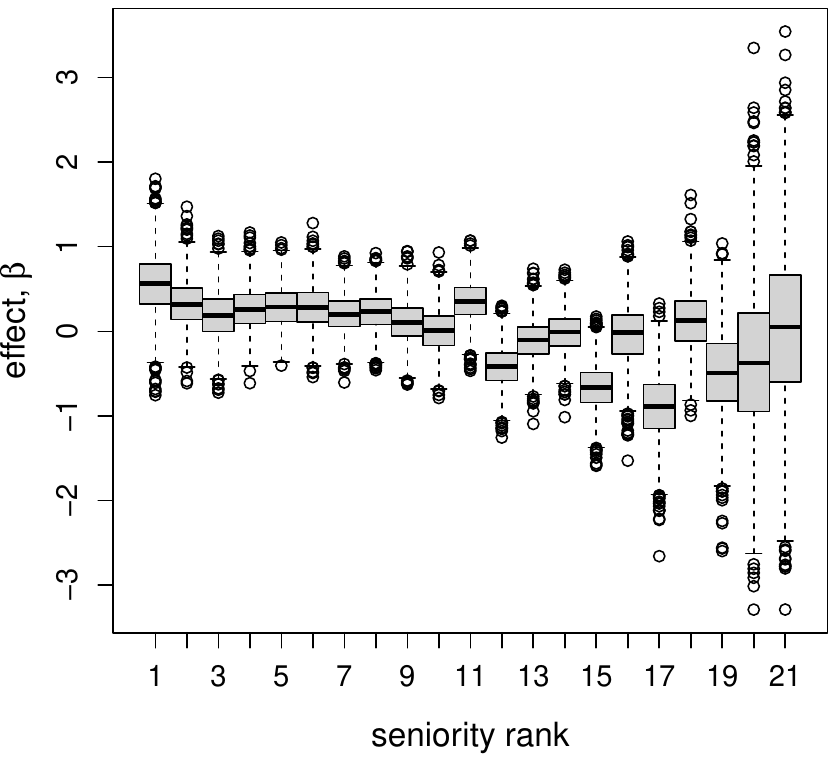}}&\quad
    \includegraphics[width=2.5in,trim=1.2in 3in 1.5in 3.6in,clip]{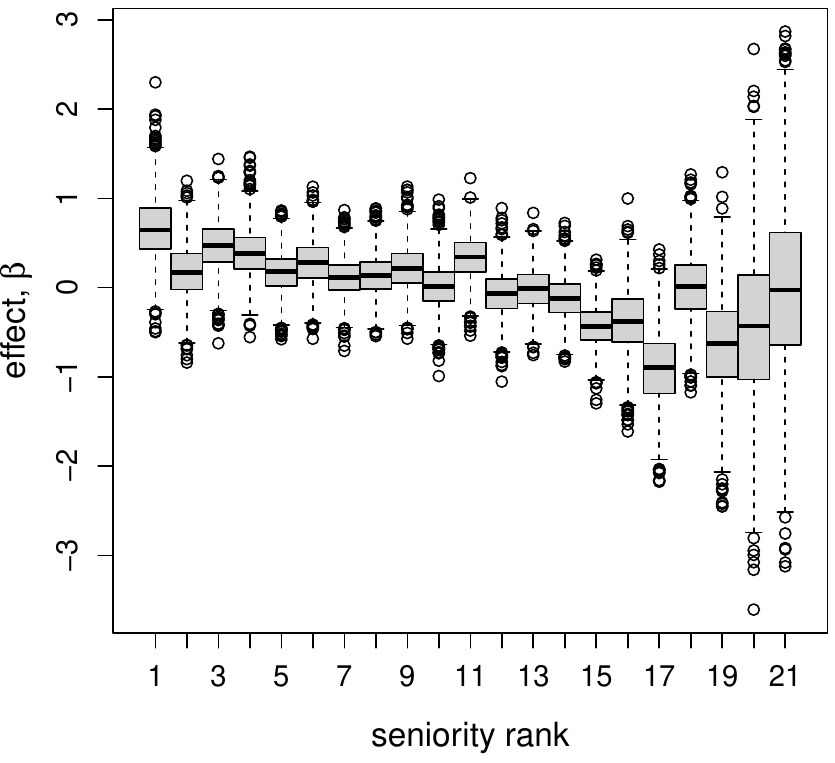}
    \end{tabular}
    \caption{(Left) Marginal posterior distributions of seniority effect parameters from the unconstrained seniority effects analysis in \sec~\ref{sec:poHB1aRS6} with likelihood $p_{(D)}$ in \eq~\ref{eq:lkd-noisy-down}. (Right) As left with likelihood $p_{(U)}$ in \eq~\ref{eq:lkd-noisy-up}. See \fig~\ref{fig:beta-poHB1aRS6-PL-version} in \app~\ref{app:PL-time-series-model} for the corresponding plot for the Plackett-Luce time-series model comparison. }
    \label{fig:beta-poHB1aRS6}
\end{figure}
In \fig~\ref{fig:beta-poHB1aRS6} we plot marginal posterior $\beta$ distributions in the posterior with likelihood $p_{(D)}$ (noise is random upward displacement) and $p_{(U)}$ (random downward displacement) respectively. There is a clear downward trend with increasing seniority-rank value (ie, lower seniority) as we expect. When the rank is large (18-21) we have few instances of bishops with that rank (see \fig~\ref{fig:seniority-rank-stuff}), and distributions trend back to the prior. Under both observation models $p_{(U)}$ and $p_{(D)}$, the effect $\beta_{11}$ is ``out of order''. This is because several bishops who spent several years at seniority-rank 11 (William Giffard, bishop of Winchester, Richard de Belmeis I, bishop of London and Henry de Blois, Bishop of Winchester) were connected with royalty, so it was their authority and not their seniority which pushed them up the lists. This is best modeled by imposing the seniority-effect order constraint $\beta\in\B_{\Sb}$ as we do in the next section.
In \sec~\ref{app:bayes-factor-decreasing} we test $\beta\in\B_\Sb$ against the unconstrained model $\beta\in\B_0$ by estimating the Bayes factor, using a Savage-Dickey estimator. We find clear evidence in favor of the constraint. In summary we see in this first analysis the structures we anticipated.

\subsection{Analysis with constrained seniority effects}\label{sec:poHB2aRS6}

The assumption of decreasing seniority effect with decreasing seniority rank is supported on historical and statistical grounds so we now impose the constraint $\beta\in \B_{\Sb}$. We omit the $\rho,\theta$ and $p$ posterior densities as they are essentially unchanged from \fig~\ref{fig:rho-th-p-poHB1aRS6}. Marginal posterior distributions for the unknown list dates $\tau_i,\ i\in \I$ are given in \fig~\ref{fig:tau-poHB2aRS6} in \app~\ref{app:poHB2aRS6}. %

In \fig~\ref{fig:POmcmc-poHB2aRS6} we plot posets $h^{(t,L)}, t\in [1134,1136]$ from one MCMC sample state (the final sample in the MCMC sample output $(h^{(t,l)})_{t\in [B,E]},\ l\in[L]$).
\begin{figure}
    \centering
    \hspace*{-0.2in}{\includegraphics[width=5.5in,trim=0in 0.5in 0in 0.3in, clip]{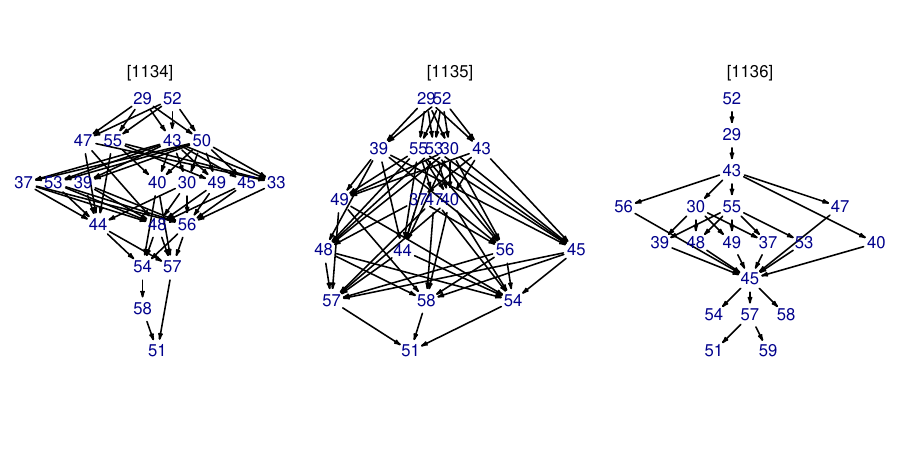}}
    \caption{(constrained seniority effects analysis of \sec~\ref{sec:poHB2aRS6}) Three of seventy-six years from one MCMC-sampled state, $h^{(t,L)},\ t\in \{1134, 1135,1136\}$. Vertex numbers correspond to bishops names in \fig~\ref{fig:bishop-names} in \app~\ref{app:poHB2aRS6}.}
    \label{fig:POmcmc-poHB2aRS6}
\end{figure}
These three years bracket 1135 when Stephen became king. The number and length of lists in this period is relatively large (see \fig~\ref{fig:list-dates-lengths}). 
In \fig~\ref{fig:CPO-poHB2aRS6} 
\begin{figure}
    \centering
    \hspace*{-0.2in}\includegraphics[width=5.25in,trim=0in 0.2in 0in 0.0in, clip]{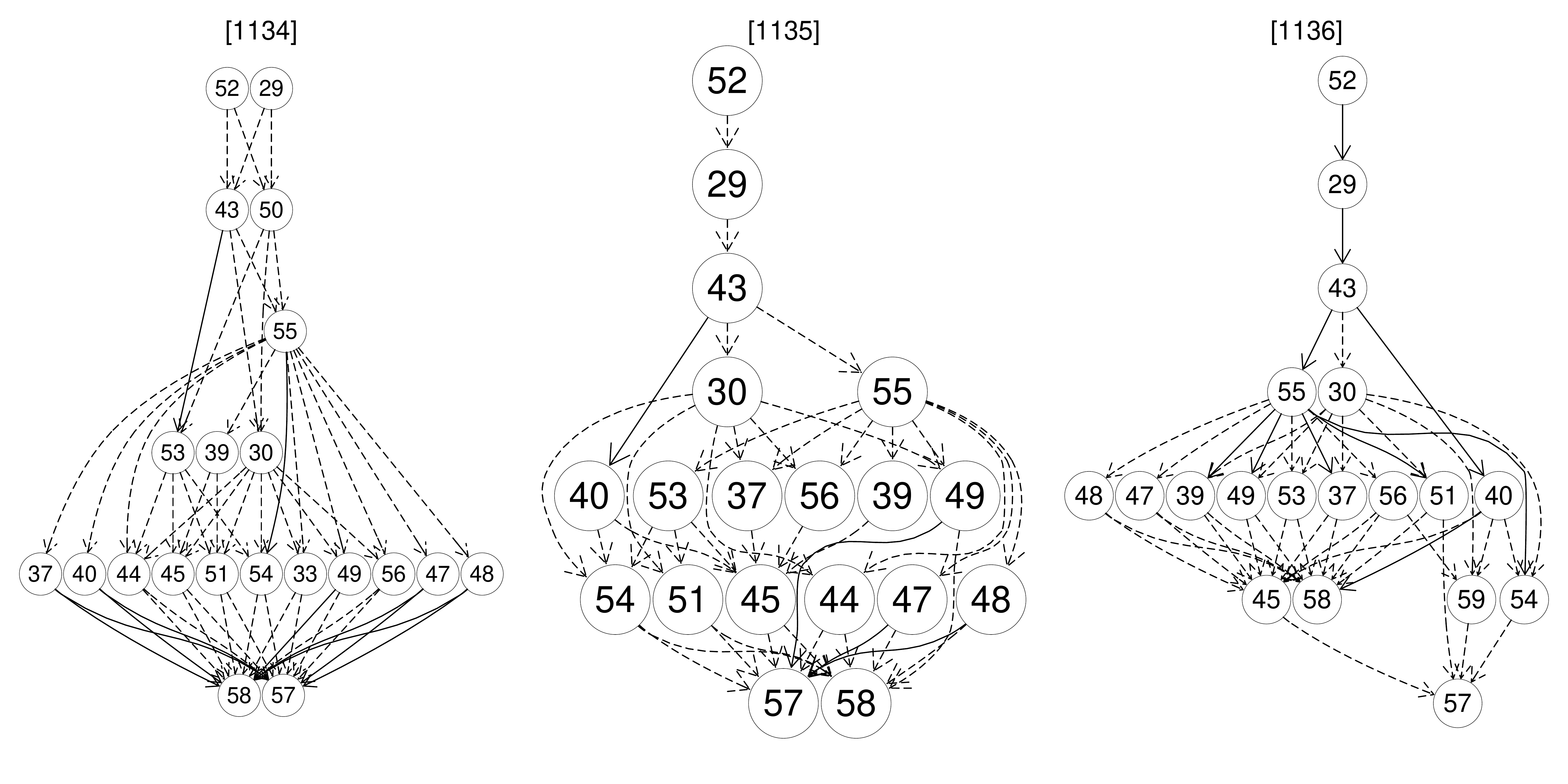}
    \caption{(constrained-effects model \sec~\ref{sec:poHB2aRS6}) Consensus posets for 1134-1136 (transitive reductions). For closures see \app~\ref{app:poHB2aRS6} \fig~\ref{fig:CPO-poHB2aRS6-close}. Dashed/solid edges have posterior support greater than $0.5/0.9$. Numbering as \fig~\ref{fig:POmcmc-poHB2aRS6}.}
    \label{fig:CPO-poHB2aRS6}
\end{figure}
we plot posterior consensus posets $\bar h^{(t)}$ estimated at the same years (transitive reductions for ease of viewing, see \app~\ref{app:poHB2aRS6} for all years). The transtive closures in \fig~\ref{fig:CPO-poHB2aRS6-close} in \app~\ref{app:poHB2aRS6} have many more strongly supported edges as a chain of weakly supported relations give strongly supported relations from chain head to tail.

In \fig~\ref{fig:Zhat-poHB2aRS6} in \app~\ref{app:poHB2aRS6} we plot the evolving mean status values $\bar Z^{(t)}_j$ for each bishop as a function of time. Bishops are grouped by diocese.
These curves have a ``sawtooth'' pattern, as the ``status'' measure $Z$ trends up through the tenure of a bishop as their seniority increases. It drops down when a new bishop enters the diocese with low seniority. 
By contrast the curves in \fig~\ref{fig:Uhat-poHB2aRS6} show the evolving mean {authority} values $\bar U^{(t)}_j$ for each bishop. These curves are flatter as the effect of seniority is removed. Nigel, bishop of Ely is revealing. His status in \fig~\ref{fig:Zhat-poHB2aRS6} is fairly flat. This is because his mean {authority} $\bar U^{(t)}_j$ declined as his seniority increased. 
\begin{figure}
    \centering
    \hspace*{-0.8in}\includegraphics[width=6.75in,trim=0in 0.25in 0.25in 0.25in, clip]{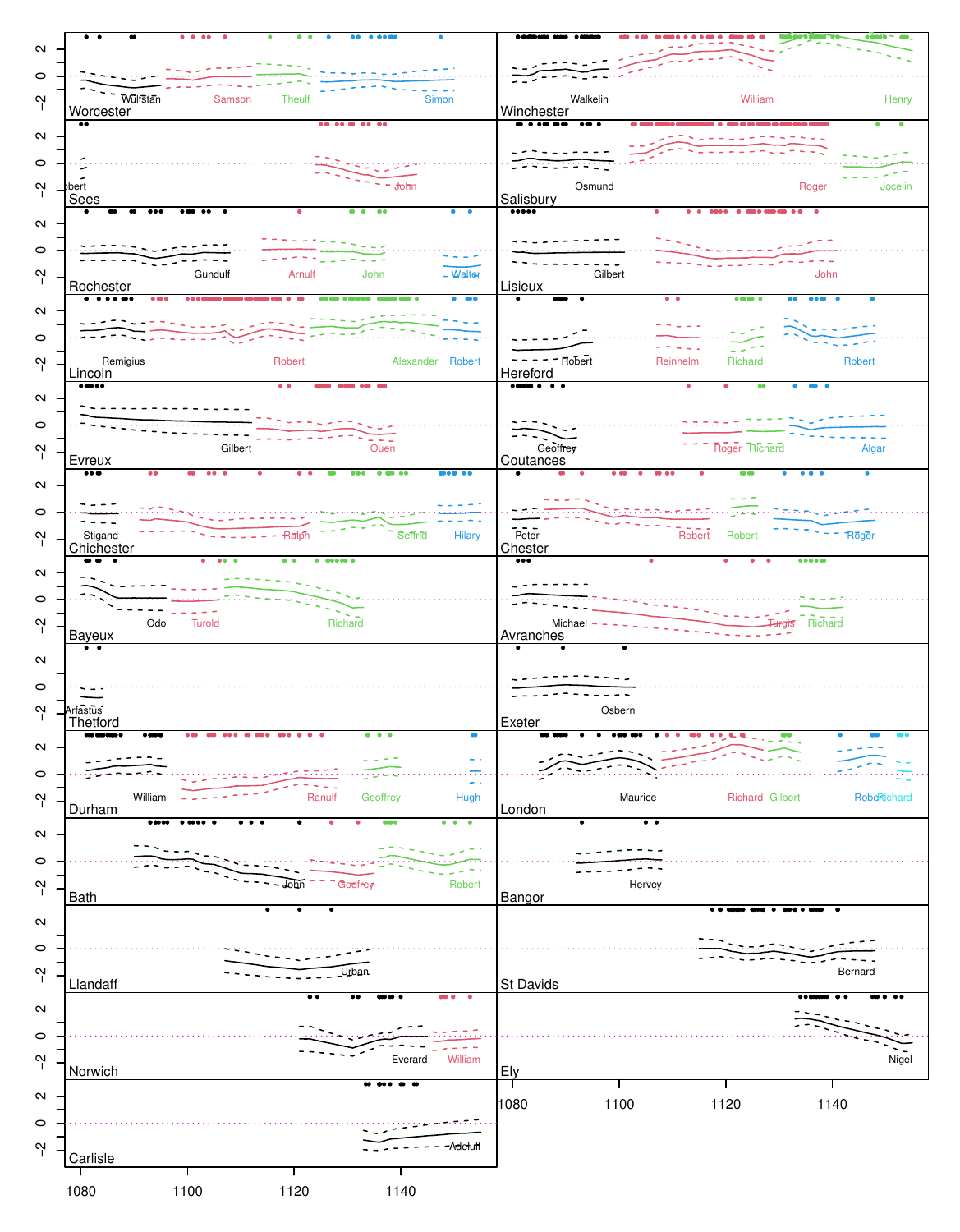}
    \caption{(constrained seniority effects analysis of \sec~\ref{sec:poHB2aRS6}) Bishop-{authority} curves $\bar U^{(t)}_j$ ($y$-axis values, solid curves) plotted for each bishop $j\in \M$ as a function of time ($x$-axis, the year) from $b_j$ to $e_j$ with uncertainty at one-sigma. The dots at the top of each graph show the times of lists in which the matched bishop below appeared.}
    \label{fig:Uhat-poHB2aRS6}
\end{figure}

The continuity in {authority} (but not status) of bishops over time within a diocese in \fig~\ref{fig:Uhat-poHB2aRS6} is noteworthy. There are some exceptions. For example, Henry de Blois started with higher {authority} than might be expected based only on the diocese. Some dioceses seem to be better (Winchester, London, Lincoln) than others (Chichester, Rochester, the dioceses in Normandy). The bishops of London and Winchester had gained precedence over their colleagues at the Council of London in 1075 and Lincoln came in the later middle ages to rank after Winchester.  However, there is uncertain evidence from as early as 1138 that the bishop of Lincoln might assume the role of London or Winchester in their absence and consequently that Lincoln already enjoyed a degree of precedence \citep{johnson13}. 

\subsection{Discussion of results}
\label{sec:discussion-of-results}

From a historical perspective, there are three significant outcomes. The first is the strong emphasis on the seniority and precedence of individual bishops in the witness lists. Historians often link the relative position of witnesses to an assessment of their political significance, but the analysis here shows that royal scribes were strongly influenced by the rules on seniority and precedence expressed at the Council of London, held by the English church in 1075 (Council of London, 1075, clause 1, \citet{clover79}).

The second is the position of Normandy within the Anglo-Norman realm. \fig~\ref{fig:Uhat-poHB2aRS6} shows that early in this period (before about 1100) Norman bishoprics enjoyed high status, but that this declined from the early twelfth century. This change is particularly marked for Avranches, Bayeux and Évreux, whilst no English bishoprics show a comparable trend. This should inform the ongoing debate about the relationship between England and Normandy (\citet{bates2013normans}, chapter 5). This change might represent a principled decision by royal scribes to rank Norman bishops lower in precedence than their English counterparts, or it might be explained politically. The smaller Norman dioceses may have been less attractive to ambitious clergymen, and there were periods when Normandy and England were ruled separately (most notably, 1144-54), so that Norman bishops were external to the English kingdom.

The third concerns how far the behaviour of individual bishops could change their status. Bishops were active politically and could fall into disgrace. Thus, Bishop Nigel of Ely had high status for a junior bishop in the 1130s, but from his disgrace in 1139 his status fell, contrary to the usual pattern. Nigel’s pattern is unique; it is not replicated by that for other disgraced bishops, such as Ranulf Flambard of Durham after 1100 and Alexander of Lincoln after 1139. These differences presumably reflect the nature of the disgrace itself. 
The estimated poset relations accord with known political favour.
For example, Henry of Blois and Odo of Bayeux, who were related to the royal house (Odo was William the Conqueror's half-brother, and Henry was Stephen's brother), are highly ranked. Referring to \fig~\ref{CPO-poHB17aRS4-all} in \app~\ref{app:poHB2aRS6}, although Henry (52) did well from the start, until 1134 he shares top spot in consensus orders with Roger of Salisbury (29).  From this date he is promoted ahead of anyone else.  This suggests that his brother becoming king in 1135 had an impact on his position.  

Referring to \fig~\ref{CPO-poHB17aRS4-all}, reconstructed orders seem relatively shallow, typically one (1097) or two (1124) groups of middle-ranked bishops and a few above or below. We may be concerned that this reflects an overly informative prior and differences in how often bishops witness. We tested this by simulating synthetic data in which the true posets were total orders (see \app~\ref{app:synth-total-orders}), but with the same list memberships as the real data.
We reconstructed the true total orders well. If the true orders were total orders we would see this in our analysis.

\section{Comparisons with other models}\label{sec:vsp-bucket-all}

In this section we define models over Bucket Orders and VSP orders. Calculation of $|\L[H]|$ for $H\in\H_\A,\ \A=\{1,\dots,m\}$ is linear in $m$ on these subspaces \citep{wells1971elements}, so if these models were preferred then we would use them. We find they are not adequate to represent a time-evolving hierarchy over long periods of time but can give a good fit over short time periods. \cite{jiang23} applies a fixed-time VSP model to all the witness list data (not just bishops). Some orders have over 200 actors, with lists exceeding 50 in length, and are out of reach for our full poset analysis. 

In \apps~\ref{app:PL-time-series-model} and \ref{app:pl-mix-elpd} we compare our model with Plackett-Luce models. Bayesian analysis of time-series Plackett-Luce in \app~\ref{app:PL-time-series-model} gives similar results for a parameter function corresponding to the average {authority} 
in \eq~\ref{eq:post-mean-U-process}. Analysis of a Plackett-Luce mixture model in \app~\ref{app:pl-mix-elpd} on short time intervals shows our model is preferred.

\subsection{Bucket Orders and Vertex-Series-Parallel partial orders}\label{sec:vsp-bucket-intro}

VSP orders \citep{valdes78,tarjan82} are built recursively from the ground set by taking series and parallel combinations of posets. We give this intuitive definition in \app~\ref{app:vsp-bucket-BF}. 
\cite{tarjan82} gives a concise characterisation. For any set $\A$, the class of all VSPs $\V_{\A}$ is identical to the set of posets $H\in\H_\A$ which do not contain a set of vertices $\A'=\{j_1,...,j_4\}$ with sub-graph $H'=H\cap (\A'\times \A')$ isomorphic to the ``forbidden subgraph'' $F=\{\e 1 2, \e 3 2, \e 3 4\}$ shown in \fig~\ref{fig:vsp-example} at right. After vertex relabelling, $F$ and $H'$ must be identical, so edges absent in $F$ are absent in $H'$. This makes it straightforward to test if a poset $H$ is a VSP-order.


A sub-class of VSPs called ``Bucket Orders'' has a particularly simple closed form for $|\L[H]|$. Actors are grouped into ``buckets''. Actors in the same bucket are unordered and a complete order holds over buckets. Formally, if $\K_\A$ is the class of Bucket Orders on $\A$ then $b\in\K_\A$ iff there is a partition $\A_1,...,\A_P$ of $\A$ into $P$ buckets such that for each $k\in [P]$ and all $j_1,j_2\in \A_k$ we have $\e{j_1}{j_2}\not \in b$ and for all pairs $1\le k_1<k_2\le P$ of buckets and all $j_1\in \A_{k_1}$ and $j_2\in\A_{k_2}$ we have $\e{j_1}{j_2}\in b$. 
VSPs and Bucket Orders are a small subset of partial orders. For example, if $|\A|=18$ (the largest for which \cite{OEIS22} gives cardinalities) we have $|\H_{[18]}|\simeq 2\times 10^{35}$, $|\V_{[18]}|/|\H_{[18]}\simeq 10^{-11}$ and $|\K_{[18]}|/|\H_{[18]}|\simeq 10^{-17}$.

\subsection{Bucket and VSP-order models} \label{vsp-bucket-model-spaces}
Suppose we are interested in learning about order relations over a period $[t_1,t_2]$ with $B\le t_1\le t_2\le E$. If we could justify restricting the process of fitted posets $h\in\H^{(t_1,t_2)}$ to a VSP-order-process $h\in \V^{(t_1,t_2)}$ with
\[
\V^{(t_1,t_2)}= \V_{\M_{t_1}}\times\V_{\M_{t_1+1}}\times...\times \V_{\M_{t_2}},
\]
or a bucket-order process $h\in \K^{(t_1,t_2)}$ with
\[
\K^{(t_1,t_2)}= \K_{\M_{t_1}}\times\K_{\M_{t_1+1}}\times...\times \K_{\M_{t_2}}, 
\]
then likelihood evaluation would be fast. However, this is not well-evidenced in our setting: the consensus order from 1136 in \fig~\ref{fig:CPO-poHB2aRS6} contains the poset $H'=\{\e{54}{57}, \e{47}{57}, \e{47}{58}\}$ in \fig~\ref{fig:intro-po-example}, with each included edge having posterior probability 0.8 or above and absent edges below 0.5, so the true poset is probably not a VSP as it contains a suborder ismorphic to $F$.

\subsection{Test results} \label{sec:vsp-bucket-model-and-test}
Bayes Factors $B_{\V,\H}$ and $B_{\K,\H}$ measuring the evidence for VSP poset-models and Bucket-Order models are defined in \sec~\ref{vsp-bucket-testing-methods}. We estimate these on the same twelve short (five-year) time intervals we used for the Plackett-Luce analysis in \app~\ref{app:pl-mix-elpd} for accurate estimation. The seniority covariate is a near-constant in these short time intervals so seniority effects were set to $\beta_r=0,\ r=1,...,\Sb$ in the fitted poset model. The likelihood $p_{(U)}$ in \eq~\ref{eq:lkd-noisy-up} was used. For analysis in interval $[t_1,t_2]$, we take $K$, the dimension of $U^{(t)}_j$,
to be $\lfloor \frac{1}{2} \max_{t\in[t_1,t_2]}m_t\rfloor$. The prior for $\rho$ is $\mbox{Beta}(1,1/6)$ and otherwise as \sec~\ref{sec:time-series-posterior}. 
\begin{figure}[ht]
    \centering
    \includegraphics[width=4.15in,trim=0in 0.2in 0.4in 0.8in, clip]{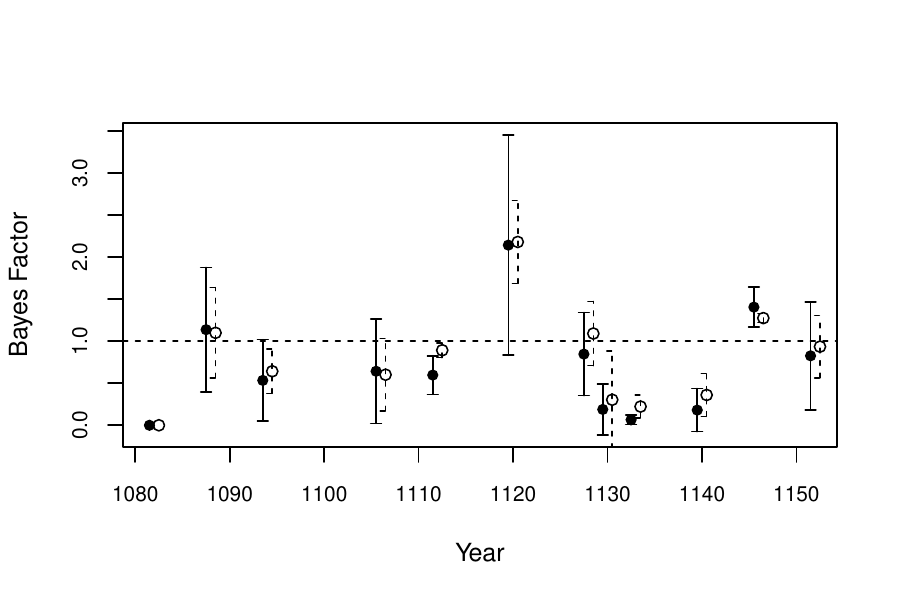}
    \caption{Bayes Factors for VSP orders (dashed lines) and Bucket Orders (solid) fitted over the year intervals in Table~\ref{tab:vsp-bucket-BF}. A value less than one is evidence against VSP or Bucket Orders. The $x$-axis value for each bar is the centre of the corresponding interval, $(t_1+t_2)/2$. Error bars are two standard deviations.}
    \label{fig:vsp-bucket-BF}
\end{figure}

Results are given in Table~\ref{tab:vsp-bucket-BF} in \app~\ref{app:vsp-bucket-BF} and plotted in \fig~\ref{fig:vsp-bucket-BF}. Each pair of points is an independent MCMC run. The prior probabilities $\pi(h\in\V^{(t_1,t_2)})$  and $\pi(h\in\B^{(t_1,t_2)})$ in Table~\ref{tab:vsp-bucket-BF} are surprisingly large given the sparsity of VSPs and Bucket Orders, so our prior for posets must favor VSPs and Bucket Orders. The rules forming VSP-orders and Bucket Orders compare groups rather than actors in a socially plausible way, so this may be a good thing.
As \fig~\ref{fig:vsp-bucket-BF} shows, the data favour partial orders or are neutral, except for 1118-1122, so Bucket Orders and VSPs may be acceptable over some short time intervals. In some intervals VSPs and Bucket Orders are rejected (1080-1084 and around 1132-1134). The evidence for posets visible in \fig~\ref{fig:vsp-bucket-BF} will accumulate over longer time series. 

\section{Conclusions}\label{sec:conclusions}

A new class of poset-models for time-series rank-data is summarised in \sec~\ref{sec:time-series-posterior}. 
The latent variable poset-parameterisation in \eq~\ref{eq:prior-for-h} made it straightforward to introduce a parameter controlling poset depth and incorporate actor-covariates informing the position of actors in the hierarchy.
We fit the model to witness-list data in which the actors are eleventh and twelfth century bishops. In \sec~\ref{sec:poHB1aRS6} we saw that the model recovered structure in the data which was anticipated by historians. In particular, the dependence of the status of a bishop on their seniority is clear in \fig~\ref{fig:beta-poHB1aRS6}. 
We checked for evidence that the depth parameter $\rho$, correlation $\theta$ and error probability $p$ varied over time by looking at short time intervals (see \app~\ref{app:constant-rho-theta-p-intervals}) and found no evidence against our assumption of constant values over time. Further support comes from model comparison against a Placket-Luce mixture in \app~\ref{app:pl-mix-elpd} which favoured our model. The time-series extension of the Plackett-Luce model in \app~\ref{app:PL-time-series-model} gave similar results for seniority effects and evolving {authority}, showing that the data overwhelm these model variations and conclusions are robust. The times-series Plackett-Luce model is fairly time consuming to fit so there was no great gain in efficiency over posets on our data.

We gave our main analysis in \sec~\ref{sec:poHB2aRS6}. This is the first quantitative analysis of these data and gave insights which historians find interesting. 
With few exceptions, witness-lists reflect precedence by diocese and seniority more than changing royal favour. 
We separated the effects of {authority} and seniority on status and confirmed \citep{johnson13} that the bishops of London, Winchester and Lincoln had high {authority} and that Rochester had no special status. 
Personal authority changes in a few cases: the high status originally given to Nigel of Ely unwound, and Roger of Salisbury bucked the trend. This was known to historians. The apparent decline in authority of Norman bishops (as expressed in the lists) was unknown. Consensus partial orders (see \fig~\ref{fig:CPO-poHB2aRS6}) gave an interpretable visualisation of the underlying social hierarchy in each year. The problem treated here appears in other guises: there is an extensive literature estimating animal dominance hierarchies with many similarities. For example, \cite{foerster16} studies a time-series of pairwise chimpanzee interactions and finds a role for seniority and evolution of status.

There is work to be done in computation and methodology to make these models more broadly applicable. First, our MCMC is time consuming (the experiments in this paper add up to about two years of MCMC if run as a single serial process). Whilst a careful analysis minimising any approximation of the target distribution is justified (there will never be any more for this period), scalable methods would be welcome, and VSP-orders may be an acceptable compromise from a modelling perspective. This may allow us to fit more complex noise models and treat ``top-$k$'' preference orders. Analysis of lists including lay witnesses \citep{jiang23} required scalable methods in order to count linear extensions in partial orders over hundreds of elements. Application of our methods to general preference orders on thousands of items and thousands of lists is not presently feasible.

Second, many applications of poset-based models will require substantial model-building, paralleling the evolution of Mallows and Plackett-Luce models and including hierarchical models and mixture models for structured populations and clustering. The clerks who made the lists may have differed on status assessment so unrepresented group structure may be present. Some developments are given in \cite{jiang24} (hereafter JN24). We worked with strong partial orders. Models for weak partial orders with ties are given in JN24. Also, our statistical model for ``noise'' in realised lists assumes errors occur in one direction only (queue jumping up or down, but not both). Noise models with bi-directional errors are explored in \cite{jiang23} and JN24. Statistical tools for selecting the number of features $K$ in the status vector of a bishop would remove the need for robustness checks (our \app~\ref{app:K-equals-18-posterior}). This is addressed for fixed-time data in JN24 using reversible-jump MCMC. 

Third, a small number of other covariates beside seniority are available in the data and might be explored in model elaboration. Covariates might enter the noise model also, to inform the probability and magnitude of displacements. One unexplored weakness of our error model is apparent in lists of length two: if the two bishops are ordered in a deep poset then the noise model assigns the same probability for the ``wrong'' order whether the two bishops are at each end of the poset (very different status) or adjacent in its transitive reduction (near-equal status). Displacements are probably ``over-dispersed'' in the noise-model.

From a historical perspective this study raises several questions. Complex precedence structures seem to have existed, but how were they known? Was there some kind of precedence handbook, or other means of transmission? Comparisons between patterns of precedence relations in pre-conquest lists and lists from later periods might be revealing. Later documents are more accurately dated and so a more fine-grained analysis may be possible. Finally, there are forgeries among acta of this period. Cases that go against the usual pattern may be an indicator of forgery. 

Software used to carry out the analysis presented in this paper are available at \\
\[
\mbox{\url{https://github.com/gknicholls/Partial-order-HMM-public.git}}.
\]

\section*{Acknowledgements}
The authors thank Dr Simon Urbanek for writing an R wrapper to compile the lecount() code in c++. 
GKN thanks Dr Oemetse Mogapi and Prof Tom Snijders for introducing him to this topic and Sir Bernard Silverman for helpful conversations regarding marginal consistency. 

\begin{funding}
JL was supported in part by Marsden grant MFP-UOA2131 and HRC 22/377/A from New Zealand Government funding.
\end{funding}

\bibliographystyle{chicago}
\bibliography{PO}

\begin{thebibliography}{}

\bibitem[\protect\citeauthoryear{Arcagni, Avellone, and Fattore}{Arcagni
  et~al.}{2022}]{arcagni22}
Arcagni, A., A.~Avellone, and M.~Fattore (2022).
\newblock Complexity reduction and approximation of multidomain systems of
  partially ordered data.
\newblock {\em Computational Statistics \& Data Analysis\/}~{\em 173}, 107520.

\bibitem[\protect\citeauthoryear{Asfaw, Vitelli, S{\o}rensen, Arjas, and
  Frigessi}{Asfaw et~al.}{2017}]{asfaw17}
Asfaw, D., V.~Vitelli, {\O}.~S{\o}rensen, E.~Arjas, and A.~Frigessi (2017).
\newblock Time-varying rankings with the {Bayesian Mallows} model.
\newblock {\em Stat\/}~{\em 6\/}(1), 14--30.

\bibitem[\protect\citeauthoryear{Bates}{Bates}{1997}]{bates97}
Bates, D. (1997).
\newblock {\em Family Trees and the Roots of Politics: The Prosopography of
  Britain and France from the Tenth to the Twelfth Century}, Chapter The
  Prosopographical Study of Anglo-Norman Royal Charters, pp.\  89--102.
\newblock Boydell Press.

\bibitem[\protect\citeauthoryear{Bates}{Bates}{1998}]{bates98}
Bates, D. (1998).
\newblock {\em The Acta of William I, 1066–1087}.
\newblock Regesta regum Anglo-Normannorum 1066-1154. Oxford.

\bibitem[\protect\citeauthoryear{Bates}{Bates}{2013}]{bates2013normans}
Bates, D. (2013).
\newblock {\em The Normans and Empire}.
\newblock Ford lectures. OUP Oxford.

\bibitem[\protect\citeauthoryear{Baxter}{Baxter}{2007}]{baxter07}
Baxter, S. (2007).
\newblock {\em {The Earls of Mercia: Lordship and Power in Late Anglo-Saxon
  England}}.
\newblock Oxford University Press.

\bibitem[\protect\citeauthoryear{Beerenwinkel, Eriksson, and
  Sturmfels}{Beerenwinkel et~al.}{2007}]{beerenwinkel2007conjunctive}
Beerenwinkel, N., N.~Eriksson, and B.~Sturmfels (2007).
\newblock Conjunctive {B}ayesian networks.
\newblock {\em Bernoulli\/}~{\em 13\/}(4), 893--909.

\bibitem[\protect\citeauthoryear{Bengio and Grandvalet}{Bengio and
  Grandvalet}{2004}]{Bengio2004}
Bengio, Y. and Y.~Grandvalet (2004).
\newblock No unbiased estimator of the variance of k-fold cross-validation.
\newblock {\em Journal of Machine Learning Research\/}~{\em 5}, 1089–1105.

\bibitem[\protect\citeauthoryear{Bogart}{Bogart}{1973a}]{Bogart73a}
Bogart, K.~P. (1973a, 5).
\newblock Maximal dimensional partially ordered sets {I}. {H}iraguchi's
  theorem.
\newblock {\em Discrete Mathematics\/}~{\em 5}, 21--31.

\bibitem[\protect\citeauthoryear{Bogart}{Bogart}{1973b}]{bogart73b}
Bogart, K.~P. (1973b).
\newblock Preference structures i: Distances between transitive preference
  relations.
\newblock {\em The Journal of Mathematical Sociology\/}~{\em 3\/}(1), 49--67.

\bibitem[\protect\citeauthoryear{Brightwell}{Brightwell}{1993}]{brightwell93}
Brightwell, G. (1993).
\newblock {\em Surveys in Combinatorics}, Volume 187 of {\em London
  Mathematical Society Lecture Note Series}, Chapter Models of random partial
  orders, pp.\  53--83.
\newblock Cambridge Univeristy Press.

\bibitem[\protect\citeauthoryear{Brightwell and Winkler}{Brightwell and
  Winkler}{1991}]{brightwell1991counting}
Brightwell, G. and P.~Winkler (1991).
\newblock Counting linear extensions.
\newblock {\em Order\/}~{\em 8\/}(3), 225--242.

\bibitem[\protect\citeauthoryear{Caron and Doucet}{Caron and
  Doucet}{2012}]{Caron12mcmc}
Caron, F. and A.~Doucet (2012).
\newblock Efficient {B}ayesian inference for generalized {B}radley-{T}erry
  models.
\newblock {\em J. Comput. Graph. Statist.\/}~{\em 21\/}(1), 174–196.

\bibitem[\protect\citeauthoryear{Caron and Teh}{Caron and
  Teh}{2012}]{caron12time}
Caron, F. and Y.~Teh (2012).
\newblock {B}ayesian nonparametric models for ranked data.
\newblock In F.~Pereira, C.~Burges, L.~Bottou, and K.~Weinberger (Eds.), {\em
  Advances in Neural Information Processing Systems}, Volume~25. Curran
  Associates, Inc.

\bibitem[\protect\citeauthoryear{Caron, Teh, and Murphy}{Caron
  et~al.}{2014}]{caron2014bayesian}
Caron, F., Y.~W. Teh, and T.~B. Murphy (2014).
\newblock Bayesian nonparametric {P}lackett--{L}uce models for the analysis of
  preferences for college degree programmes.
\newblock {\em The Annals of Applied Statistics\/}, 1145--1181.

\bibitem[\protect\citeauthoryear{Clover and Gibson}{Clover and
  Gibson}{1979}]{clover79}
Clover, H. and M.~Gibson (1979).
\newblock {\em The Letters of Lanfranc, Archbishop of Canterbury}.
\newblock Oxford.

\bibitem[\protect\citeauthoryear{Cronne and Davis}{Cronne and
  Davis}{1968}]{cronne68}
Cronne, H.~A. and R.~H.~C. Davis (1968).
\newblock {\em Regesta Regis Stephani ac Mathildis Imperatricis ac Gaufridi et
  Henrici Ducum Normannorum, 1135–1154}.
\newblock Regesta regum Anglo-Normannorum 1066-1154. Oxford.

\bibitem[\protect\citeauthoryear{Davis and Whitwell}{Davis and
  Whitwell}{1913}]{davis13}
Davis, H. and R.~Whitwell (1913).
\newblock {\em Regesta regum Anglo-Normannorum 1066-1154: Regesta Willelmi
  Conquestoris et Willelmi Rufi 1066-1100}.
\newblock Regesta regum Anglo-Normannorum 1066-1154. Oxford.

\bibitem[\protect\citeauthoryear{Deng, Han, Li, and Liu}{Deng
  et~al.}{2014}]{deng14}
Deng, K., S.~Han, K.~J. Li, and J.~S. Liu (2014).
\newblock {B}ayesian aggregation of order-based rank data.
\newblock {\em Journal of the American Statistical Association\/}~{\em
  109\/}(507), 1023--1039.

\bibitem[\protect\citeauthoryear{Diaconis}{Diaconis}{1988}]{diaconis88}
Diaconis, P. (1988).
\newblock Group representations in probability and statistics.
\newblock {\em Lecture Notes-Monograph Series\/}~{\em 11}, i--192.

\bibitem[\protect\citeauthoryear{Durante, Dunson, and Vogelstein}{Durante
  et~al.}{2017a}]{durante17}
Durante, D., D.~B. Dunson, and J.~T. Vogelstein (2017a).
\newblock Nonparametric {B}ayes modeling of populations of networks.
\newblock {\em Journal of the American Statistical Association\/}~{\em
  112\/}(520), 1516--1530.

\bibitem[\protect\citeauthoryear{Durante, Dunson, and Vogelstein}{Durante
  et~al.}{2017b}]{durante17rejoinder}
Durante, D., D.~B. Dunson, and J.~T. Vogelstein (2017b).
\newblock Rejoinder: Nonparametric {B}ayes modeling of populations of networks.
\newblock {\em Journal of the American Statistical Association\/}~{\em
  112\/}(520), 1547--1552.

\bibitem[\protect\citeauthoryear{Dushnik and Miller}{Dushnik and
  Miller}{1941}]{dushnik1941}
Dushnik, B. and E.~W. Miller (1941).
\newblock Partially ordered sets.
\newblock {\em American Journal of Mathematics\/}~{\em 63\/}(3), 600--610.

\bibitem[\protect\citeauthoryear{??}{FEA}{2022}]{FastiEcclesiaeAnglicanae}
FEA (2022).
\newblock Fasti {E}cclesiae {A}nglicanae, 1066-1300.
\newblock London: Institute of Historical Research; British History Online;
  Last accessed 22 October 2022;
  http://www.british-history.ac.uk/fasti-ecclesiae/.

\bibitem[\protect\citeauthoryear{Fishburn and Gehrlein}{Fishburn and
  Gehrlein}{1975}]{fishburn75}
Fishburn, P.~C. and W.~V. Gehrlein (1975).
\newblock A comparative analysis of methods for constructing weak orders from
  partial orders.
\newblock {\em J. Math. Sociology\/}~{\em 4\/}(1), 93–102.

\bibitem[\protect\citeauthoryear{Fligner and Verducci}{Fligner and
  Verducci}{1986}]{fligner86}
Fligner, M.~A. and J.~S. Verducci (1986).
\newblock Distance based ranking models.
\newblock {\em Journal of the Royal Statistical Society. Series B
  (Methodological)\/}~{\em 48\/}(3), 359--369.

\bibitem[\protect\citeauthoryear{Foerster, Franz, Murray, Gilby, Feldblum,
  Walker, and Pusey}{Foerster et~al.}{2016}]{foerster16}
Foerster, S., M.~Franz, C.~M. Murray, I.~C. Gilby, J.~T. Feldblum, K.~K.
  Walker, and A.~E. Pusey (2016).
\newblock Chimpanzee females queue but males compete for social status.
\newblock {\em Scientific Reports\/}~{\em 6\/}(1), 35404--35415.

\bibitem[\protect\citeauthoryear{Friedell}{Friedell}{1967}]{friedell67}
Friedell, M.~F. (1967).
\newblock Organizations as semilattices.
\newblock {\em American Sociological Review\/}~{\em 32\/}(1), 46--54.

\bibitem[\protect\citeauthoryear{Froehlich, Fellmann, Sueltmann, Poustka, and
  Beissbarth}{Froehlich et~al.}{2007}]{froelich07}
Froehlich, H., M.~Fellmann, H.~Sueltmann, A.~Poustka, and T.~Beissbarth (2007).
\newblock Large scale statistical inference of signaling pathways from {RNA}i
  and microarray data.
\newblock {\em BMC Bioinformatics\/}~{\em 8\/}(386).

\bibitem[\protect\citeauthoryear{Gazes, Hampton, and Lourenco}{Gazes
  et~al.}{2017}]{gazes17}
Gazes, R.~P., R.~R. Hampton, and S.~F. Lourenco (2017).
\newblock Transitive inference of social dominance by human infants.
\newblock {\em Developmental Science\/}~{\em 20\/}(2), e12367.

\bibitem[\protect\citeauthoryear{Gionis, Mannila, Puolam{\"a}ki, and
  Ukkonen}{Gionis et~al.}{2006}]{mannila06}
Gionis, A., H.~Mannila, K.~Puolam{\"a}ki, and A.~Ukkonen (2006).
\newblock Algorithms for discovering bucket orders from data.
\newblock In {\em Proceedings of the 12th ACM SIGKDD international conference
  on Knowledge discovery and data mining}, pp.\  561--566.

\bibitem[\protect\citeauthoryear{Given-Wilson}{Given-Wilson}{1991}]{given-wilson91}
Given-Wilson, C. (1991).
\newblock Royal charter witness lists 1327-1399.
\newblock {\em Medieval Prosopography\/}~{\em 12\/}(2), 35--93.

\bibitem[\protect\citeauthoryear{Glickman and Hennessy}{Glickman and
  Hennessy}{2015}]{glickman2015stochastic}
Glickman, M.~E. and J.~Hennessy (2015).
\newblock A stochastic rank ordered logit model for rating multi-competitor
  games and sports.
\newblock {\em Journal of Quantitative Analysis in Sports\/}~{\em 11\/}(3),
  131--144.

\bibitem[\protect\citeauthoryear{Guiver and Snelson}{Guiver and
  Snelson}{2009}]{guiver2009bayesian}
Guiver, J. and E.~Snelson (2009).
\newblock {B}ayesian inference for {Plackett-Luce} ranking models.
\newblock In {\em proceedings of the 26th annual international conference on
  machine learning}, pp.\  377--384.

\bibitem[\protect\citeauthoryear{Gwee, Gormley, and Fop}{Gwee
  et~al.}{2023}]{gwee23}
Gwee, X.~Y., I.~C. Gormley, and M.~Fop (2023).
\newblock A latent shrinkage position model for binary and count network data.
\newblock {\em Bayesian Analysis\/}, 1 -- 29.

\bibitem[\protect\citeauthoryear{Haskins}{Haskins}{1938}]{haskins38}
Haskins, G.~L. (1938).
\newblock Charter witness lists in the reign of king john.
\newblock {\em Speculum\/}~{\em 13\/}(3), 319--325.

\bibitem[\protect\citeauthoryear{Hiraguchi}{Hiraguchi}{1951}]{Hiraguchi51}
Hiraguchi, T. (1951).
\newblock On the dimension of partially ordered sets.
\newblock {\em The science reports of the Kanazawa University\/}~{\em 1},
  77--94.

\bibitem[\protect\citeauthoryear{Hunter}{Hunter}{2004}]{hunter04}
Hunter, D.~R. (2004).
\newblock {MM algorithms for generalized Bradley-Terry models}.
\newblock {\em The Annals of Statistics\/}~{\em 32\/}(1), 384 -- 406.

\bibitem[\protect\citeauthoryear{Irurozki, Calvo, and Lozano}{Irurozki
  et~al.}{2016}]{irurozki16}
Irurozki, E., B.~Calvo, and J.~A. Lozano (2016).
\newblock {PerMallows: An R Package for Mallows and generalized Mallows
  models}.
\newblock {\em Journal of Statistical Software\/}~{\em 71\/}(12), 1–30.

\bibitem[\protect\citeauthoryear{Irurozki, Calvo, and Lozano}{Irurozki
  et~al.}{2019}]{irurozki19}
Irurozki, E., B.~Calvo, and J.~A. Lozano (2019).
\newblock {Mallows and generalized Mallows model for matchings}.
\newblock {\em Bernoulli\/}~{\em 25\/}(2), 1160 -- 1188.

\bibitem[\protect\citeauthoryear{Jiang and Nicholls}{Jiang and
  Nicholls}{2024}]{jiang24}
Jiang, C. and G.~K. Nicholls (2024).
\newblock Non-parametric {B}ayesian inference for partial orders with ties from
  rank data observed with {M}allows noise.
\newblock https://arxiv.org/abs/2408.14661.

\bibitem[\protect\citeauthoryear{Jiang, Nicholls, and Lee}{Jiang
  et~al.}{2023}]{jiang23}
Jiang, C., G.~K. Nicholls, and J.~E. Lee (2023, 31 Jul--04 Aug).
\newblock {B}ayesian inference for vertex-series-parallel partial orders.
\newblock In R.~J. Evans and I.~Shpitser (Eds.), {\em Proceedings of the
  Thirty-Ninth Conference on Uncertainty in Artificial Intelligence}, Volume
  216 of {\em Proceedings of Machine Learning Research}, pp.\  995--1004. PMLR.

\bibitem[\protect\citeauthoryear{Johnson and Cronne}{Johnson and
  Cronne}{1966}]{johnson66}
Johnson, C. and H.~A. Cronne (1966).
\newblock {\em Regesta Henrici Primi : 1100-1135}.
\newblock Regesta regum Anglo-Normannorum 1066-1154. Oxford.

\bibitem[\protect\citeauthoryear{Johnson}{Johnson}{2013}]{johnson13}
Johnson, D. (2013, 10).
\newblock {Bishops and deans: London and the province of Canterbury in the
  twelfth century*}.
\newblock {\em Historical Research\/}~{\em 86\/}(234), 551--578.

\bibitem[\protect\citeauthoryear{Kangas, Hankala, Niinim\"{a}ki, and
  Koivisto}{Kangas et~al.}{2016}]{koivisto16b}
Kangas, K., T.~Hankala, T.~Niinim\"{a}ki, and M.~Koivisto (2016).
\newblock Counting linear extensions of sparse posets.
\newblock In {\em Proceedings of the Twenty-Fifth International Joint
  Conference on Artificial Intelligence}, IJCAI'16, pp.\  603–609. AAAI
  Press.

\bibitem[\protect\citeauthoryear{Kangas, Koivisto, and Salonen}{Kangas
  et~al.}{2019}]{koivisto19}
Kangas, K., M.~Koivisto, and S.~Salonen (2019).
\newblock A faster tree-decomposition based algorithm for counting linear
  extensions.
\newblock {\em Algorithmica\/}, 1--18.

\bibitem[\protect\citeauthoryear{Karzanov and Khachiyan}{Karzanov and
  Khachiyan}{1991}]{Karzanov91}
Karzanov, A. and L.~Khachiyan (1991).
\newblock On the conductance of order {M}arkov chains.
\newblock {\em Order\/}~{\em 8}, 7--15.

\bibitem[\protect\citeauthoryear{Kass and Raftery}{Kass and
  Raftery}{1995}]{kass95}
Kass, R.~E. and A.~E. Raftery (1995).
\newblock Bayes factors.
\newblock {\em Journal of the American Statistical Association\/}~{\em
  90\/}(430), 773--795.

\bibitem[\protect\citeauthoryear{Keynes}{Keynes}{1980}]{keynes80}
Keynes, S. (1980).
\newblock {\em The Diplomas of King Aethlred “the Unready” 978–1016}.
\newblock Cambridge Studies in Medieval Life and Thought: Third Series.
  Cambridge University Press.

\bibitem[\protect\citeauthoryear{Kleitman and Rothschild}{Kleitman and
  Rothschild}{1975}]{kleitman1975asymptotic}
Kleitman, D.~J. and B.~L. Rothschild (1975).
\newblock Asymptotic enumeration of partial orders on a finite set.
\newblock {\em Transactions of the American Mathematical Society\/}~{\em 205},
  205--220.

\bibitem[\protect\citeauthoryear{Knuth and Szwarcfiter}{Knuth and
  Szwarcfiter}{1974}]{knuth1974structured}
Knuth, D.~E. and J.~L. Szwarcfiter (1974).
\newblock A structured program to generate all topological sorting
  arrangements.
\newblock {\em Information Processing Letters\/}~{\em 2\/}(6), 153--157.

\bibitem[\protect\citeauthoryear{Licence}{Licence}{2020}]{licence20}
Licence, T. (2020, 09).
\newblock {\em {Edward the Confessor: Last of the Royal Blood}}.
\newblock Yale University Press.

\bibitem[\protect\citeauthoryear{Lu and Boutilier}{Lu and
  Boutilier}{2014}]{lu2014effective}
Lu, T. and C.~Boutilier (2014).
\newblock Effective sampling and learning for {M}allows models with
  pairwise-preference data.
\newblock {\em J. Mach. Learn. Res.\/}~{\em 15\/}(1), 3783--3829.

\bibitem[\protect\citeauthoryear{Luce}{Luce}{1977}]{luce77}
Luce, R. (1977).
\newblock The choice axiom after twenty years.
\newblock {\em Journal of Mathematical Psychology\/}~{\em 15\/}(3), 215--233.

\bibitem[\protect\citeauthoryear{Luce}{Luce}{1959}]{luce1959possible}
Luce, R.~D. (1959).
\newblock On the possible psychophysical laws.
\newblock {\em Psychological review\/}~{\em 66\/}(2), 81.

\bibitem[\protect\citeauthoryear{Mallows}{Mallows}{1957}]{mallows1957non}
Mallows, C.~L. (1957).
\newblock Non-null ranking models. i.
\newblock {\em Biometrika\/}~{\em 44\/}(1/2), 114--130.

\bibitem[\protect\citeauthoryear{Mannila}{Mannila}{2008}]{mannila08}
Mannila, H. (2008).
\newblock Finding total and partial orders from data for seriation.
\newblock In J.-F. Boucault (Ed.), {\em Discovery Science}, Volume 5255 of {\em
  LNAI}, Berlin Heidelberg, pp.\  16--25. Springer-Verlag.

\bibitem[\protect\citeauthoryear{Mannila and Meek}{Mannila and
  Meek}{2000}]{mannila00}
Mannila, H. and C.~Meek (2000).
\newblock Global partial orders from sequential data.
\newblock In {\em Proceedings of the Sixth ACM SIGKDD International Conference
  on Knowledge Discovery and Data Mining}, KDD '00, New York, NY, USA, pp.\
  161–168. Association for Computing Machinery.

\bibitem[\protect\citeauthoryear{Martin}{Martin}{2002}]{martin02}
Martin, J.~L. (2002).
\newblock Some algebraic structures for diffusion in social networks.
\newblock {\em The Journal of Mathematical Sociology\/}~{\em 26\/}(3),
  123--146.

\bibitem[\protect\citeauthoryear{McKeough and Glickman}{McKeough and
  Glickman}{2024}]{glickman24}
McKeough, K. and M.~Glickman (2024).
\newblock {P}lackett–{L}uce modeling with trajectory models for measuring
  athlete strength.
\newblock {\em Journal of Quantitative Analysis in Sports\/}~{\em 20\/}(1),
  21--35.

\bibitem[\protect\citeauthoryear{Meilă and Chen}{Meilă and
  Chen}{2016}]{meila16}
Meilă, M. and H.~Chen (2016).
\newblock {B}ayesian non-parametric clustering of ranking data.
\newblock {\em IEEE Transactions on Pattern Analysis and Machine
  Intelligence\/}~{\em 38\/}(11), 2156--2169.

\bibitem[\protect\citeauthoryear{Meil\u{a} and Chen}{Meil\u{a} and
  Chen}{2010}]{meila2012dirichlet}
Meil\u{a}, M. and H.~Chen (2010).
\newblock Dirichlet process mixtures of generalized {M}allows models.
\newblock In {\em Proceedings of the Twenty-Sixth Conference on Uncertainty in
  Artificial Intelligence}, UAI'10, Arlington, Virginia, USA, pp.\  358–367.
  AUAI Press.

\bibitem[\protect\citeauthoryear{Mogapi}{Mogapi}{2009}]{mogapi09}
Mogapi, O. (2009).
\newblock {\em A Latent Partial Order Model for Social Networks}.
\newblock Ph.\ D. thesis, University of Oxford.

\bibitem[\protect\citeauthoryear{Mollica and Tardella}{Mollica and
  Tardella}{2017}]{Mollica14}
Mollica, C. and L.~Tardella (2017).
\newblock {B}ayesian {P}lackett-{L}uce mixture models for partially ranked
  data.
\newblock {\em Psychometrika\/}~{\em 82\/}(2), 442--458.

\bibitem[\protect\citeauthoryear{Mollica and Tardella}{Mollica and
  Tardella}{2020}]{Mollica20}
Mollica, C. and L.~Tardella (2020).
\newblock {PLMIX}: an {R} package for modelling and clustering partially ranked
  data.
\newblock {\em Journal of Statistical Computation and Simulation\/}~{\em
  90\/}(5), 925--959.

\bibitem[\protect\citeauthoryear{{Muir Watt}}{{Muir Watt}}{2015}]{muirwatt15}
{Muir Watt}, A. (2015).
\newblock {\em Inference for partial orders from random linear extensions}.
\newblock Ph.\ D. thesis, University of Oxford.

\bibitem[\protect\citeauthoryear{Nicholls and Muir~Watt}{Nicholls and
  Muir~Watt}{2011}]{nicholls11}
Nicholls, G.~K. and A.~Muir~Watt (2011).
\newblock Partial order models for episcopal social status in 12th century
  {E}ngland.
\newblock {\em IWSM 2011\/}, 437.

\bibitem[\protect\citeauthoryear{Niinim\"{a}ki, Parviainen, and
  Koivisto}{Niinim\"{a}ki et~al.}{2016}]{koivisto16a}
Niinim\"{a}ki, T., P.~Parviainen, and M.~Koivisto (2016).
\newblock Structure discovery in {B}ayesian networks by sampling partial
  orders.
\newblock {\em Journal of Machine Learning Research\/}~{\em 17\/}(57), 1--47.

\bibitem[\protect\citeauthoryear{OEIS Foundation~Inc}{OEIS
  Foundation~Inc}{2022}]{OEIS22}
OEIS Foundation~Inc, . (2022).
\newblock The on-line encyclopedia of integer sequences.
\newblock Published electronically at https://oeis.org.

\bibitem[\protect\citeauthoryear{Plackett}{Plackett}{1975}]{plackett1975analysis}
Plackett, R.~L. (1975).
\newblock The analysis of permutations.
\newblock {\em Journal of the Royal Statistical Society: Series C (Applied
  Statistics)\/}~{\em 24\/}(2), 193--202.

\bibitem[\protect\citeauthoryear{Rising}{Rising}{2021}]{rising21}
Rising, J. (2021).
\newblock Uncertainty in ranking.
\newblock https://arxiv.org/abs/2107.03459.

\bibitem[\protect\citeauthoryear{Roberts}{Roberts}{1990}]{roberts90}
Roberts, J.~M. (1990).
\newblock Modeling hierarchy: Transitivity and the linear ordering problem.
\newblock {\em The Journal of Mathematical Sociology\/}~{\em 16\/}(1), 77--87.

\bibitem[\protect\citeauthoryear{Rousseau and Mengersen}{Rousseau and
  Mengersen}{2011}]{rousseau11}
Rousseau, J. and K.~Mengersen (2011).
\newblock Asymptotic behaviour of the posterior distribution in overfitted
  mixture models.
\newblock {\em Journal of the Royal Statistical Society: Series B (Statistical
  Methodology)\/}~{\em 73\/}(5), 689--710.

\bibitem[\protect\citeauthoryear{Russell}{Russell}{1937}]{russell37}
Russell, J.~C. (1937).
\newblock Social status at the court of {K}ing {J}ohn.
\newblock {\em Speculum\/}~{\em 12\/}(3), 319--329.

\bibitem[\protect\citeauthoryear{Sakoparnig and Beerenwinkel}{Sakoparnig and
  Beerenwinkel}{2012}]{beerenwinkel12}
Sakoparnig, T. and N.~Beerenwinkel (2012).
\newblock {Efficient sampling for Bayesian inference of conjunctive Bayesian
  networks}.
\newblock {\em Bioinformatics\/}~{\em 28\/}(18), 2318--2324.

\bibitem[\protect\citeauthoryear{Seshadri, Ragain, and Ugander}{Seshadri
  et~al.}{2020}]{seshadri20}
Seshadri, A., S.~Ragain, and J.~Ugander (2020).
\newblock Learning rich rankings.
\newblock In H.~Larochelle, M.~Ranzato, R.~Hadsell, M.~Balcan, and H.~Lin
  (Eds.), {\em Advances in Neural Information Processing Systems}, Volume~33,
  pp.\  9435--9446. Curran Associates, Inc.

\bibitem[\protect\citeauthoryear{Sharpe, Carpenter, Doherty, Hagger, and
  Karn}{Sharpe et~al.}{2014}]{sharpe14}
Sharpe, R., D.~Carpenter, H.~Doherty, M.~Hagger, and N.~Karn (2014).
\newblock The charters of {W}illiam {II} and {H}enry {I}.
\newblock Online: Accessed October 2022;
  \verb|actswilliam2henry1.wordpress.com/the-charters/|.

\bibitem[\protect\citeauthoryear{Shizuka and McDonald}{Shizuka and
  McDonald}{2012}]{shizuka12}
Shizuka, D. and D.~B. McDonald (2012).
\newblock A social network perspective on measurements of dominance
  hierarchies.
\newblock {\em Animal Behaviour\/}~{\em 83\/}(4), 925--934.

\bibitem[\protect\citeauthoryear{Sivula, Magnusson, and Nehtari}{Sivula
  et~al.}{2022}]{Sivula2022}
Sivula, T., M.~Magnusson, and A.~Nehtari (2022).
\newblock Unbiased estimator for the variance of the leave-one-out
  cross-validation estimator for a {B}ayesian normal model with fixed variance.
\newblock {\em Communications in Statistics - Theory and Methods\/}, 1--23.

\bibitem[\protect\citeauthoryear{S{\o}rensen, Crispino, Liu, and
  Vitelli}{S{\o}rensen et~al.}{2020}]{vitelli20sofware}
S{\o}rensen, {\O}., M.~Crispino, Q.~Liu, and V.~Vitelli (2020).
\newblock {BayesMallows: An R Package for the Bayesian Mallows Model}.
\newblock {\em {The R Journal}\/}~{\em 12\/}(1), 324--342.

\bibitem[\protect\citeauthoryear{Tkachenko and Lauw}{Tkachenko and
  Lauw}{2016}]{tkachenko2016plackett}
Tkachenko, M. and H.~W. Lauw (2016).
\newblock Plackett-{L}uce regression mixture model for heterogeneous rankings.
\newblock In {\em Proceedings of the 25th ACM International on Conference on
  Information and Knowledge Management}, pp.\  237--246.

\bibitem[\protect\citeauthoryear{Valdes}{Valdes}{1978}]{valdes78}
Valdes, J. (1978).
\newblock {\em Parsing {F}lowcharts and {S}eries-{P}arallel {G}raphs.}
\newblock Ph.\ D. thesis, Stanford, CA, USA.
\newblock AAI7905944.

\bibitem[\protect\citeauthoryear{Valdes, Tarjan, and Lawler}{Valdes
  et~al.}{1982}]{tarjan82}
Valdes, J., R.~E. Tarjan, and E.~L. Lawler (1982).
\newblock The recognition of series parallel digraphs.
\newblock {\em SIAM Journal on Computing\/}~{\em 11\/}(2), 298--313.

\bibitem[\protect\citeauthoryear{van Wietmarschen}{van
  Wietmarschen}{2022}]{vanWietmarschenNous22}
van Wietmarschen, H. (2022).
\newblock What is social hierarchy?
\newblock {\em Noûs\/}~{\em 56\/}(4), 920--939.
\newblock (page 926, {\it "In sum,..."}).

\bibitem[\protect\citeauthoryear{Vasconcelos}{Vasconcelos}{2008}]{vasconcelos08}
Vasconcelos, M. (2008).
\newblock Transitive inference in non-human animals: An empirical and
  theoretical analysis.
\newblock {\em Behavioural Processes\/}~{\em 78\/}(3), 313--334.

\bibitem[\protect\citeauthoryear{Vehtari, Gelman, and Gabry}{Vehtari
  et~al.}{2017}]{Vehtari17}
Vehtari, A., A.~Gelman, and J.~Gabry (2017).
\newblock Practical {B}ayesian model evaluation using leave-one-out
  cross-validation and {WAIC}.
\newblock {\em Statistics and Computing\/}~{\em 27\/}(5), 1413–1432.

\bibitem[\protect\citeauthoryear{Vitelli, {{\O}}ystein S{{\o}}rensen, Crispino,
  Frigessi, and Arjas}{Vitelli et~al.}{2018}]{vitelli18}
Vitelli, V., {{\O}}ystein S{{\o}}rensen, M.~Crispino, A.~Frigessi, and E.~Arjas
  (2018).
\newblock Probabilistic preference learning with the {M}allows rank model.
\newblock {\em Journal of Machine Learning Research\/}~{\em 18\/}(158), 1--49.

\bibitem[\protect\citeauthoryear{Wells}{Wells}{1971}]{wells1971elements}
Wells, M.~B. (1971).
\newblock Elements of combinatorial computing.

\bibitem[\protect\citeauthoryear{Winkler}{Winkler}{1985}]{winkler1985random}
Winkler, P. (1985).
\newblock Random orders.
\newblock {\em Order\/}~{\em 1\/}(4), 317--331.

\end{thebibliography}

\newpage

\title{Bayesian inference for partial orders from random linear extensions: power relations from 12th Century Royal Acta (Supplementary Material)}

\begin{aug}
\author[AA]{\fnms{Geoff K. Nicholls}\ead[label=ee1]{nicholls@stats.ox.ac.uk}}
\author[BB]{\fnms{Jeong Eun Lee}\ead[label=ee2]{kate.lee@auckland.ac.nz}}
\author[CC]{\fnms{Nicholas Karn}\ead[label=ee3]{N.E.Karn@soton.ac.uk}}
\author[DD]{\fnms{David Johnson} \ead[label=ee4]{david.johnson@spc.ox.ac.uk}}
\author[EE]{\fnms{Rukuang Huang} \ead[label=ee5]{rukuang.huang@jesus.ox.ac.uk}}
\and
\author[FF]{\fnms{Alexis Muir-Watt} \ead[label=ee6]{alexis.muirwatt@gmail.com}}
\address[AA]{Department of Statistics, The University of Oxford,
UK.\printead[presep={,\ }]{ee1}}

\address[BB]{Department of Statistics, The University of Auckland,
New Zealand \printead[presep={,\ }]{ee2}}
\address[CC]{Faculty of Arts and Humanities, University of Southampton,
UK\printead[presep={,\ }]{ee3}}
\address[DD]{St Peter's college, The University of Oxford,
UK \printead[presep={,\ }]{ee4}}
\address[EE]{Department of Psychiatry, University of Oxford,
UK \printead[presep={,\ }]{ee5}}
\address[FF]{Private researcher, London, UK \printead[presep={,\ }]{ee6}}
\end{aug}


\appendix

\section{Data registration and works by historians}\label{app:data}

\subsection{Data sources and data-collection} \label{app:data-sources}
The process behind the making of this dataset is complex. The acta were written by royal scribes, usually clergymen with legal training and political skill. Most were written at the behest of individuals or institutions who wanted to obtain rights or privileges or ensure the intervention of royal agents in their favour. The scribes handed over the documents to these individuals or institutions, and they were responsible for preserving them. There was no central governmental archive until the 1190s, so there were no copies kept at Westminster or elsewhere. 

Over time, many acta were destroyed accidentally or purposefully; those which survive tend to come from institutions such as towns or great churches which had a stable existence over centuries, and which had the capacity to protect and manage an archive. The activities of manuscript collectors has also had a significant impact on what has survived. The acta are now mostly in the possession of the major research libraries and archives, with a small number in private hands. 

Over the last century and a half, researchers have worked to identify where these documents survive, and to obtain transcripts of them. Editions and catalogues have been made from those collections \citep{davis13,johnson66,cronne68,bates98}, and, more recently, databases \citep{sharpe14}. These are widely used by researchers and scholars and are key resources for major questions in political, social and cultural history.

\subsection{Data registration} \label{app:data-registration}
Our registration aimed to ensure that the data conformed to the observation process we describe in the paper. We clip the uncertainty ranges of the lists $\tau^-_i,\tau^+_i,\ i\in I$ so that they lie within the range allowed by the bishops they list, that is, we impose \[
[\tau^-_i,\tau^+_i]\subseteq \bigcap_{j=1}^{n_i} [b_{y_j},e_{y_j}].
\] 
This can lead to lists being dropped. For example, the list with data-set id 677 (see below), which has a data-set year indicator of 1139 (with no uncertainty), contains Philip, Bishop of Bayeux. He was in post 1142-1163 so the uncertainty range would be clipped to the empty set. The list with data-set id 2627 has both Hugh and John, bishops of Lisieux, with an obvious second list appended, starting with Henry I. These issues could be easily fixed. However, we simply dropped lists of this sort as the number was small (five).

The list with data-set id 2364 is a fairly typical example of a witness list we retained. The bishops do not quite appear as a group, as Ranulf is ahead of Maurice. We extract the sub-list of bishops in positions $3,4,5,6,7,9$ and code them as numbers using the table in \fig~\ref{fig:bishop-names} yielding the registered list $Y=(23, 19, 30, 32, 5, 18)$.

 \begin{description}
\item[List id 2364] Year range [1107,1108], Witness List:
[1] Matilda II, wife of Henry I, queen of England;
[2] Anselm, abbot of Bec, archbishop of Canterbury; 
[3] Robert, Bloet, bishop of Lincoln; 
[4] Robert, de Limesey, bishop of Chester; 
[5] John, Bishop of Lisieux; 
[6] Richard, Bishop of Bayeux; 
[7] Gundulf, bishop of Rochester; 
[8] Ranulf, chancellor of Henry I; 
[9] Maurice, bishop of London; 
[10] William, de Roumare, earl of Lincoln 1140; 
[11] Robert, II, count of Meulan; 
[12] David, King of Scots; 
[13] Robert, de Ferrers, earl of Derby.
 
 \item[List id 2627] Year range [1130,1135], List:
 [1] Odo, Stigandus;  
 [2] Osmund, Boenot;  
 [3] Serlo, de Mansione, Malgerii;  
 [4] William, de Mirebel;   
 [5] Hugh, Buscard; 
 [6] Ranulf, de Iz;   
 [7] Grento, de Vals;  
 [8] Ralph, de Vals;   
 [9] William I, king of England;
 [10] John, archbishop of Rouen;
 [11] Hugh, bishop of Lisieux;  
 [12] Michael, bishop of Avranches;  
 [13] Durand, abbot of Troarn;  
 [14] Ainard, abbot of St Mary's, Dives;    
 [15] Nicholas, abbot of St Ouen;    
 [16] Roger, de Montgomery, earl of Shrewsbury; 
 [17] Roger, de Beaumont;  
 [18] William, de Breteuil, count; 
 [19] Henry I; 
 [20] Matilda II, wife of Henry I, queen of England; 
 [21] John, Bishop of Lisieux;  
 [22] Rabel, de Tancarville, chamberlain;   
 [23] Thurstan, Archbishop of York; 
 [24] Robert, earl of Gloucester;  
 [25] William, de Warenne II, earl of Surrey d. 1138; 
 [26] Robert, de Beaumont, earl of Leicester;  
 [27] Payn, Peverel.
 
 \item[List id 677] Year range [1139,1139], Witness List:
 [1]  Robert, Losinga, bishop of Hereford;  
 [2]  Philip, Bishop of Bayeux; 
 [3]  Roger, archdeacon of Fecamp, temp. Stephen; 
 [4]  Walter, archdeacon of Oxford;  
 [5]  Waleran, count of Meulan;  
 [6]  Ingelram de Say, temp. Stephen; 
 [7]  Walter, of Salisbury, temp. Stephen;  
 [8]  Robert, de Vere;  
 [9]  William, de Pont de, l'Arche.  
 \end{description}


\subsection{Dioceses of interest} \label{app:data-dioceses} The data display bishops with 31 distinct diocese names: Lincoln, Durham, Chester, Sherborne, 
  Winchester, Chichester, Bayeux, Lisieux,  
  Evreux, Sees,  Avranches,  Coutances, 
  Exeter, London, Rochester,  Worcester, 
  Salisbury,  Bath,  Thetford, Wells, 
  Le Mans, Hereford, Bangor, 
  Ely, St Davids,  Norwich,  
  Carlisle, St Asaph, Tusculum, 
  the Orkneys and  Llandaff.
We excluded the Bishop of Tusculum who appeared at the end of two long bishop-lists on the grounds that they they were not part of the group of interest (they would have been listed according to other criteria and in fact seem to have been regarded that way by the clerks as they were placed outside the group of bishops).
We removed bishops who only appeared in a single list.
Depending on the period $[t_1,t_2]$ of interest, this could reduce an included list from length two to one and removing that list could reduce the number of lists a bishop appeared in. We repeated this thinning to convergence.  For the period $[1080,1155]$ this removed $3$ lists of lengths 2, 2 and 3 so the number of lists went down from $374$ to $371$. In addition, six lists lost one bishop and one lost 3. It reduced the number of bishops from $81$ to $67$. The bishops of six dioceses, Sherborne, Wells, Le Mans, St Asaph, Tusculum and The Orkneys witnessed no lists in the final data set.

\subsection{Discussion of data registration} \label{app:data-registration-discuss}
Removing bishops from lists in this way was done to avoid carrying out inference for a subset of latent $U$-process parameters (see \sec~\ref{sec:prior-main}) which the data don't inform as we would be integrating out near-prior marginals using MCMC. If we include them, the $U$-dimension is $\dim(U)=K\times 1408$ (with $K=11$ in \sec~\ref{sec:poHB2aRS6} and $K=22$ in \app~\ref{app:K-equals-18-posterior}). Removing them reduces the latent dimension to $\dim(U)=K\times 1223$, a 13\% reduction in the number of continuous variables.
The prior is marginally consistent, so in this respect the analysis is unchanged. We loose $3$ of $374$ lists and $7$ lists are shortened, four loosing their last element.

Removing bishops from longer lists could remove information informing the parameters of the bishops remaining (in all the models we considered, including the poset and Plackett-Luce models). If we start with a list $Y=(j_1,...,j_n)$ with $n$ large then this supports the edge $\e{j_1}{j_n}$ relatively strongly. If the list is reduced to a list $(j_1,j_n)$ 
then the edge $\e{j_1}{j_n}$ is more weakly supported. If we have two lists $(j_1,j_2)$ and $(j_2,j_3)$ then we have support for $\e{j_1}{j_3}$. If we remove bishop $j_2$ then we loose this information.
These examples suggest that dropping bishops makes our estimated posets less certain, and contain fewer order relations.
However, dropping bishops can remove evidence against orders. If we have three lists $(j_1,j_2)$, $(j_2,j_3)$ and $(j_3,j_1)$ then there is evidence for no order between $j_1$ and $j_3$. If we remove $j_2$ then we are left with evidence for the order $\e{j_3}{j_1}$. 

In our setting 98\% of the lists were unchanged so the effect of thinning will be slight. In an earlier version of this paper we made a more substantial thinning of the data, dropping eleven dioceses. This reduced the the number of bishops to $59$ and shortened 10\% of the lists. Results were very similar to those here, down to small details (such as $\beta_{11}$ being ``out of order'' in \fig~\ref{fig:beta-poHB1aRS6} and the authority curves in \fig~\ref{fig:Uhat-poHB2aRS6} for matching dioceses being indistinguishable). Over the 76 years 1080-1155 we estimate around 13,000 linearly independent posterior probabilities for relations between bishops present in both analyses. The median absolute difference in estimated posterior probabilities $|\hat \xi^{(1,t)}_{\e {j_1}{j_2}}-\hat \xi^{(2,t)}_{\e {j_1}{j_2}}|$ in \eq~\ref{eq:post-mean-reln-probs} for relations between bishops in the old and new analyses (``1'' and ``2'') was approximately 3\%.

\begin{figure}
    \centering
    \hspace*{-0.3in}\includegraphics[width=7in, trim=1in 0.5in 1.3in 0.2in,clip]{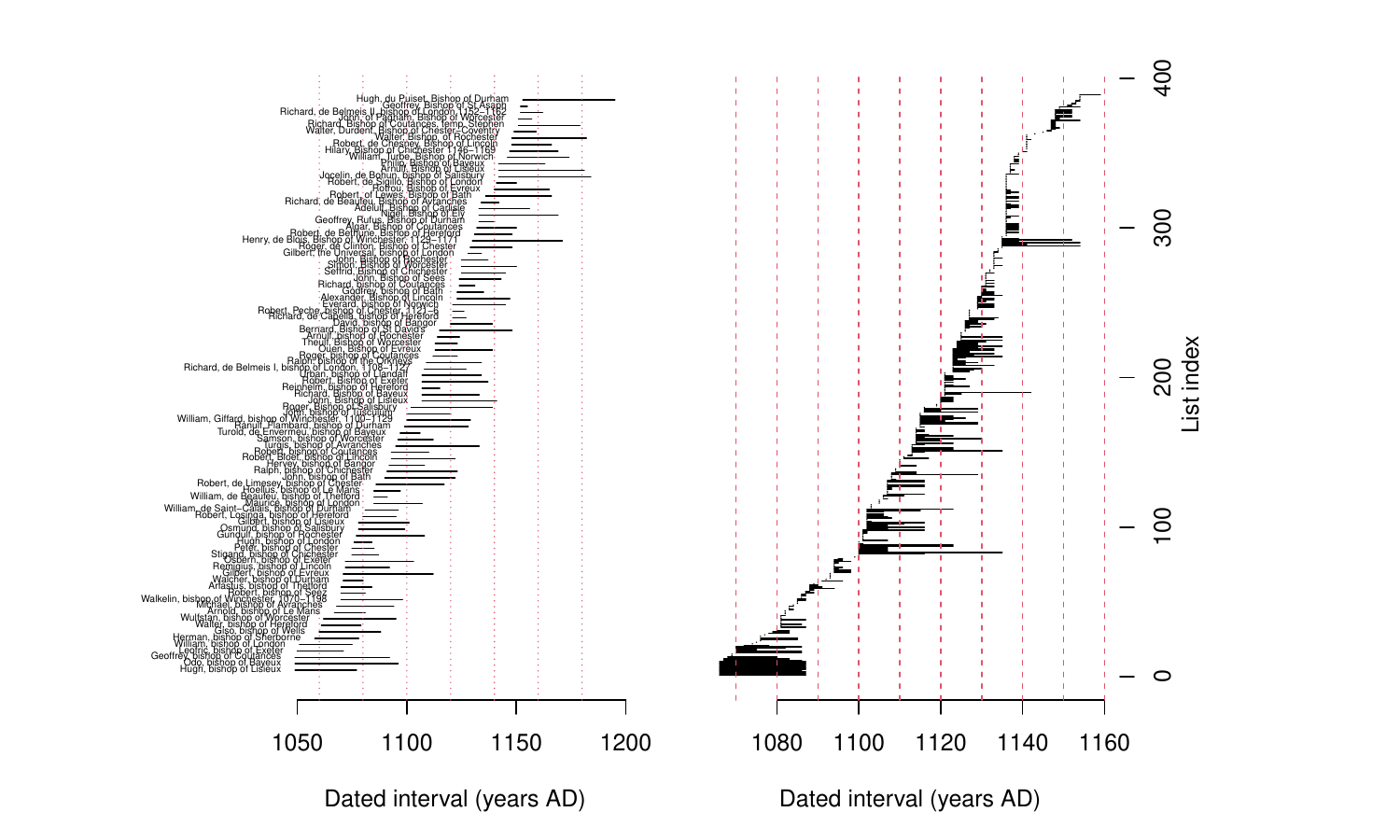} 
    \caption{Bishop (left) and list (right) date intervals. Each horizontal bar above left spans the interval $[b_j,e_j]$ over which the indicated bishop $j\in\M$ held their appointment. Each bar above right spans the uncertainty interval $[\tau^-_i,\tau^+_i]$ for list $i\in \I$.}
    \label{fig:lists-and-bishops}
\end{figure}


\subsection{Works by historians analysing witness lists}\label{app:lit-rev-history}

Historians have used the evidence of witness lists in a number of ways: (a) to determine the identity of those who were close to the grantors of documents (on the assumption that witnesses were likely present with the person in whose name a document was issued);  (b) to trace the career of individual witnesses (since if a witness appears after a certain date with a certain title, but appears without that title previously, it is reasonable to assume that that title was acquired before the date of the document in which it is first mentioned, and after those earlier ones in which it is not); and (c) as the basis for a statistical analysis of the witnesses mentioned in the lists.  Statistical techniques have, however, scarcely been used at all.

The first modern statistical analysis of royal witness lists was published by \cite{russell37}.   Russell’s aims resembled those of the current research since he attempted to use the evidence of witness lists from royal charters issued by King John (king of England 1199  – 1216) to determine the relative status of individuals mentioned in those lists. This started a debate. In 1938, a further article was published \citep{haskins38} which questioned Russell’s statistical methodology and thus, in turn, suggested his conclusions were unsound. Haskins writes ``the editor of the charter rolls for the reign of John, saw fit to deny explicitly the
suggestion that there was any order among the witnesses which might have been
dictated by considerations of rank or precedence'' and ``the precedence of witnesses to private grants is too erratic to serve as an index to their respective station''. We revisited \cite{russell37} and find the statistical reasoning (a permutation test) is incomplete rather than incorrect (a statistic measuring the dispersion of random permutations is compared to its value on data; they differ as anticipated but the variation isn't quantified). However, the point being made in the single paragraph which contains numbers, that the lists are too similar to be generated uniformly at random, is so obvious it is barely worth justifying. 
We agree that the lists are a noisy indicator of precedence. However, Haskins did not have available the tools of modern statistical inference for quantifying uncertainty in the structure behind the noise.

Witness lists did continue to be the subject of analysis, although this has been heavily biased towards attempts to discover the names of those most prominent in the royal household or administration.  A seminal article \citep{given-wilson91} on royal acta from the reigns of King Edward III (king of England 1327 –1377) and King Richard II (king of England 1377 – 1399) provided some exception to this.   Given-Wilson provided a quantitative as well as qualitative analysis; but, although he noted the general fact that individuals were grouped together into groups related to their status and societal role (bishops, earls and so on), Given-Wilson’s concern was not with relative status but rather with the frequency with which individuals attested royal charters.  He thus produced a series of tables in which individuals were grouped by category and in which the frequency of their appearance was measured. 

Historians of the period before the Norman Conquest in 1066, and of the years which immediately followed it, have paid closer attention to the ways in which witness lists can be used to illustrate questions of status.  However, as with historians of the later period, they have not applied statistical techniques to the question.  Amongst those who have considered the issue, \cite{keynes80} noted the fact that whilst again documents from the reign of King Æthelred II (king of England 978 – 1016) demonstrate the fact that particular groups of individuals always appeared together, these groups did not always appear in the same order;  \cite{baxter07} reviewed the relative positions amongst the earls in a similar period, and noted that the comparative positions of one earl in relation to another was subject to change as a reflection of political instability.   In a discussion more directly relevant to the current article, \cite{licence20} discussed the place of bishops in witness lists produced under King Edward the Confessor (king of England 1042 – 1066).   Licence detected a role for bishops in the production of documents for beneficiaries in their diocese, and suggested that generally speaking their involvement influenced their place in the list with the local bishop ‘concluding the senior group’ of witnesses; where the diocesan was not responsible for the drafting, Licence suggested that the man who was might occupy this position. 

In a review of royal acta produced following the Conquest, \cite{bates97} stressed the impact of a document’s diplomatic and whether its beneficiary was sited in Normandy or England.   Again, Bates observed the probability that witnesses were not always physically present when certain types of document were drawn up (but, it is worth noting that this need not have influenced the order in which witnesses were inscribed).  Bates further noted that a ‘more probing analysis – impossible to undertake here – would examine closely the pattern of attestations’.   It is the purpose of this paper, albeit for a slightly later period, to begin that analysis.  In particular, to examine the ways in which the witness lists of royal acta reflect the status and position of one particular group – the bishops – in the period under discussion.  At the Council of London of 1075, it was ruled that English bishops should enjoy the following order of precedence: Canterbury, York, London and then Winchester, followed by other bishops in the order of their appointment.  We seek here to examine whether or not the witness lists reflect any such order.  Later in the middle ages, the order of precedence was further extended to the bishops of other sees.  We thus also seek to show whether or not at this earlier date there were any indications that such an extended order was already being considered.

\newpage\section{Properties of the observation model}
\label{app:lkd-properties}

\subsection{Context dependence}\label{app:lkd-context-dependence}

Our queue-based idealisation of the witness process in \sec~\ref{sec:noise-free-lkd} led us to a context dependent model for random orders. For example, the poset $H=\{\e{2}{3}\}$ on $[M]=\{1,2,3\}$ has three linear extensions $\L[H]=\{(1,2,3), (2,1,3), (2,3,1)\}$ so if all the actors are present ($O=[M]$) then the probability actor $1$ comes before $2$ in a random list is one third. However, if only $1$ and $2$ are present ($O=\{1,2\}$) then $H[O]=\emptyset$ and $\L[\emptyset]=\{(1,2), (2,1)\}$ so $1$ beats $2$ with probability one half. This is called context dependence as the presence of actor $3$ changes the probability for $1$ come ahead of $2$. Plackett-Luce models, which satisfy the Choice Axiom of Luce \citep{luce77}, are context independent.

\subsection{Counting Linear extensions}
\label{app:counting-linear-extensions}
In this \app, which picks up from \sec~\ref{sec:noise-free-lkd}, we outline how we count linear extensions of a given poset $H\in \H_{[m]}$.

Recall from \sec~\ref{sec:noise-free-lkd} that $C(H)$ is the number of linear extensions of poset $H$ and $C_j(H)$ is the number that start with $j$. If $\T(H)=\{j\in [m]: C_j(H)>0\}$ is the set of ``top elements'' of $H$ then
\[ 
C(H)=\sum_{j\in \T(H)} C_j(H).
\]
This is computed using a suborder recursion.
If $j\in \T(H)$ and we let
\begin{equation}\label{eq:remove-one-bishop-suborder}
  H_{-j}=H[[m]\setminus \{j\}]    
\end{equation} 
give the poset $H$ with vertex $j$ removed then $C_j(H)=C(H_{-j})$, so we can prune away the vertices one at a time, starting at the top of the poset.
The count is available in closed form (or rapidly) for some subclasses of posets including the Bucket Orders and vertex-series-parallel orders (see \sec~\ref{sec:vsp-bucket-all}). For example, the empty poset, with $H=\emptyset$, has $C(H)=m!$. We implemented a simple R program using a recursion which terminates at these ``easy'' cases. We have also faster code based on \cite{koivisto19} which we used in later work. Links to these are shared in our online code.

\subsection{Proof of Proposition~\ref{prop:noise-free-special-case}}
\label{app:noise-free-special-case} The ``queue jumping'' error model reduces to the uniform distribution on linear extensions of the poset when the probability to jump up the queue (the noise probability $p$) is equal zero.

\noindent {\it Proposition~\ref{prop:noise-free-special-case}.}
{\it
   The noisy free and noisy models (\eqs~\ref{eq:lkd-noise-free}, \ref{eq:lkd-noisy-down} and \ref{eq:lkd-noisy-up}) all coincide at $p=0$. 
}
\begin{proof} Take $H\in\H_{[m]}$ and $Y\in\L[H]$ so that $C(H(Y_{j:m}))>0$ for each $j=1,\dots,m$ (a poset has at least one linear extension by \cite{dushnik1941}). Substituting $p=0$ in \eq~\ref{eq:lkd-noisy-down} and expanding,
\begin{align*}
  p_{(D)}(Y|H,p=0)&=\frac{C_{Y_1}(H[Y_{1:m}])}{C(H[Y_{1:m}])}\times \frac{C_{Y_2}(H[Y_{2:m}])}{C(H[Y_{2:m}])}\times\ldots\times \frac{C_{Y_{m-1}}(H[Y_{m-1:m}])}{C(H[Y_{m-1:m}])}\\
  &=C(H[Y_{1:m}])^{-1}
\end{align*}
which is $P(Y|H)$ in \eq~\ref{eq:lkd-noise-free}. The second line follows because the last numerator factor is equal one, $C_{Y_{m-1}}(H[Y_{m-1:m}])=1$ (the poset $H[Y_{m}]$ we get when we remove $Y_{m-1}$ is a poset with just a single node) and $C_{Y_j}(H[Y_{j:m}])=C(H[Y_{j+1:m}])\ne 0$ (the number of linear extensions of $H[Y_{j:m}]$ that start with $Y_j$ appearing in the numerator is equal to the number of linear extensions of $H[Y_{j+1:m}]$, the same poset but with $Y_j$ removed, and this is in the denominator of the next factor) and so the product of counts is telescoping. 

When $Y\notin\L[H]$ then there is $j\in [m]$ such that $C_{Y_j}(H[Y_{j:m}])=0$ (one of the actors was chosen out of order, so it didn't head any linear extension of $H[Y_{j:m}]$).
In this case we must have $C_{Y_j}(H[Y_{j:m}])\ne C(H[Y_{j+1:m}])$ because $C(H[Y_{j+1:m}])>0$ always. Since at least one numerator factor is zero, and all the denominator factors are positive, $p_{(D)}(Y|H,p=0)=0$ in this case.
It follows that $p_{(D)}(Y|H,p=0)=p(Y|H)$ for all $Y\in\P_{[m]}$. 

The proof that $p_{(U)}(Y|H,p=0)=p(Y|H)$ is similar.
\end{proof}

\section{Properties of priors}
\subsection{Proof of Proposition~\ref{prop:marginal-consistent}}\label{app:marginal-consistency}
When $\beta=0_{\Sb}$ we have $Z^{(t)}_j=U^{(t)}_j=Z(U^{(t)}_j,0_{\Sb};s)$ for each $t\in [B,E]$ and each $j\in \M_t$.
Let 
\[
\pi_{\H^{(B,E)}}(h|\rho,\theta,\beta=0_{\Sb})=\int \mathbb{I}_{h=h(Z(U,0_{\Sb};s))}\pi(U|\rho,\theta)\, dU
\]
give the marginal distribution of $h$ given $\rho,\theta$ and $\beta=0_{\Sb}$. The marginal distribution of $h$ is
\begin{equation}\label{eq:h-marginal-discrete-prior-beta0}
      \pi_{\H^{(B,E)}}(h|\beta=0_{\Sb})=\int_{[0,1]^2} \pi_{\H^{(B,E)}}(h|\rho,\theta,\beta=0_{\Sb})\pi(\rho,\theta)\, d\rho d\theta.
\end{equation}

We can remove $j\in\M$ ``from the beginning'' or from the realised poset ``at the end'' and get the same random poset. In the former, the random poset at time $t$ is $h(Z^{(t)}_{-j})$ where $Z^{(t)}_{-j}$ 
is an $m_t-1 \times K$ matrix with one independent row $Z^{(t)}_{j'}$ for each $j'\in \M_t\setminus \{j\}$. 
Let $h(Z_{-j})=(h(Z^{(t)}_{-j}))_{t=B}^E$ be a time series of partial orders realised in this way.
In the latter, we take a suborder at the end. Let $\M^{-j}_t=\M_t\setminus\{j\}$ and let $h^{(t)}_{-j}=h(Z^{(t)})[\M^{-j}_t]$ be the suborder at time $t$ when we remove $j$ from the order realised on the full $Z$-matrix. Let $h_{-j}=(h^{(t)}_{-j})_{t=B}^E$  give the time series. Let 
\[
\H^{(B,E)}_{-j}=\H_{\M^{-j}_B}\times \H_{\M^{-j}_{B+1}}\times ... \times \H_{\M^{-j}_E},
\]
so that $h(Z_{-j})$ and $h_{-j}$ are both in $\H^{(B,E)}_{-j}$.
We begin by stating the simple property really underpinning marginal consistency.

\begin{proposition}\label{prop:marginal-consistent-preliminary}
The random poset processes obtained by mapping the full $Z$-process to a poset process and removing $j\in\M$ from each poset in which it appears (``at the end''), or removing $j$ from the $Z$-process (``from the beginning'') are equal, so $h(Z^{(t)}_{-j})=h(Z^{(t)})[\M^{-j}_t\}]$ for $t\in[B,E]$.
\end{proposition}
\begin{proof}
This follows from the fact that the relations $\e{j_1}{j_2}$ in $h$ determined by \eq~\ref{eq:z-to-h-mapping} are not affected by the presence or absence of a third element $j$.
\end{proof}

\noindent {\it Proposition~\ref{prop:marginal-consistent}}.
{\it
Let $g\in \H^{(B,E)}_{-j}$ be given. The prior in 
\eq~\ref{eq:h-marginal-discrete-prior-beta0} is marginally consistent,
\[
\pi_{\H^{(B,E)}_{-j}}(g|\beta=0_{\Sb})=\sum_{h\in \H^{(B,E)}} \mathbb{I}_{g=h_{-j}} \pi_{\H^{(B,E)}}(h|\beta=0_{\Sb}).
\]
for each $j\in \M$.
}
\begin{proof}
Marginal consistency follows from Proposition~\ref{prop:marginal-consistent-preliminary} using a Chinese Restaurant Process style construction.
Let $\M=\{j_1,...,j_M\}$ (labelling in any order). Add each actors (``customer'') $j_1,...,j_M$ one at a time, using the fact that the $Z_j$-processes are jointly independent for $j\in \M$. 

\begin{enumerate}
    \item Initialise:
\begin{enumerate}
    \item Fix $K\ge 1$, simulate $(\rho,\theta)\sim\pi(\rho,\theta)$ and set $\beta=0_{\Sb}$.
    \item For $t\in[B,E]$ set $\M_{t,1}=\{j_1\}\cap\M_t$, $h^{(t)}_{(1)}=\emptyset$ and simulate $Z_{j_1}\sim \mbox{VAR}^{(b_{j_1},e_{j_1})}_{K,\rho,\theta}(1)$ (recall $\beta=0_{\Sb}$). 
\end{enumerate}
    \item For $i=2,...,M$ do
    \begin{enumerate}
        \item Simulate $Z_{j_i}\sim \mbox{VAR}^{(b_{j_i},e_{j_i})}_{K,\rho,\theta}(1)$. 
        \item For $t\in[B,E]$ set $\M_{t,i}=\{j_1,...,j_i\}\cap\M_t$ and construct $h^{(t)}_{(i)}$ by adding to $h^{(t)}_{(i-1)}$ the order relations between $j_i$ and $j_1,...,j_{i-1}$, as follows:
        For $j\in \{j_1,...,j_{i-1}\}$ the order relation $\e{j_i}{j}$ (resp. $\e{j}{j_i}$) is added if $Z^{(t)}_{j_i,k}>Z^{(t)}_{j,k}$ (resp. $Z^{(t)}_{j,k}>Z^{(t)}_{j_i,k}$) for each $k=1,...,K$.
    \end{enumerate}
    \item Set $h^{(t)}=h^{(t)}_{(M)}$ for each $t\in [B,E]$ and $h=(h^{(t)})_{t=B}^E$.
\end{enumerate}
The final partial-order time series $h\sim \pi_{\H^{(B,E)}}(\cdot|\beta=0_{\Sb})$ does not depend on the order in which actors $j_1,...,j_M$ arrive so we may make $j_M=j$ the last arrival. Also, by Proposition~\ref{prop:marginal-consistent-preliminary}, $g=h^{(t)}_{-j}=h^{(t)}_{(M-1)}$ is the poset state of the process for each $t\in [B,E]$ before $j_M=j$ arrives. By construction,
$g\in \H^{(B,E)}_{-j}$ and $g\sim \pi_{\H^{(B,E)}_{-j}}(\cdot|\beta=0_{\Sb})$. If we stop the process at step $i=M-1$ then the poset we obtain must be the marginal over all continuations of the process to the next step, giving Proposition~\ref{prop:marginal-consistent}.
\end{proof}

Every poset $h^{(t)}_{(i)},\ i=1,...,M-1$ is a suborder of the posets which come after it in the actor sequence, and adding or removing an actor does not change the order relations between other actors who have already arrived. This may be surprising, as we might expect new relations to add a cascade of relations by transitive closure. However,
it is straightforward, as the relation between $Z_{j_1}$ and $Z_{j_2}$ in \eq~\ref{eq:z-to-h-mapping} isn't affected by the presence or absence of $Z_{j_3}$.

We removed covariates in the analysis above. In some settings, covariate levels $s_{t,j},\ j\in \M_t$ characterise relations {\it between} actors active at time $t$ and would change if we add or remove a actor. However, if covariate levels $s_{t,j}$ are intrinsic then they don't change when we add and remove actors. 
In that case $Z^{(t)}_{j_1}=U^{(t)}_{j_1}+1_K\beta_{s_{t,j_1}}$ doesn't change when we add or remove actor $j_2$ and so relations between $j_1$ and a third actor $j_3$ don't change. It is then straightforward to show that $\pi_{\H^{(B,E)}}(h|\beta)$ is consistent for every $\beta$ and so $\pi_{\H^{(B,E)}}(h)$ is also.


\subsection{Proof of Proposition~\ref{prop:basis-Z-for-h}}\label{app:basis-Z-for-h} In this section we prove the universal expression property claimed in \sec~\ref{sec:prior-properties-Krho}.\\

\noindent Proposition~\ref{prop:basis-Z-for-h}.
{\it Suppose $\min_{t\in [B,E]} m_t \ge 4$. The probability $\pi_{\H^{(B,E)}}(h)$ in \eq~\ref{eq:h-marginal-discrete-prior}, given by the generative model \eq~\ref{eq:prior-for-h} with $K\ge \lfloor \Sm/2\rfloor$, assigns a positive probability mass $\pi_{\H^{(B,E)}}(h)>0$ to every time-series $h\in\H^{(B,E)}$.}\\

Before giving a proof we note that the condition $\min_{t\in [B,E]} m_t \ge 4$ requires the number of active actors to always exceed four. This condition is needed for Hiraguchi's Theorem to hold. See \cite{Hiraguchi51} for the original result and \cite{Bogart73a} for an accessible proof and discussion. In our application in \sec~\ref{sec:data} $m_t\ge 11$ so the condition is satisfied.

\begin{proof}
We first show that {\it (1) every poset in the time series $h$ has dimension at most $K$.} A total order is a poset in which every pair of elements is ordered. The intersection of $K$ total orders $\ell=(\ell^{(1)},...,\ell^{(K)})$ on the ground set $[m]$ is a poset $H=\cap_{k=1}^K \ell^{(k)}$ on $m$ elements satisfying $H\in \H_{[m]}$. 
A given poset may be decomposed into total orders in many different ways (and always at least one, \cite{dushnik1941}). For example, the poset $H$ in \fig~\ref{fig:po-example} has three linear extensions, $Y^{(1)}=(1,2,3,4,5), Y^{(2)}=(1,2,4,3,5)$ and $Y^{(3)}=(1,4,2,3,5)$. We can turn these into total orders in the obvious way, setting $\ell^{(i)}=\{\e{Y^{(i)}_{j}}{Y^{(i)}_{j'}}\in Y^{(i)}\times Y^{(i)}: j<j'\},\ i=1,2,3$. Then $H$ is both $H=\ell^{(1)}\cap\ell^{(3)}$ and $H=\ell^{(1)}\cap\ell^{(2)}\cap\ell^{(3)}$. 

The dimension of a poset $\dim(H)$ is the smallest number of total orders that intersect to give $H$. It is known \citep{Hiraguchi51} that if $H\in\H_{[m]}$ then $\dim(H) \le \lfloor m/2\rfloor$. Restoring time to the notation, since $D=\max_t m_t$, we can take any $K\ge \lfloor D/2\rfloor$ and be sure that $\dim(h^{(t)})\le K$ for every $t\in [B,E]$ and all $h\in \H^{(B,E)}$, where $\H^{(B,E)}$ (defined in \eq{eq:h-sequence-space-defn}) is the set of all possible poset time series $h=(h^{(B)},h^{(B+1)},\ldots,h^{(E)})$ on the sequence of ground sets $\M_B,\M_{B+1},\ldots,\M_E$.

We next show that {\it (2) the mapping $h^{(t)}=h(Z^{(t)})$ takes the intersection over columns of the rank order for each column of $Z^{(t)}$.} Consider a $m_t\times K$ matrix $Z^{(t)}$ in a particular year $t\in [B,E]$. The rule mapping $Z^{(t)}$ to $h^{(t)}$ in \eq~\ref{eq:z-to-h-mapping} just takes the intersection of the rank orders of the columns of $Z^{(t)}$: for $k=1,...,K$, the rank order in column $k$ is
\[
\ell^{(t,k)}=\{\e ij\in \M_t\times\M_t: Z^{(t)}_{i,k}>Z^{(t)}_{j,k}\} 
\]
so \eq~\ref{eq:z-to-h-mapping} is the same as
\[
h(Z^{(t)})=\cap_{k=1}^K \ell^{(t,k)}.
\]
Write $\ell^{(t)}=(\ell^{(t,k)})_{k=1}^K$ for the $K$ total orders and write $\ell^{(t)}=\ell(Z^{(t)})$ to show that $\ell^{(t)}$ is determined from $Z^{(t)}$.

We next show {\it (3) for any partial-order time series $h$ there is a set of latent-variable time series $Z$ each realising $h$. The set has non-zero measure.} We realise any particular time series $h\in \H_{[A,B]}$ if we realise $Z$ such that $h(Z)=h$,
so we realise $h$ if we realise a sequence of $Z^{(t)}$-matrices with column rank orders $\ell^{(t)}=\ell(Z^{(t)})$ that intersect to give $h^{(t)}$. Since $\dim(h^{(t)})\le K$ we know that such a sequence of rank orders exists, so we get $h$ with positive prior probability if the probability to realise such a $Z$ is positive.

The prior for $Z$ is given in terms of a continuous mixture (over $\rho$ and $\theta$) of multivariate normal $U$ and $\beta$ prior densities. It has a density with respect to Lebesgue measure on the space
\[
\Z^{(B,E)}=\R^{K m_B}\times \R^{K m_{B+1}} \times ... \times \R^{K m_E}
\]
which is strictly positive for all $Z\in \Z^{(B,E)}$ (so every possible sequence of $Z$-matrices can be realised). Let $h\in\H^{(B,E)}$ be given and let $\ell^{(t)}$ be any decomposition of $h^{(t)}$ into $K$ total orders (at least one decomposition exists, and if $\dim(h^{(t)})<K$ then we just repeat total orders as this does not change their intersection). Let
\[
\Z_{t,\ell}=\{z\in \R^{K m_{t}}: \ell(z)=\ell^{(t)}\}
\]
be the set of $Z$-matrices in $\R^{K m_{t}}$ representing year $t$ that yield the desired column rank orders. 

The sets $\Z_{t,\ell}$ are not empty: we can take the $m_t$ entries in the $k$'th column $Z^{(t)}_{\M_t,k}=(Z^{(t)}_{j,k})_{j\in \M_t}$ of $Z^{(t)}$ to be any set of real numbers that match the column rank order constraint
imposed by $\ell^{(t,k)}$. It follows that the volume measure of $\Z_{t,\ell}$ is strictly greater than zero.

Any $z\in \Z_{t,\ell}$ satisfies $h(z)=h^{(t)}$. However $\ell^{(t)}$ above is just one decomposition of $h^{(t)}$ into $K$ (not necessarily distinct) total orders, so $Z_{t,\ell}$ will typically be smaller than the set of all $z$'s giving $h^{(t)}$. Let
\[
    \Z(h^{(t)})=\{z\in \R^{m_{t} K}: h(z)=h^{(t)}\}
\]
be that set, with $\Z(h^{(t)})\supseteq \Z_{t,\ell}$.
Let
\[
\Z(h)=\Z(h^{(B)})\times \Z(h^{(B+1)})\times ... \times \Z(h^{(E)}),
\]
and
\[
\Z_{\ell}=\Z_{B,\ell}\times \Z_{B+1,\ell}\times ... \times \Z_{E,\ell}
\]
so that again, $\Z(h)\supseteq \Z_{\ell}$ and the measure of $\Z_{\ell}$ is not zero. 

Finally, {\it (4) the prior probability to realise any time series $h\in \H^{(B,E)}$ is not zero.} Indeed,
\begin{align*}
    \pi_{\H^{(B,E)}}(h)&=\Pr(Z\in \Z(h))\\
    &\ge \Pr(Z\in \Z_\ell)\\
    &>0,
\end{align*}
as the density of $Z$ is not zero on the set $\Z_\ell$ and this set has non-zero measure.
\end{proof}

\newpage\section{Prior and Posterior depth distributions}\label{app:prior-simulation}

In this section we report some summaries from prior simulation of depth distributions and compare them with posterior depth distributions. We have said that we would like prior depth distributions to be reasonably uniform. In \fig~\ref{fig:prior-depth-dbns-curves-NF9-18} we plot prior depth distributions. There is one curve for each of the $T=76$ years showing the prior depth distribution for that year. The two plots show distributions for $K=11$ and $K=22$ (left and right respectively).

In \sec~\ref{sec:prior-prob-dbns} we mentioned that the uniform prior on posets $H\sim \mbox{Unif}(\H_{[m]})$ concentrates on posets of depth three as $m\to \infty$ \citep{kleitman1975asymptotic}. We illustrate that along with the depth distributions of our priors.
In \fig~\ref{fig:prior-depth-dbns-curves-NF9-18} (left and right), the dashed curve shows the prior depth distribution determined by the uniform prior. 
A random poset $H\sim \mbox{Unif}(\H_{[m]})$ simulated using the MCMC algorithm given in \citet{muirwatt15} is shown in \fig~\ref{fig:uniform-po-50}. This illustrates the rather dramtic weighting present in this prior. 

Our own prior distributions tend to rule out partial orders of depth 1 and are otherwise reasonably flat, weighting somewhat in favour of lower depth. Taking $K=22$ raises the prior probability for depth 1 slightly. This is hard to see in \fig~\ref{fig:prior-depth-dbns-curves-NF9-18} but visible in the left panels in \figs~\ref{fig:prior-post-depth-boxes-poHB2aRS6} and \ref{fig:prior-post-depth-boxes-poHB17aRS4}. \fig~\ref{fig:prior-post-depth-boxes-poHB2aRS6} shows prior and posterior depth distributions for the analysis in \sec~\ref{sec:poHB2aRS6} with $K=11$. We see that the whiskers do not extend to depth 1 in all years in the prior plot at left.  \fig~\ref{fig:prior-post-depth-boxes-poHB17aRS4} shows prior and posterior depth distributions for the analysis in \app~\ref{app:K-equals-18-posterior} which is the same as that in shown in \fig~\ref{fig:prior-post-depth-boxes-poHB2aRS6} except that $K=22$. We see that the whiskers do now extend to depth 1 in all years in the prior plot at left.  

The prior we have given has reasonably flat support over depths that the data actually supports.
We may be concerned that the depths with highest posterior probability in \figs~\ref{fig:prior-post-depth-boxes-poHB2aRS6} and \ref{fig:prior-post-depth-boxes-poHB17aRS4} coincide with relatively higher prior-probability depths in \fig~\ref{fig:prior-depth-dbns-curves-NF9-18} (so, around 2-5). However, this is not prior domination. The data easily dominates the prior. We show this in \app~\ref{app:synth-total-orders} where we take synthetic data with the same list content (so the same $o_i,\ i\in\I$ as the real data, but simulated data $y'_i,\ i\in \I$).
We take ``true'' partial orders which are complete (or near complete) orders so the true depth is around $m_t,\ t\in [B,E]$ for each year. We reconstruct this depth well, so the prior is in no way dominating the reconstructed depth. At the end of \app~\ref{app:synth-total-orders} we remark on the corresponding issue of potential prior bias at depth 1.  


{
\begin{figure}[H]
    \centering
    \begin{tabular}{cc}
    \hspace*{-0.0in}\includegraphics[width=2.6in, trim=0in 0.2in 0.4in 0.7in, clip]{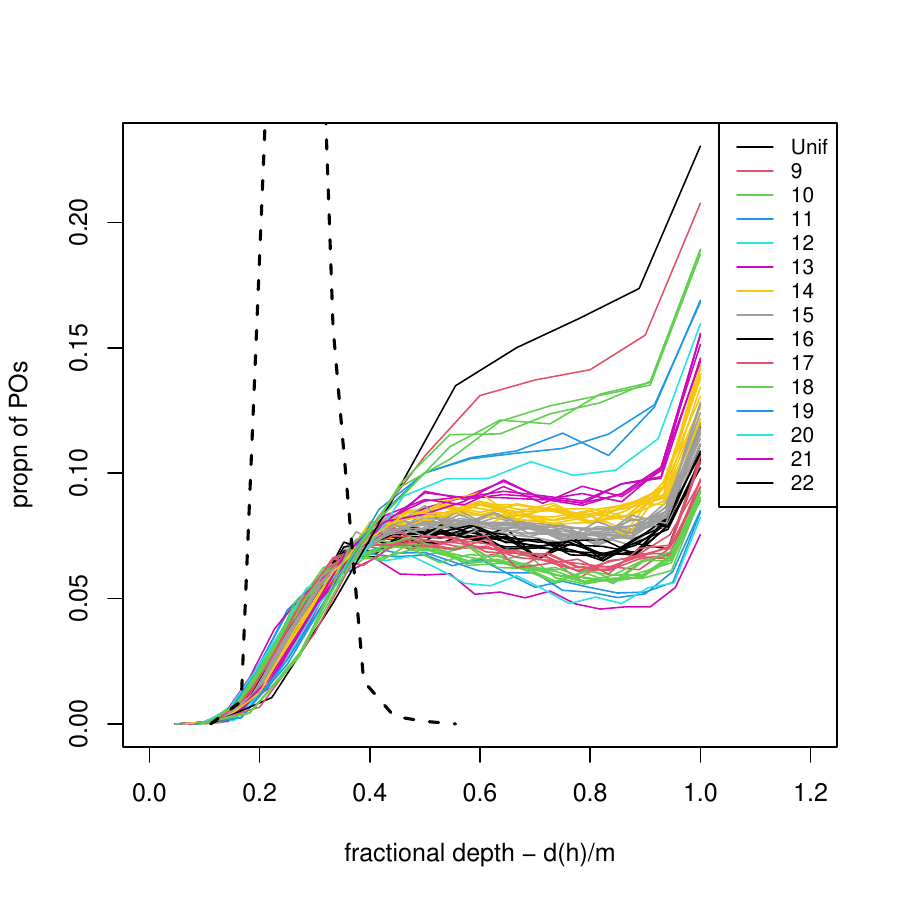}
    &
    \includegraphics[width=2.6in, trim=0.4in 0.2in 0in 0.7in, clip]{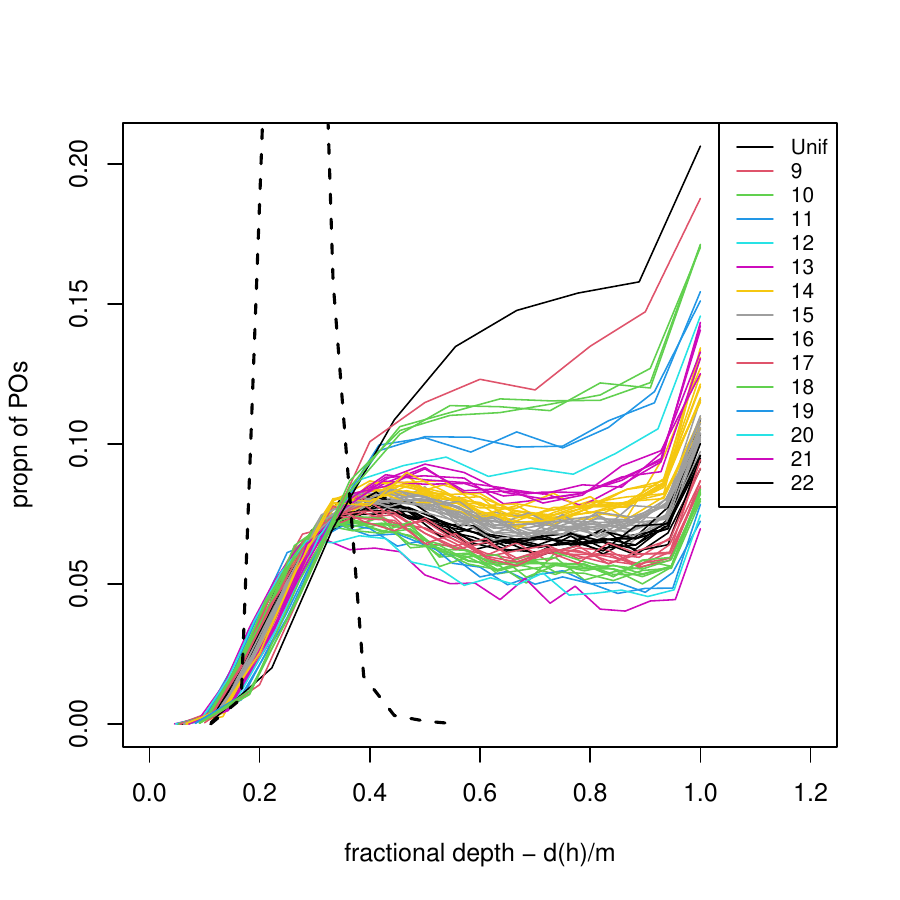}
    \end{tabular}
    \caption{Monte-Carlo estimates of the prior depth
    distribution for the poset at each time - using the prior for $h$ summarised in \sec~\ref{sec:time-series-posterior}: (Left) $K=11$ columns in the $Z$-matrix; (Right) $K=22$. This distribution varies as the number of bishops (indicated by color) in the poset varies over time. The $x$-axis shows relative depth $d(h^{(t)})/m_t$ (equal one for a total order). The dashed curve is the depth distribution for a uniform poset with $m_t=18$. }
    \label{fig:prior-depth-dbns-curves-NF9-18}
\end{figure}

\begin{figure}[H]
    \centering
    \begin{tabular}{cc}
    \hspace*{-0.7in}\includegraphics[width=3.3in, trim=0.0in 0.2in 0.4in 0.7in, clip]{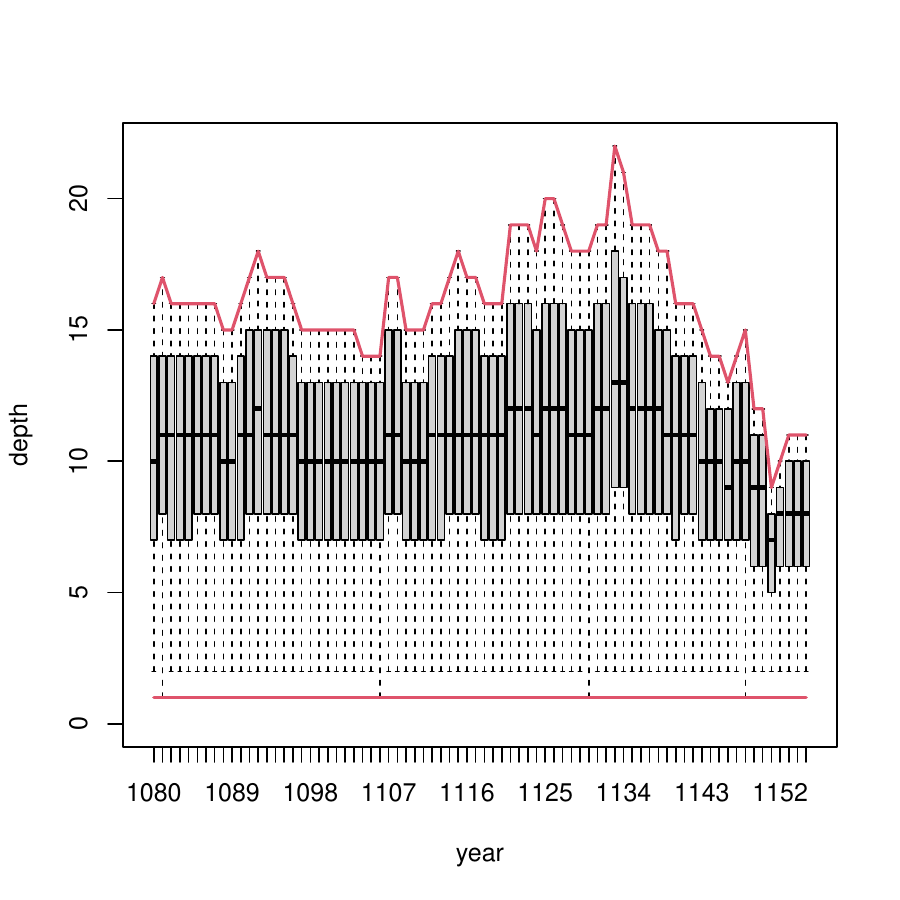}
    &
    \includegraphics[width=3.3in, trim=0.4in 0.2in 0.0in 0.7in, clip]{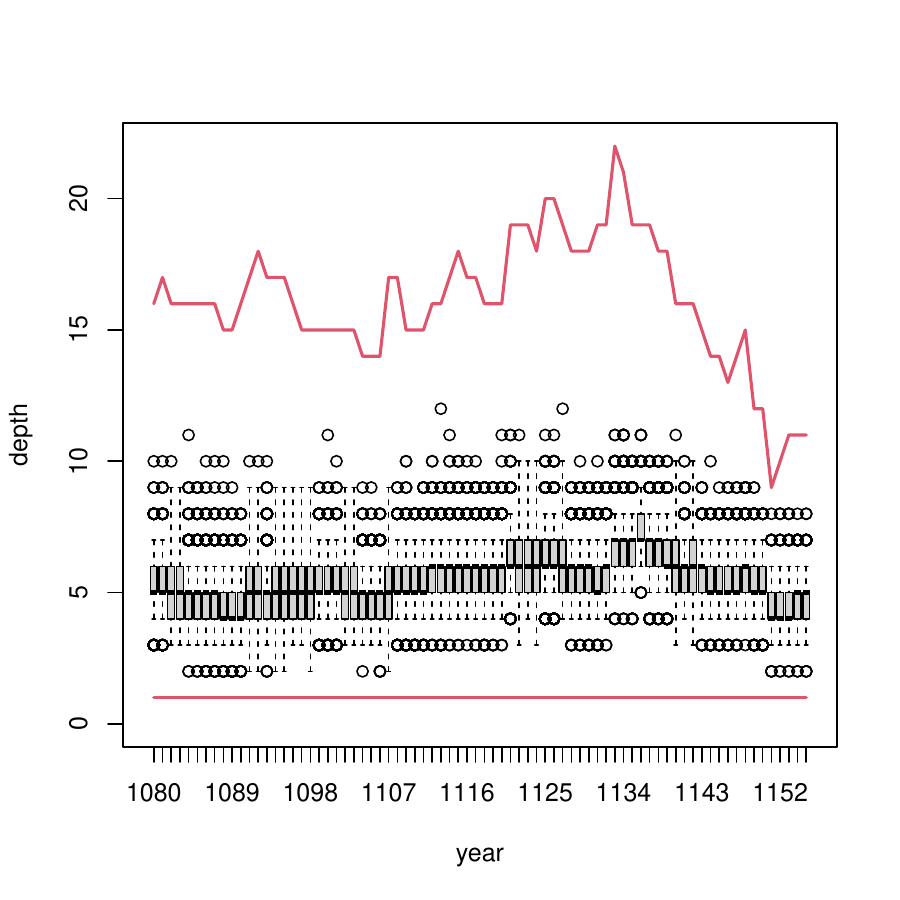}
    \end{tabular}
    \caption{(Left) Prior and (Right) posterior depth distributions from the analysis in \sec~\ref{sec:poHB2aRS6} with $K=11$. The x-axis shows years. Each box represents the posterior distribution for its year. It is computed using prior samples (at Left) and MCMC posterior sample depths (at Right). The upper and lower bounding lines indicate the least (equal one) and greatest (equal $m_t$ in year $t$) depth possible in a given year.}
    \label{fig:prior-post-depth-boxes-poHB2aRS6}
\end{figure}

\begin{figure}[H]
    \centering
    \begin{tabular}{cc}
    \hspace*{-0.7in}{\includegraphics[width=3.3in, trim=0.0in 0.2in 0.4in 0.7in, clip]{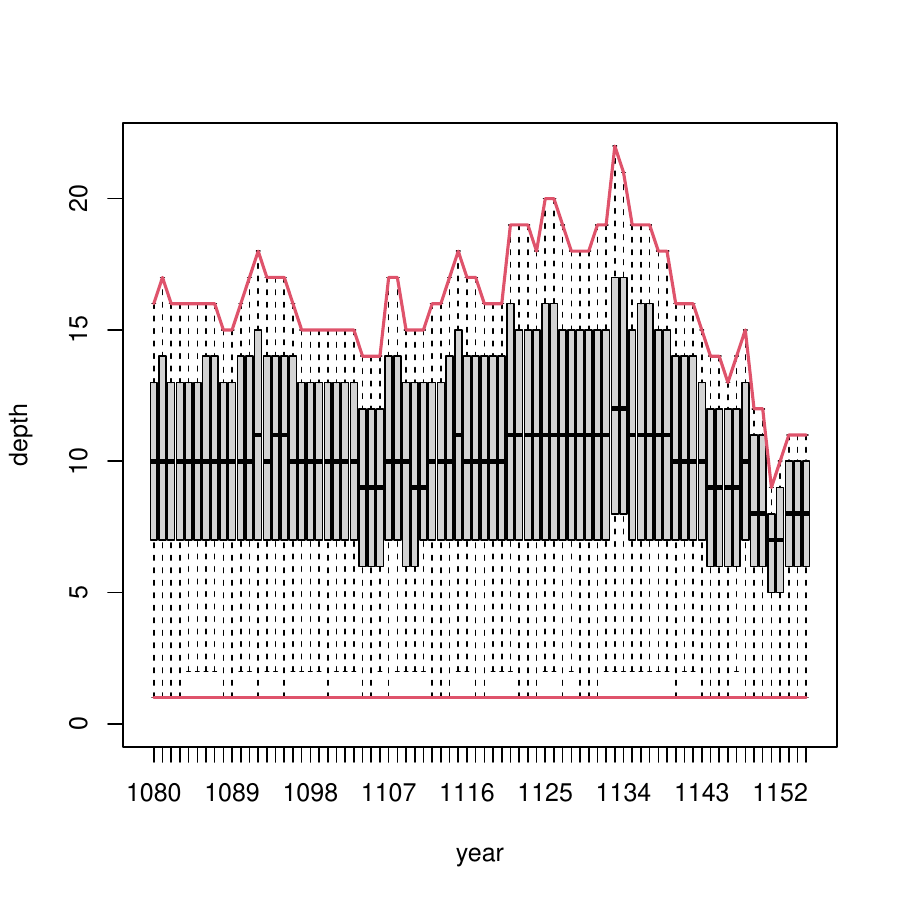}} &
    {\includegraphics[width=3.3in, trim=0.4in 0.2in 0.0in 0.7in, clip]{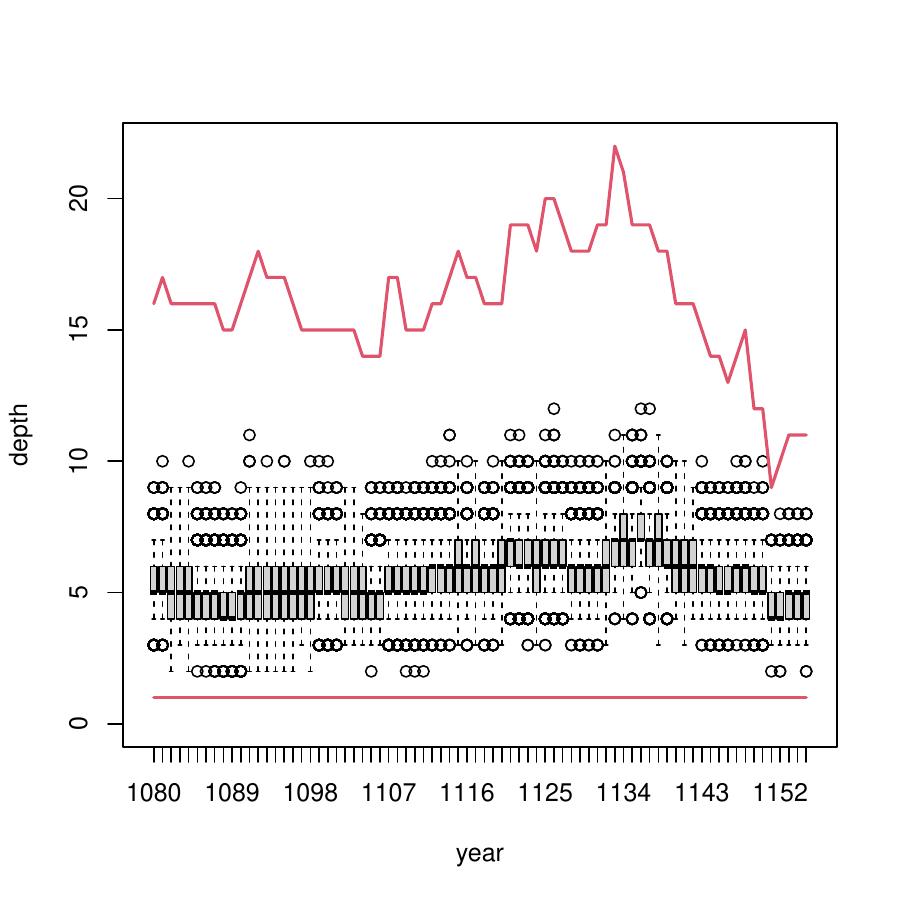}}
    \end{tabular}
    \caption{(Left) Prior and (Right) posterior depth distributions from the analysis in \app~\ref{app:further-results}, as \fig~\ref{fig:prior-post-depth-boxes-poHB2aRS6}, here $K=22$.}
    \label{fig:prior-post-depth-boxes-poHB17aRS4}
\end{figure}

\begin{figure}[H]
    \centering
    {\includegraphics[width=3in, trim=1.2in 1.1in 0.8in 0.9in, clip]{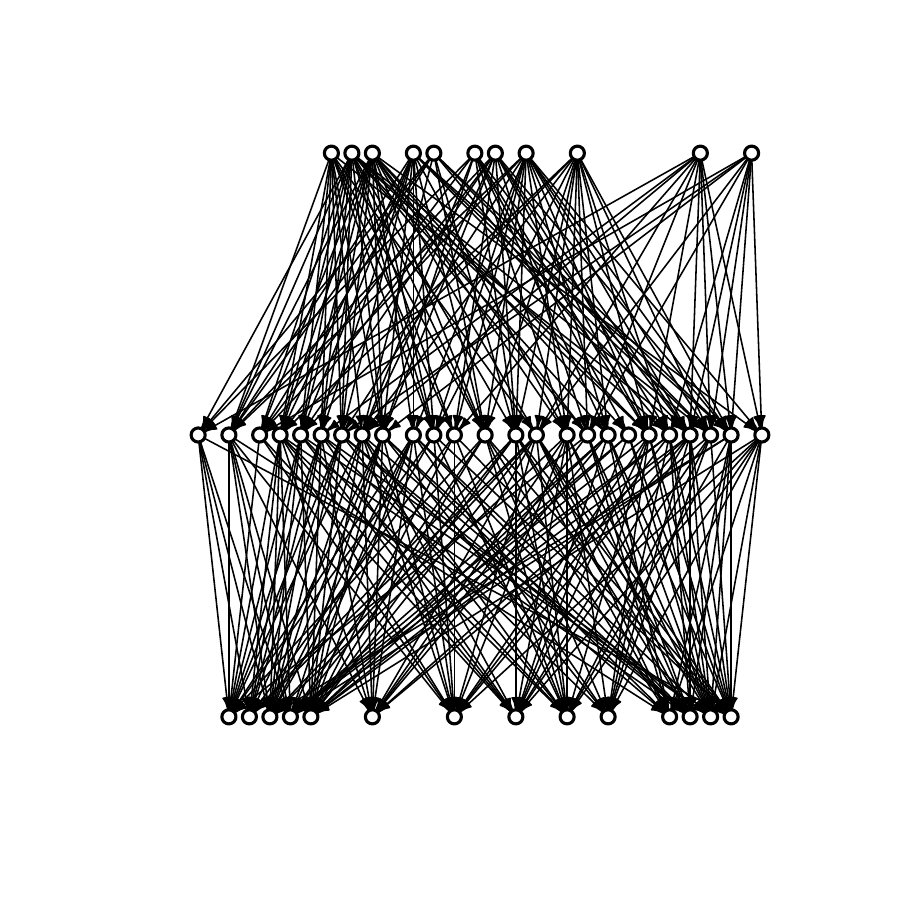}}
    \caption{A poset on 50 nodes drawn approximately uniformly at random from $\H_{[50]}$ using the MCMC algorithm given in \citet{muirwatt15}, illustrating the concentration on orders of depth three.}
    \label{fig:uniform-po-50}
    \vfill\eject
\end{figure}
}

\newpage\section{MCMC} \label{app:mcmc}
Proposals for $\rho$ and $\theta$ (which are in $[0,1]$) are prior draws with a Hastings ratio which does not involve the final likelihood factor in \eq~\ref{eq:posterior}. Candidate states $\beta_r\to{\beta'}_r,\ r=1,...,\Sb$ for the components of $\beta\in \B_0$ or $\B_{\Sb}$ are simple random walk proposals with a fixed bandwidth (rejecting if the order condition in $\B_{\Sb}$ is violated). 

Proposals for the $\mbox{VAR}^{(b_{j},e_{j})}_{K,\rho,\theta}(1)$ time series $U$
for each component $j\in \M_t$ target each time $t\in [b_j,e_j]$ in turn. The candidate state $\tilde U^{(t)}_j\in \R^K$ is a draw from the prior conditioned on $\rho$ and $\theta$ and the values $U^{t-1}_j,U^{t+1}_j$ of the process before and after the target vector, so that
\begin{align*}
\tilde U^{(t)}_j &\sim N\left(\frac{\theta}{1+\theta^2}(U^{(t-1)}_j+U^{(t+1)}_j) ,\frac{1-\theta^2}{1+\theta^2}\Sigma^{(\rho)}\right),\ & \mbox{if $b_j<t<e_j$},\\
    \tilde U^{(b_j)}_j &\sim N\left(\theta U^{(b_j+1)}_j,(1-\theta^2)\Sigma^{(\rho)}\right),\ &
    \mbox{if $t=b_j<e_j$},\\
    \tilde U^{(e_j)}_j &\sim N\left(\theta U^{(e_j-1)}_j,(1-\theta^2)\Sigma^{(\rho)}\right)\ &
    \mbox{if $b_j<t=e_j$},\\
    \tilde U^{(t)}_j &\sim N\left(0_K,\Sigma^{(\rho)}\right)\ &
    \mbox{if $t=b_j=e_j$}
\end{align*}
with special cases at the start and end of the active period for bishop $j$.
This determines the candidate state $\tilde U$ (which differs from $U$ at just a single vector $U^{(t)}_j\to \tilde U^{(t)}_j$). 

At $U$ and $\beta$ updates we update $\tilde Z=Z(\tilde U,\beta;s)$ and $\tilde Z=Z(U,\tilde \beta;s)$ respectively to determine the candidate poset 
$\tilde h=h(\tilde Z)$. In a $U$-update at year $t$, the poset series $h$ and $\tilde h$ are equal except for the replacement $h^{(t)}\to \tilde h^{(t)}$ so cancellations leave the Hastings ratio depending on the likelihood for lists $\I_t(\tau)=\{i\in\I: \tau_i=t\}$ only. The acceptance probability is
\[
\alpha(\tilde U^{(t)}|U^{(t)})=\min\left\{1,\ \prod_{i\in\I_t(\tau)}\frac{p(y_i|\tilde h^{(t)}[o_i],p)}{p(y_i|h^{(t)}[o_i],p)}\right\}.
\] 
Updates to components of $\beta$ affect $Z$ over all times so the full likelihood ratio is needed. 

Proposals for $\tau$ are prior draws. The Hastings ratio depends only on likelihood factors for the list $i\in\I$ for which the date $\tau_i$ is being updated. If $\tilde\tau_i$ is the candidate value then the acceptance probability is
\[
\alpha(\tilde\tau_i|\tau_i)=\min\left\{1,\ \frac{p(y_i|h^{(\tilde\tau_i)}[o_i],p)}{p(y_i|h^{(\tau_i)}[o_i],p)}\right\}
\]
Proposals for $p$ are also prior draws. This time the Hastings ratio is the full likelihood ratio, $p(y|h,\tau,\tilde p)/p(y|h,\tau,p)$.

MCMC runs were initialised at either disordered or ordered initialisations. In the disordered initialisation all parameters were given starting values as far as possible from equilibrium values, the effects $\beta=0_{\Sb}$, the times $\tau$ are prior draws subject to $\tau\in [B,E]^N$ and the
$U$-matrices specified so that the order was empty or near empty.
In the ordered initialisation parameters start at values that might be typical for equilibrium, the effects $\beta$ and the lists times $\tau$ as for disordered, and the $U$-matrices specified so that the initial posets $h^{(t,0)}$ are prior draws which dont conflict the data. This gives much deeper posets for the ordered initialisation. 
We checked that two runs started in ordered and disordered initialisations gave the same posterior densities for marginal parameter posteriors. We checked effective sample sizes and inspected MCMC traces.




\newpage\section{Posterior summary statistics}\label{app:posterior-summary-stats}
\subsection{Consensus Partial Order}\label{app:CPO}
We now define our principle posterior summaries.
We seek a point estimate of the unknown true evolving status hierarchy of posets $h$. The consensus poset with threshold $\xi\in [0,1]$ displays all the relations supported at time $t$ with estimated posterior probability 
\begin{equation}\label{eq:post-mean-reln-probs}
\hat \xi^{(t)}_{\e {j_1}{j_2}}=L^{-1}\sum_{l=1}^L \mathbb{I}_{\e ij\in h^{(t,l)}}
\end{equation}
greater than or equal $\xi$. It is a directed graph
\begin{equation}\label{eq:CPO-defn}
  \bar h^{(t)}(\xi)=(E_t(\xi),\M_t)  
\end{equation}
on the active bishops in year $t$ with edges 
\[
E_t(\xi) =\left\{\e {j_1}{j_2}\in \M_t\times \M_t:\hat \xi^{(t)}_{\e {j_1}{j_2}}\ge \xi \right\}.
\]
The consensus poset need not be acyclic, even when $\xi\ge 0.5$, so $\bar h^{(t)}(\xi)\not\in \H_t$ is possible.
This seems to be possible but improbable: all the consensus partial orders in this paper are partial orders.

\subsection{Testing for an effect due to seniority} \label{app:bayes-factor-decreasing} In \sec~\ref{sec:poHB1aRS6} we made an analysis of the full data with no constraint on the seniority effects, so $\beta\in\B_0$. Here we test for $\beta\in\B_\S$. 

We estimate a Bayes factor for the seniority effects $\beta_r,\ r=1,...,S$ to be decreasing with decreasing seniority (so increasing seniority-rank $r$). Let $\B_{S'}=\{\beta\in \B_0: \beta_1>\beta_2>...>\beta_{S'}\}$ be the event that the first $S'\le \Sb$ effects are ordered. It is feasible to compare models with $\beta\in \B_0$ against $\beta\in\B_{S'}$ for $S'$ up to about $S'=8$ using a Savage-Dickey estimator, as $\B_{S'}\subset\B_0$. We have
\begin{align}
    B_{S',0}&=\frac{p(y|\beta\in \B_{S'})}{p(y|\beta\in \B_0)}\nonumber\\
    &=\frac{\pi(\beta\in\B_{S'}|y)}{\pi(\beta\in\B_{S'})} \label{eq:beta-Bayes-factor}\\
    &=S'!\, E_{\rho,\theta,U,\beta,\tau,p|y}(\mathbb{I}_{\beta\in\B_{S'}}),\nonumber
\end{align}
as the marginals $\pi(\beta\in \B_0|y)=\pi(\beta\in \B_0)=1$ when there is no constraint, and $\pi(\beta\in \B_{S'})=1/S'!$. The expectation can be estimated using samples from the unconstrained posterior $\pi(\rho,\theta,U,\beta,\tau,p|y)$. The point here is that the prior probability for $\beta\in\B_{S'}$ is very small for $S'$ at all large, so if we see any posterior samples with $\beta\in\B_{S'}$ this shows the data is weighting strongly in favour of these configurations. However, even though we may have $\pi(\beta\in\B_{S'}|y)\gg\pi(\beta\in\B_{S'})$ the numerator in \eq~\ref{eq:beta-Bayes-factor} is very small for large $S'$ and hard to estimate accurately, so we restrict discussion to $S'\le 8$.

We carried out this test using the setup in \sec~\ref{sec:poHB1aRS6} (with $K=11$ and likelihood $p_{(U)}$, so the model which gives the marginal$\beta$-distributions in the right plot in \fig~\ref{fig:beta-poHB1aRS6}). We find $\hat B_{6,0}=4.0(1)$, $\hat B_{7,0}=10(5)$ and $\hat B_{8,0}=40(28)$ with standard errors from iid Bernoulli sampling in parenthesis (long ordered sequences are rare events in the MCMC, so effectively iid). The last of these is based on two MCMC samples with $\beta\in\B_8$. However, given the improbability of these states in the prior ($1/8!$), seeing any samples at all with this pattern is evidence for decreasing seniority effects with decreasing seniority rank. 

\newpage\section{Further results and comparison analyses}
\label{app:further-results}

\subsection{Additional figures for the analysis in \sec~\ref{sec:poHB1aRS6}}\label{app:poHB1aRS6}

\fig~\ref{fig:mcmc-poHB1aRS6} shows MCMC traces for selected parameters and functions in the MCMC for the analysis in \sec~\ref{sec:poHB1aRS6} using the MCMC algorithm of \app~\ref{app:mcmc}. This paper contains results from dozens of MCMC runs so we present this as specimen output. Another specimen set of MCMC traces are given in \fig~\ref{fig:synth-mcmc-poHB22aRS4} in \app~\ref{app:synth-poHB22aRS4} where we analyse synthetic data with the same model and structure as in this section (obtaining good recovery of truth). For the main analyses in \sec~\ref{sec:poHB1aRS6} and \ref{sec:poHB2aRS6} we ran two independent MCMC chains with ordered and unordered start states (see paragraph end of \app~\ref{app:mcmc}) and obtained excellent agreement.

\fig~\ref{fig:poHB1-14-16-rho-theta-p} shows comparisons of posterior distributions for the parameters $\rho, \theta$ and $p$ for the unconstrained seniority-effects model $\beta\in\B_0$ (so, as \sec~\ref{sec:poHB1aRS6}) across runs with different model settings. The prior for each parameter is also plotted (except $\theta$ which is uniform). We vary the likelihood noise model ($p_{U}$ and $p_{(D)}$) and the number of columns $K$ in the status matrix $Z^{(t)}$. The serial-correlation parameter $\theta$ and the queue-jumping error probability $p$ are almost unchanged when these other aspects of the model are varied. 

The $\rho$ posterior for $K=2,11$ and $22$ (left panel of \fig~\ref{fig:poHB1-14-16-rho-theta-p}) shift to larger values as $K$ increases. Why does this happen? We can think of the posets $h$ as fairly well informed by the data. In this setting the posets $h$ restrict the possible $U$-matrices as the paths in $U$ have to generate the posets in $h$. The $U$ matrices function as ``data'' for the latent parameters $\rho$ and $\theta$. When we increase $K$ the posets, informed by the data, are not changing (very much). The paths in $U$ are longer at larger $K$ but they must have the same patterns of intersection to give the same (or more or less the same) paths $U$ and hence the same posets. When paths are longer the hazard for them crossing is greater. However, the fit can be maintained by making them flatter: each path must be more strongly correlated at large $K$ than at small $K$. We conclude that there is some degree of non-identifiability between $K$ and $\rho$. 

\begin{figure}
    \centering \hspace*{-0.2in}\includegraphics[width=5.5in,trim=0.2in 0.1in 0.15in 0in,clip]{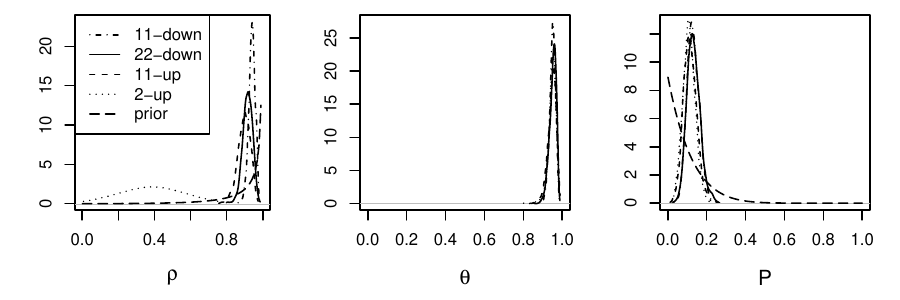}
    \caption{Posterior parameter densities for $\rho,\theta$ and $p$ from a selection of different models for the same data in the analysis of \sec~\ref{sec:poHB1aRS6}. The priors are given in \sec~\ref{sec:post-sum} with $\beta\in \B_0$ and variations: (dot-dashed) $K=11$, likelihood $p_{(D)}$ in \eq~\ref{eq:lkd-noisy-up}; (dashed) $K=11$, likelihood $p_{(U)}$ in \eq~\ref{eq:lkd-noisy-down}; (solid) $K=22$, likelihood $p_{(D)}$; (dotted) $K=2$, likelihood $p_{(D)}$; (long-dashed, in $\rho$ and $p$ graphs) prior. The prior for $\theta$ is uniform. 
    }
    \label{fig:poHB1-14-16-rho-theta-p}
\end{figure}

\begin{figure}
    \centering
    \hspace*{-0.75in}\includegraphics[width=6.5in]{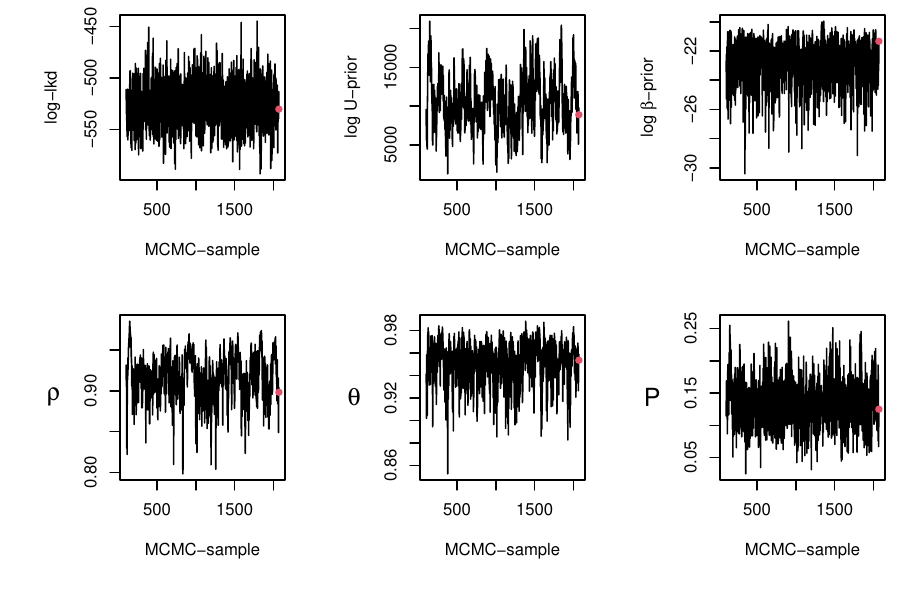}
    \caption{Selected MCMC traces from the unconstrained seniority-effect analysis in \sec~\ref{sec:poHB1aRS6}.}
    \label{fig:mcmc-poHB1aRS6}
\end{figure}

\subsection{Figures from the analysis in \sec~\ref{sec:poHB2aRS6}}\label{app:poHB2aRS6}

In this section we show some further results from the analysis presented in \sec~\ref{sec:poHB2aRS6}, the analysis with constrained seniority effect parameters $\beta\in\B_{\Sb}$ and $K=11$ columns in the status matrix $Z^{(t)}$.

In \fig~\ref{fig:CPO-poHB2aRS6-close} we show transitive closures of the consensus poset $\bar h^{(t)}(\xi)$ defined in \eq~\ref{eq:CPO-defn} for the years $t\in [1134, 1136]$. Edges with thresholds $\xi=0.5$ and $\xi=0.9$ are shown dashed and solid respectively. The reductions are shown in \fig~\ref{fig:CPO-poHB2aRS6}.
These transitively closed consensus posets are hard to read so we prefer reductions. We show them to emphasise that many well-supported relations (solid edges) are found, as a sequence of weakly supported relations (dashed edges $\e{j_1}{j_2},\e{j_2}{j_3}$) in the reduction tends to yield a well supported relation (a solid edge $\e{j_1}{j_3}$) in the closure.

\fig~\ref{fig:bishop-names} gives the mapping from names to numbers in all graphs in this paper.

\fig~\ref{fig:Zhat-poHB2aRS6} shows the posterior mean status curves $\bar Z_j^{(t)},\ t\in [b_j,e_j]$ for each bishop $j\in\M$. This is defined as for $\bar U_j^{(t)},\ t\in [b_j,e_j]$ in \eq~\ref{eq:post-mean-U-process}. This figure should be compared with \fig~\ref{fig:Uhat-poHB2aRS6}. For further discussion of this pair of figures see \sec~\ref{sec:poHB2aRS6}. 

\fig~\ref{fig:tau-poHB2aRS6} shows 90\% HPD sets for the posterior distribution of the list times $\tau_i,\ i\in \I$. The $x$-axis in this graph is the list index (sorted, so that prior mean list times increase). Recall the data gives us a time range $\tau_i\in [\tau_i^-,\tau_i^+]$ for each list. The red bars show this time interval constraint, and also represent the prior, which is uniform. The dates of some lists are known without uncertainty (so $\tau_i^-=\tau_i^+$). These lists have been omitted from \fig~\ref{fig:tau-poHB2aRS6}. It can be seen that the data do inform the list times a little (the grey set is restricted relative to the red bar) though the benefit is seen mainly in lists with substantial prior time ranges (the longer red bars). We learn little here as the prior constraints $\tau_i\in [\tau^-_i,\tau^+_i]$ are already quite tight. The principle purpose of including the uncertain list dates in the analysis is to feed through the uncertainty in list dates into the estimates of average {authority} $\bar U^{(t)}_j$ and consensus partial orders $\bar h^{(t)}(\xi)$ rather than to do list dating.

\begin{figure}[H]
    \centering
    \hspace*{-0.2in}\includegraphics[width=6.25in,trim=0in 0.2in 0in 0.0in, clip]{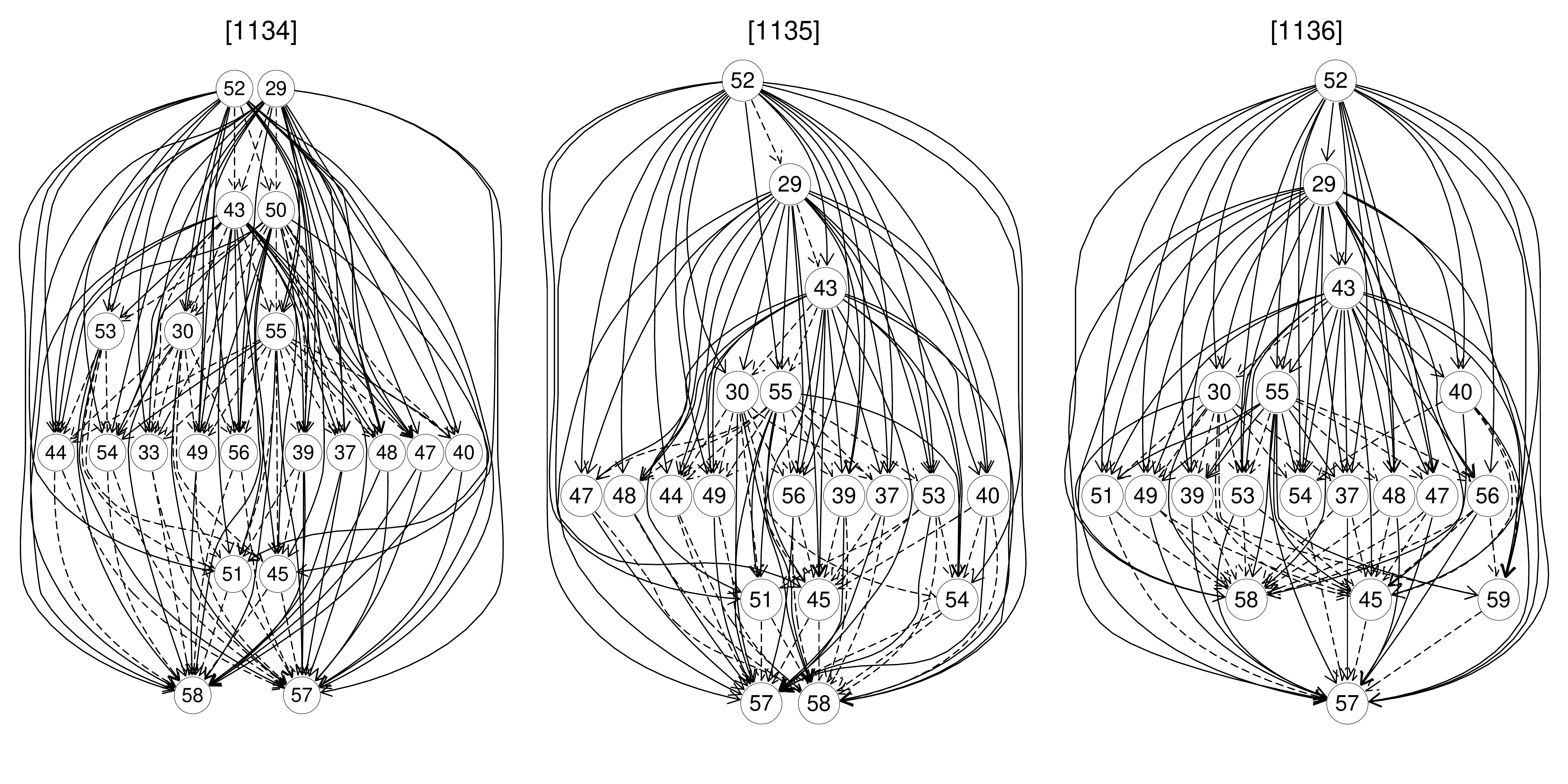}
    \caption{Consensus partial orders for the years 1134-1136 (a selection from the years 1080-1155) under the constrained-effects model of \sec~\ref{sec:poHB2aRS6}. These are the transitive closures corresponding to the reductions in  \fig~\ref{fig:CPO-poHB2aRS6} in \sec~\ref{sec:poHB2aRS6}. Dashed edges have posterior support at least $\xi=0.5$. Black edges have support above $0.9$. Vertex numbering as \fig~\ref{fig:POmcmc-poHB2aRS6}.}
    \label{fig:CPO-poHB2aRS6-close}
\end{figure}

\begin{figure}[H]
    \centering
    \includegraphics[width=6in, trim=0.6in 2.4in 0.6in 2.4in, clip]{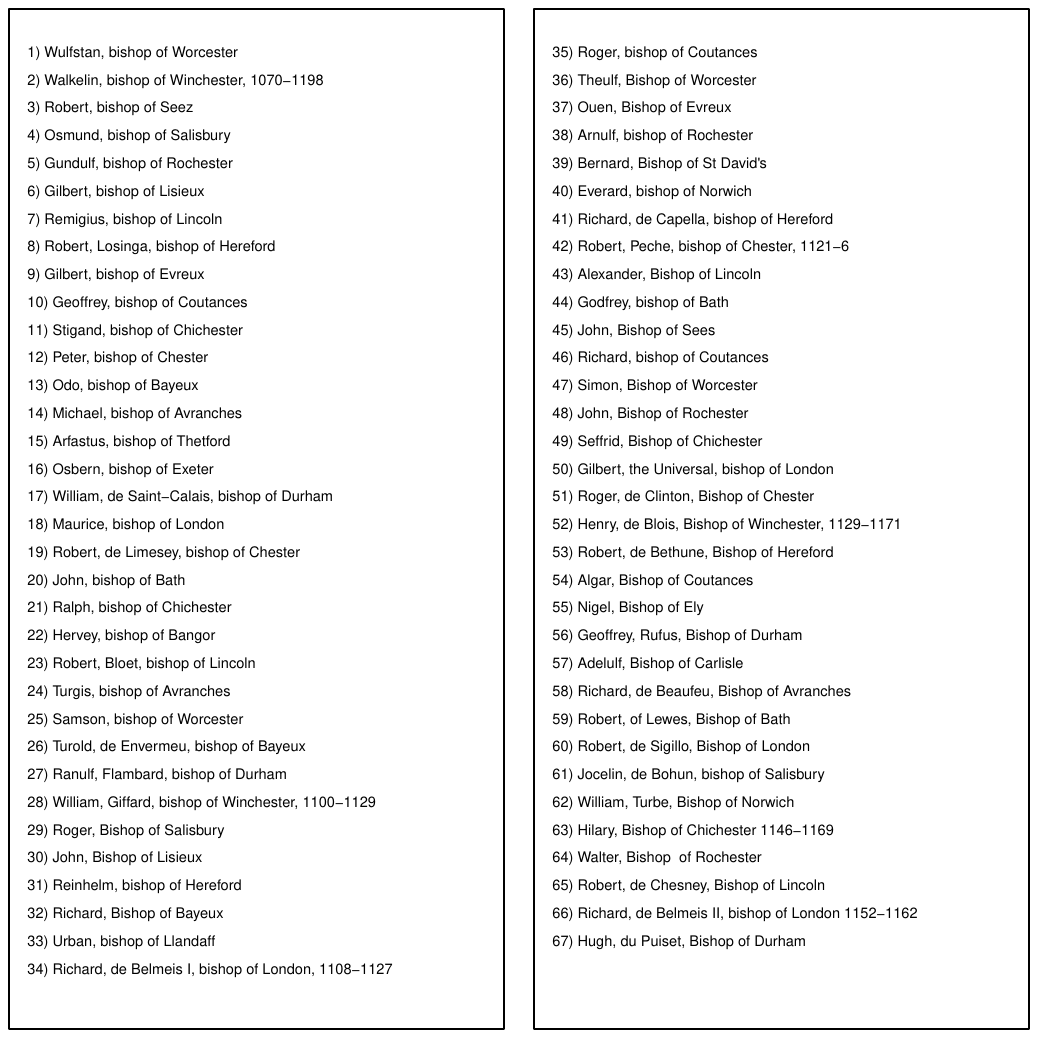}
    \caption{Name index for vertices in the order graphs in \figs~\ref{fig:intro-po-example}, \ref{fig:POmcmc-poHB2aRS6}, \ref{fig:CPO-poHB2aRS6}, \ref{fig:CPO-poHB2aRS6-close} and \ref{CPO-poHB17aRS4-all}.}
    \label{fig:bishop-names}
\end{figure}

\begin{figure}[H]
    \centering
    \hspace*{-0.8in}\includegraphics[width=6.75in,trim=0in 0.25in 0.25in 0.25in, clip]{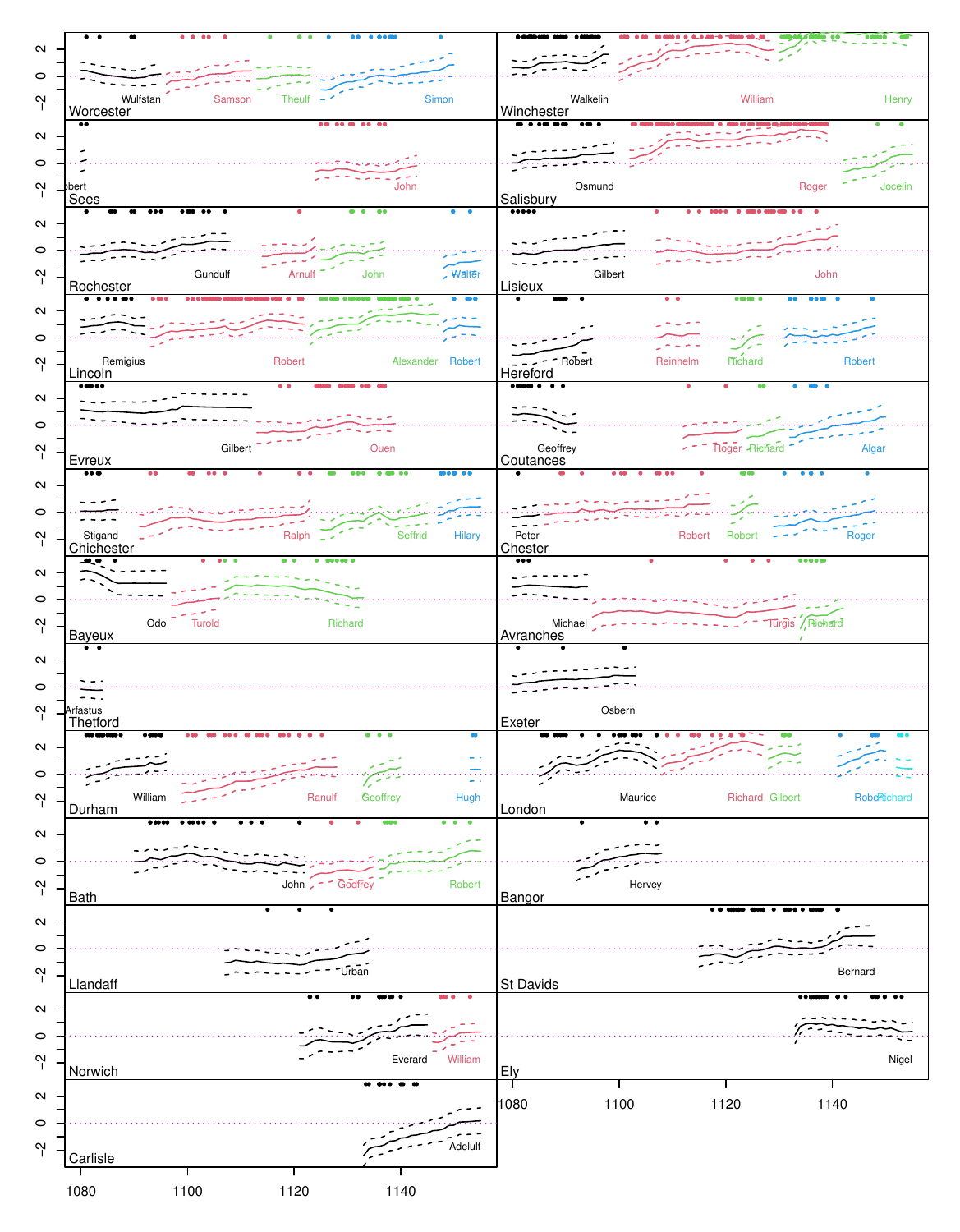}
    \caption{Bishop-status curves $\hat Z^{(t)}_j$ ($y$-axis values, solid curves) plotted for each bishop $j\in \M$ as a function of time ($x$-axis, the year) from $b_j$ to $e_j$ with standard errors at one-sigma. The dots along the top of each graph show the times of the lists in which the color-matched bishop below appeared. This is output from the constrained seniority effects analysis in \sec~\ref{sec:poHB2aRS6}.}
    \label{fig:Zhat-poHB2aRS6}
\end{figure}

\begin{figure}[H]
    \centering
    \hspace*{-0.4in}{\includegraphics[width=5.5in, trim=0.05in 0.6in 0.4in 0.8in, clip]{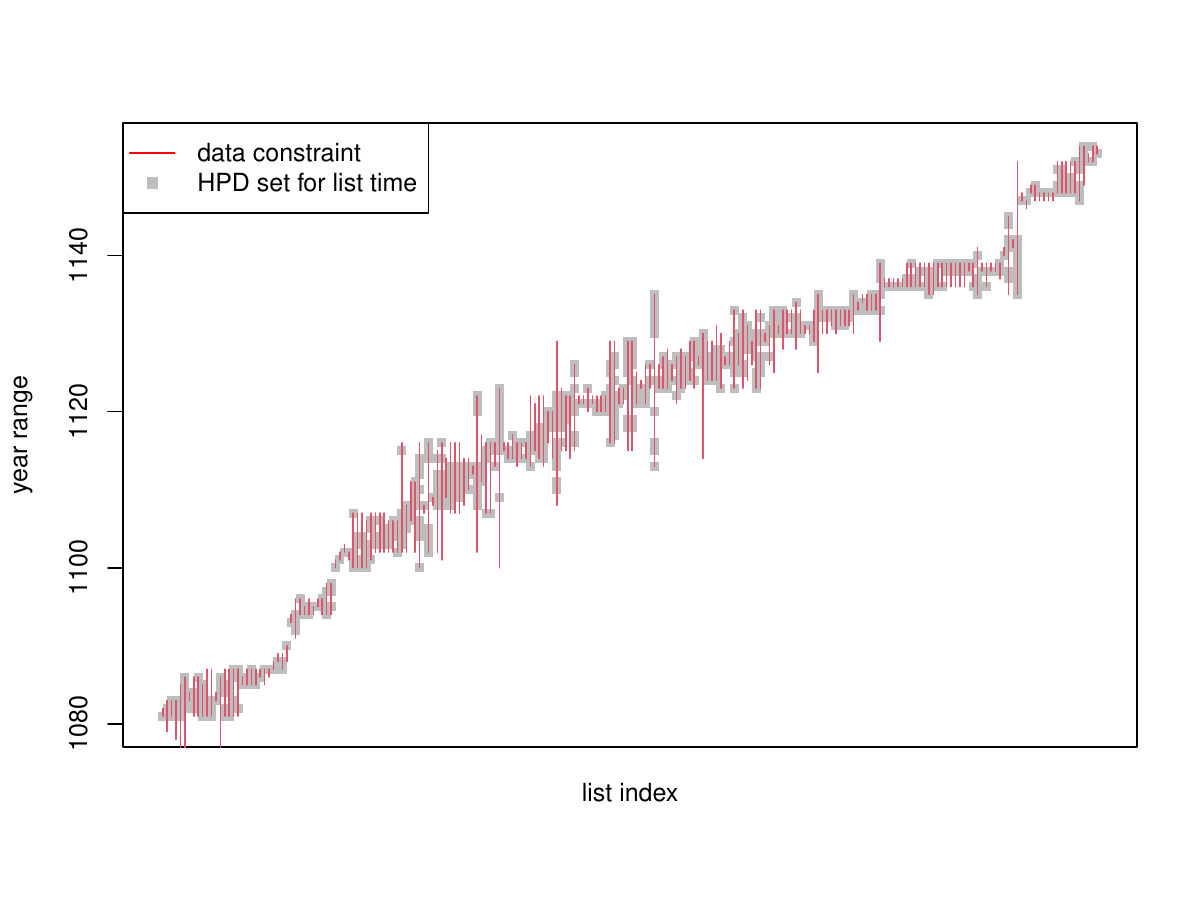}}
    \caption{Posterior distributions of the unknown list times $\tau_i,\ i\in \I$. Some lists have known times, shown are the posteriors for uncertain list times. The $x$-axis is the list index (sorted by prior mean time) and the $y$-axis is in years CE. Each thin vertical line shows the known constrained interval $[\tau^-_i,\tau^+_i]$ for each list. The grey region over each vertical line is a 90\% HPD set for the estimated age of the associated list. This is output from the constrained seniority effects analysis in \sec~\ref{sec:poHB2aRS6}.}
    \label{fig:tau-poHB2aRS6}
\end{figure}

Finally we plot the full set of consensus partial orders for the analysis in this section. These are most easily viewed in the online version of this paper using a zoom function. The graphs for the years 1134-1136 were selected and plotted in \fig~\ref{fig:CPO-poHB2aRS6} for ease of viewing. 

\begin{figure}[H]
    \centering
    \hspace*{-0.2in}\includegraphics[width=6.25in,trim=0in 0.2in 0in 0.0in, clip]{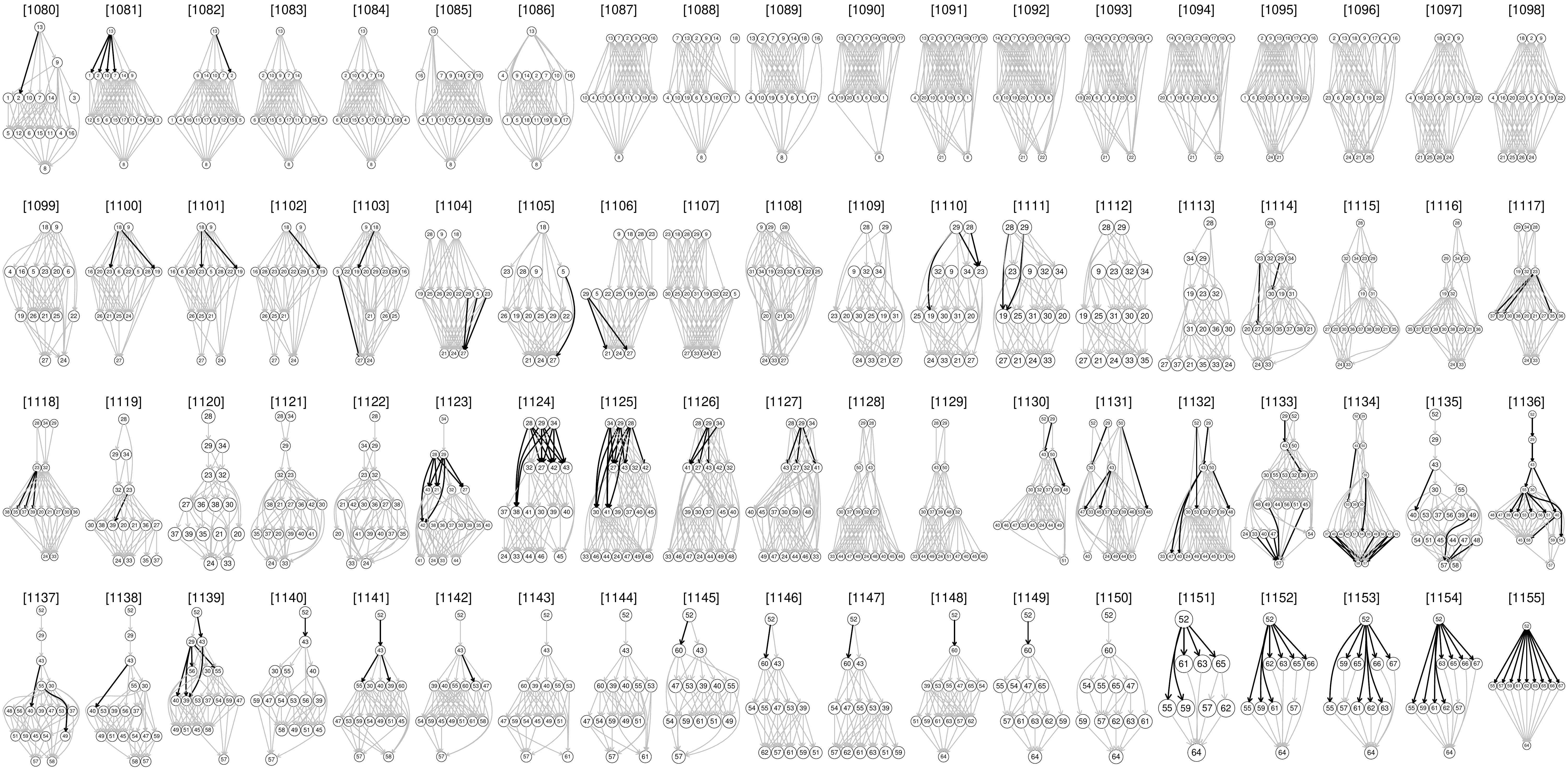}
    \caption{Consensus partial orders for all years 1080-1155 under the constrained-effects model of \sec~\ref{sec:poHB2aRS6}. These are transitive reductions. Grey edges have posterior support at least $\xi=0.5$. Black/solid edges have support above $0.9$. Vertex numbering as \fig~\ref{fig:POmcmc-poHB2aRS6}.}
    \label{CPO-poHB17aRS4-all}
\end{figure}

\subsection{Results from analyses like \sec~\ref{sec:poHB2aRS6} but with $K=2$ and $K=22$}\label{app:K-equals-18-posterior}

Here we gather together results from analyses identical to that in \sec~\ref{sec:poHB2aRS6} (constrained seniority effects, $\beta\in \B_{\Sb}$) but with $K=11$ there replaced by $K=2$ and then $K=22$. We looked at $K=2$ and $K=22$ for the unconstrained analyses in \app~\ref{app:poHB1aRS6} and considered how our estimates for $\rho, \theta$ and $p$ were affected. When we make a constrained analysis we are aiming to produce final results to pass onto historians and in particular, relations need to be well-estimated.

Changing the number of features/columns in the {authority} vector $U^{(t)}_j$ changes the prior distribution over posets. The depth distributions are similar, but there is higher weight on lower depth posets at larger $K$. This can be seen by comparing the two graphs in \fig~\ref{fig:prior-depth-dbns-curves-NF9-18}.
The effect is visible in \fig~\ref{fig:prior-post-depth-boxes-poHB2aRS6} where the "whiskers" in the boxplot for $K=22$ often reach from the maximum depth to the minimum depth, which is rare in the $K=11$ case.

It may be desirable to use the $K=22$ prior as the $K=11$ prior puts very little probability on posets of depth one (not zero, but small). Depth one corresponds to the empty order. While putting low prior probability on the empty order reflects prior knowledge (it is very unrepresentative of what we know about the actors in our data), it is a little stronger than is really desirable, so some sensitivity analysis is in order.

In contrast the $K=2$ prior weights quite strongly in favor of deeper partial orders. This might be improved by some compensating adjustment of the $\rho-$prior. However, for the purpose of this comparison we simply changed $K$ and left all other aspects of the model unchanged. If we could ``get away'' with using $K=2$ we certainly would, as the dimension of the authority parameters $U$ is $\dim(U)=K\times \sum_t m_t$. In \sec~\ref{sec:poHB2aRS6}, where $K=11$ we fit with $\dim(U)=13,453$. This drops to $\dim(U)=2446$ with $K=2$ and makes for easier MCMC sampling. The class of two dimensional posets is still a very rich model space (it includes all VSPs). Also, at least for short time intervals, VSPs fit fairly well, with a Bayes Factor against the full poset model which is only weakly in favor of posets. VSPs have dimension at most two, so all VSPs can be represented by taking $K=2$. This suggests that quite small values of $K$ may sometimes be adequate. 

\begin{table}[]
    \centering
    \begin{tabular}{c|cc}
    $K$ & \% relations & MAD \\
    \hline
     2    &  95\% & 0.02\\
     11    &  100\% & 0 \\
     22  & 97\% & 0.01\\
     \hline\\
    \end{tabular}
    \caption{Results from comparison of posterior distributions at $K=2$, $11$ and $22$ for the constrained analysis of \sec~\ref{sec:poHB2aRS6} but varying the dimension $K$ of the latent authority vectors $U_{j,:}^{(t)},\ j\in \M,\ t\in[B,E]$. Comparisons are against $K=11$ so it shares 100\% of its relations with itself. The \% relations column gives the percentage of agreeing estimated relations (including ``no-relation'') between the consensus posets for the two analyses in each comparison. The ``MAD'' column gives the Median Absolute Difference in the estimated probability for relations. Reported numbers are accurate to at least the reported significant digits.}
    \label{tab:compare_K}
\end{table}


Results are given in Table~\ref{tab:compare_K}. We find conclusions are very similar for the three $K$-values. At the $\xi=0.5$ threshold for consensus posets, $95\%$ of estimated relations are the same between $K=2$ and $K=11$. However, the two analyses may support a relation at levels just above and below the consensus poset threshold, so this doesn't tell the whole story. Drilling down, if $i,j\in\mathcal{M}_t$ are two bishops active at time $t$ and $\hat \xi^{(K,t)}_{\e {j_1}{j_2}}$ (see \eq~\ref{eq:post-mean-reln-probs} in \app~\ref{app:CPO}) is the estimated posterior probability for the relation $j_1\succ_{h^{(t)}}\! j_2$ in an analysis with $K\in\{2,11,22\}$ then the number of linearly independent $\hat \xi^{(K,t)}_{\e {j_1}{j_2}}$-values we estimate in each analysis is $18912$. The Median Absolute Difference (MAD) 
$|\hat \xi^{(11,t)}_{\e {j_1}{j_2}}-\hat \xi^{(2,t)}_{\e {j_1}{j_2}}|$ in these estimated probabilities is $0.02$. Going from $K=2$ to $K=11$ isn't changing conclusions a great deal. However, if we do the same analysis comparing $K=11$ and $K=22$ we get slightly better agreement: $97\%$ of relations agree and the MAD, $|\hat \xi^{(22,t)}_{\e {j_1}{j_2}}-\hat \xi^{(11,t)}_{\e {j_1}{j_2}}|$, is $0.01$. We interpret these results to mean that once $K$ is ``large enough'' we estimate the same relations.

 

\newpage\section{Parameter estimates from multiple short time-windows}\label{app:constant-rho-theta-p-intervals}

In the poset models of \secs~\ref{sec:poHB1aRS6}, \ref{sec:poHB2aRS6} and \app~\ref{app:K-equals-18-posterior} we assume that $p,\rho,\theta$ do not vary over time. One simple way to check this is to estimate their posterior distributions over short time windows. We use the same intervals (first column of Table~\ref{tab:vsp-bucket-BF}) as \secs~\ref{sec:vsp-bucket-model-and-test} and \ref{app:pl-mix-elpd}. Results are shown in
\fig~\ref{fig:constant-rho-theta-p-intervals}.
\begin{figure}[H]
    \centering
    \hspace*{-0.2in}\includegraphics[width=5.5in,trim=0.2in 0.1in 0.15in 0in,clip]{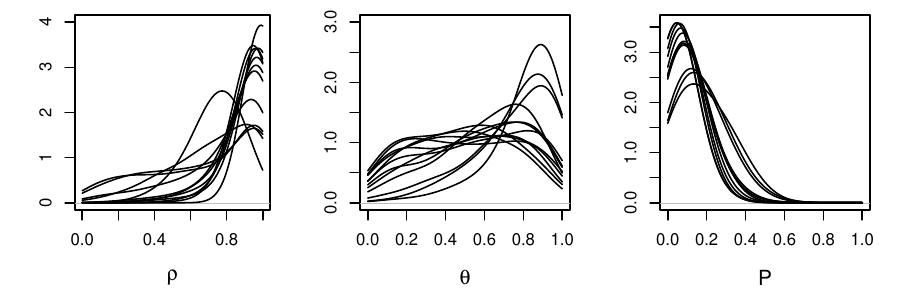}
    \caption{Posterior distributions for $\rho$, $\theta$ and $p$ estimated in each of the twelve time intervals listed in column 1 of Table~\ref{tab:vsp-bucket-BF} and discussed in \app~\ref{app:constant-rho-theta-p-intervals}}
    \label{fig:constant-rho-theta-p-intervals}
\end{figure}
The distributions for $\rho$ and $p$ show very little variation over time. The distributions for $\theta$, the degree of correlation between the poset in one year and the next, do vary somewhat. This density is focused on large values for the first interval 1080-1084 and for intervals after about 1135. The other intervals are close to the (uniform) prior. However, this simply reflects the fact that we get little information about a correlation from a short time interval ($5$ years).
We see no evidence here against our assumption that $\rho,\theta$ and $p$ are constant in time.

\newpage\section{VSP order and Bucket Order analysis - further details}
\label{app:vsp-bucket-BF}

\subsection{Further background on VSPs} 
Let $\A_1,\A_2$ be disjoint sets of actor labels and let $a_1\in\H_{\A_1}$ and $a_2\in \H_{\A_2}$ be posets. Let
\begin{align*}
a_1\os a_2 &= a_1\cup a_2 \bigcup_{j_1\in \A_1\atop j_2\in \A_2}\{\e{j_1}{j_2}\},\\
\intertext{denote series combination setting all actors in $a_1$ ``above'' those in $a_2$ and let}
a_1\op a_2 &= a_1\cup a_2,
\end{align*}
denote the parallel combination in which all actors in $a_1$ are unordered with respect to those in $a_2$. The class $\V_{\A}$ of VSP orders on $\A$ is defined recursively \citep{tarjan82}: 
    for $j\in \A$ let $\V_j=\{\emptyset\}$ to be the empty order on a single actor $j$;
$\V_{\A}$ is the set of all posets $H$ which can be decomposed as $H=a_1\os a_2$ or $H=a_1\op a_2$ for some partition $\A_1,\A_2$ of $\A$ and some $a_1\in\V_{\A_1}$ and $a_2\in\V_{\A_2}$. 

A sequence of series or parallel combinations defining a given VSP is represented by a Binary Decomposition Tree (BDT, \cite{valdes78}): leaves match elements of $\A$ and each internal vertex is an $\os$ or $\op$ operation on the VSPs rooted by its left and right child vertices. The BDT displayed in \fig~\ref{fig:vsp-example} represents the poset in \fig~\ref{fig:po-example}, which happens to be a VSP. The ``plus'' and ``minus'' symbols on the vertices below an $S$-vertex indicate the upper and lower VSP in a $\os$-operation. Any given VSP may have many BDT representations (\cite{jiang23} counts them). If $H\in \V_\A$ is a VSP then $|\L[H]|$ can be computed in a time linear in $|\A|$ \citep{wells1971elements}.
\begin{figure}
    \centering
    \begin{minipage}[c]{2.5in}
    \includegraphics[width=2in, trim=1.2in 1.1in 0.8in 0.7in]{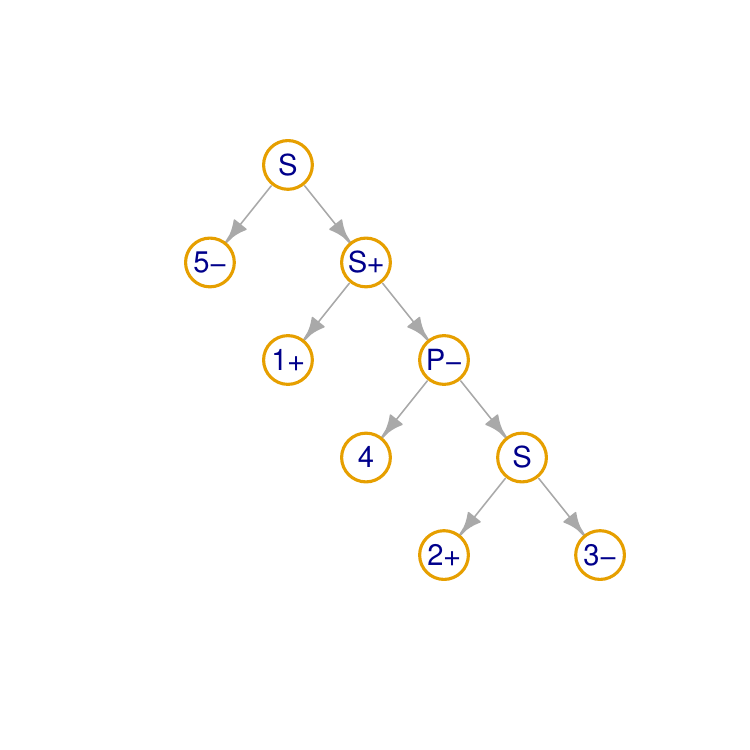}
    \end{minipage}
    \hspace*{0.3in}
    \begin{minipage}[c]{1.5in}
    \includegraphics[width=1in, trim=1.2in 1in 0.8in 0.9in]{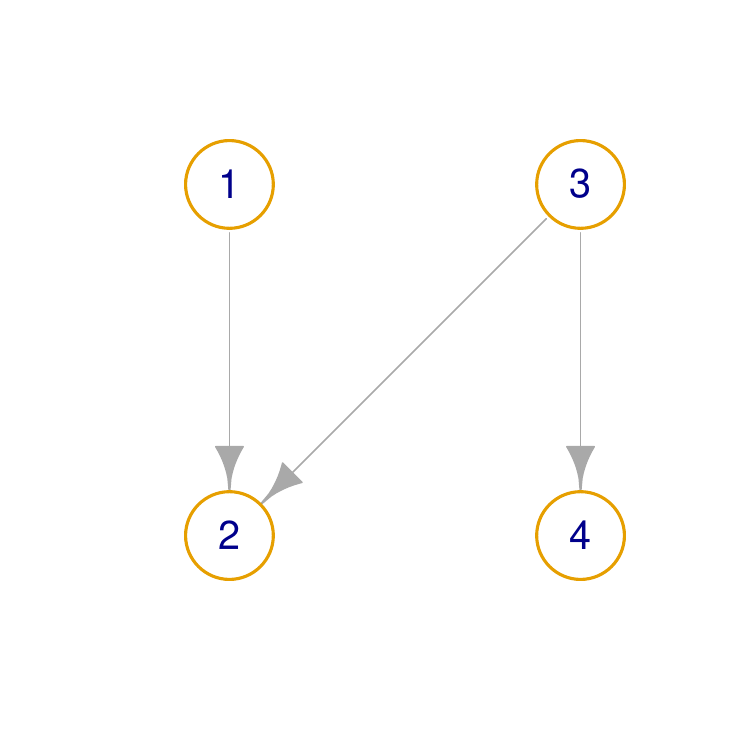}
    \end{minipage}
    \caption{(left) BDT representation of the poset in \fig~\ref{fig:po-example}. Working up from the bottom, 2 is above 3, 4 is unordered with respect to 2 and 3, 1 is above all three actors, and finally 5 is below all of 1-4. (right) The forbidden sub-graph $F$. A VSP order is any poset which does not contain a sub-graph isomorphic to $F$.}
    \label{fig:vsp-example}
\end{figure}

\subsection{Testing partial orders against restricted orders} \label{vsp-bucket-testing-methods}
In this section we consider data overlapping a time-window $[t_1,t_2]$ by at least a half. Let
\[
y^{(t_1,t_2)}=\left\{y_i: i\in \I,\  \frac{1+\min(t_2,\tau^+_i)-\max(t_1,\tau^-_i)}{1+\tau^+_i-\tau^-_i}\ge 0.5\right\}.
\] 
denote this windowed data. Consider the Bayes factor 
\begin{align*}
    B_{\V,\H}&=\frac{p(y^{(t_1,t_2)}|h\in\V^{(t_1,t_2)})}{p(y^{(t_1,t_2)}|h\in\H^{(t_1,t_2)})},
\end{align*}
with marginal likelihoods $p(y^{(t_1,t_2)}|h\in\V^{(t_1,t_2)})$ and $p(y^{(t_1,t_2)}|h\in\H^{(t_1,t_2)})$. Replace $\V$ with $\K$ to give $B_{\K,\H}$ for Bucket Orders. 
Because the sets $\K^{(t_1,t_2)}\subset\V^{(t_1,t_2)}\subset \H^{(t_1,t_2)}$ are nested, these Bayes factors are easily estimated using Savage-Dickey estimators and MCMC samples from the full posterior $\pi(\rho,\theta,U,\beta,\tau,p|y^{(t_1,t_2)})$ in \eq~\ref{eq:posterior}. We have
\begin{align}\label{eq:vsp-bucket-BF}
    B_{\V,\H}&=\frac{\pi(\V^{(t_1,t_2)}|y^{(t_1,t_2)})}{\pi(\V^{(t_1,t_2)})},
\end{align}
where  
\[
\pi(\V^{(t_1,t_2)}|y^{(t_1,t_2)})=E_{\rho,\theta,U,\beta,\tau,p|y^{(t_1,t_2)}}\left(\mathbb{I}_{h(Z(U,\beta;s))\in\V^{(t_1,t_2)}}\right)
\]
is the posterior expectation and $\pi(h\in\V^{(t_1,t_2)})$ is the corresponding prior expectation.

\subsection{VSP and Bucket Order Bayes Factor Estimates}
The estimates used to make the plot in \fig~\ref{fig:vsp-bucket-BF} are given here.
These numbers are discussed in \sec~\ref{sec:vsp-bucket-model-and-test}. The effective sample sizes are for the MCMC output $\mathbb{I}_{h^{(t)}\in\X^{(t_1,t_2)}}$ where $\X$ is $\B$ or $\V$ and $h^{(t)}=h(Z(U^{(t)},\beta;s)),\ t\in [B,E]$ is taken from a chain targeting $\pi(\rho,\theta,U,\beta,\tau,p|y^{(t_1,t_2)})$ using the lists in the time interval $[t_1,t_2]$. These are the samples used to estimate the numerator in \eq~\ref{eq:vsp-bucket-BF}. The large ESS
values in this table are associated with low posterior probabilities for the events $h\in \B^{(t_1,t_2)}$ and $h\in \V^{(t_1,t_2)}$.
The resulting rare event process is approximately iid. It is unclear to us how to measure uncertainty reliably for these cases, but these Bayes factors must be small, as the prior probabilities for the same events are relatively large, and the posterior probabilities are at least much smaller. 

\begin{center}
\begin{table}
\caption{VSP and Bucket Order Bayes Factors  \label{tab:vsp-bucket-BF}}
\hspace*{-0.2in}\begin{tabular}{|r|rrrrr|rrrrr|}
\multicolumn{11}{c}{}\\
\hline 
&\multicolumn{5}{|c}{Bucket Orders} &\multicolumn{5}{|c|}{VSP orders} \\
Period & Prior & Post. & BF & ESS & BF std.err. & Prior & Post. & BF & ESS & BF std.err.\\ \hline
1080-1084 &  0.085 & 0 & 0 & - & - & 0.22 & 0 & 0 & - & - \\
1086-1090 &  0.15 & 0.17 & 1.2 & 45 & 0.38 & 0.31 & 0.34 & 1.1 & 33 & 0.27 \\
1092-1096 &  0.27 & 0.14 & 0.53 & 29 & 0.24 & 0.48 & 0.31 & 0.65 & 53 & 0.13 \\
1104-1108 &  0.18 & 0.12 & 0.66 & 31 & 0.32 & 0.36 & 0.21 & 0.6 & 29 & 0.21 \\
1110-1114 &  0.42 & 0.25 & 0.6 & 81 & 0.12 & 0.72 & 0.64 & 0.89 & 230 & 0.044 \\
1118-1122 &  0.1 & 0.22 & 2.1 & 38 & 0.65 & 0.24 & 0.53 & 2.2 & 69 & 0.25 \\
1126-1130 &  0.16 & 0.13 & 0.86 & 75 & 0.25 & 0.32 & 0.35 & 1.1 & 60 & 0.19 \\
1128-1132 &  0.084 & 0.016 & 0.19 & 96 & 0.15 & 0.22 & 0.067 & 0.31 & 15 & 0.3 \\
1132-1134 &  0.088 & 0.0053 & 0.06 & 950 & 0.027 & 0.22 & 0.045 & 0.21 & 230 & 0.063 \\
1138-1142 &  0.13 & 0.024 & 0.18 & 80 & 0.13 & 0.28 & 0.11 & 0.37 & 69 & 0.13 \\
1144-1148 &  0.4 & 0.59 & 1.5 & 99 & 0.12 & 0.68 & 0.87 & 1.3 & 330 & 0.027 \\
1150-1154 &  0.2 & 0.17 & 0.83 & 32 & 0.32 & 0.39 & 0.36 & 0.93 & 45 & 0.19 \\ 
\hline
\end{tabular}
\end{table} 
\end{center}

\newpage\section{Plackett-Luce models}\label{app:PL-appendix-section}

In \apps~\ref{app:PL-time-series-model} and \ref{app:pl-mix-elpd} we define and fit two different types of Plackett-Luce model \citep{luce1959possible, plackett1975analysis}. 
This is done for comparison purposes. 

\subsection{Comparison with results from a Plackett-Luce time-series model}\label{app:PL-time-series-model}

Here we specify and fit a time-series Plackett-Luce model like that given in \citet{glickman2015stochastic}. Our purpose is to test the conclusions we reached using the poset model and show that we get similar results when we use a qualitatively different model for the same data. We fit the Plackett-Luce time-series model to the full data from 1080-1155 and compare with results in \secs~\ref{sec:poHB1aRS6} and \ref{sec:poHB2aRS6}. The model is simpler, there is no partial order, but it does share some parameters ($\beta$) and some parameters play very similar roles ($\lambda$ and $U$). Results are remarkably similar. In the next section we will consider formal model comparisons.

\subsubsection{Parameters and Likelihood}
For $t\in [B,E]$ denote by $\lambda^{(t)}_{j}$ the authority of bishop $j\in \M$ in year $t$. Notation for seniority $s$ and seniority effects $\beta$ are unchanged from \sec~\ref{sec:covariate-effects}.
The scalar time series $\lambda_j=(\lambda^{(t)}_j)_{t=B}^{E}$ plays a similar role to the multidimensional time series of {authority} variables $U_j$ defined in \sec~\ref{sec:prior-prob-dbns}. However, for ease of coding we defined $\lambda^{(t)}_j$ for all $t\in [E,B]$ not just $t\in [b_j,e_j]$ as we did for $U_j$. This means $\lambda=(\lambda_j)_{j\in\M}$ is a ``full'' $M\times T$ matrix. Denote by $\lambda^{(t)}=(\lambda^{(t)}_j)_{j\in\M}$ the column vector of $\lambda$-values in each time $t\in [B,E]$. 

The times of some of the lists are uncertain so we fixed these to be the centres of their intervals, $\hat\tau_i=(\tau_i^-+\tau_i^+)/2$.
The Plackett-Luce likelihood is then
\begin{align}\label{eq:Plackett-Luce-likelihood}
    p_{PL}(y | \lambda, \beta) & = \prod_{i=1}^N \prod_{j=1}^{n_i} \frac{\exp(\lambda^{(\hat\tau_i)}_{y_{ij}} + \beta_{s_{\hat\tau_i,y_{ij}}})}{\sum_{k=j}^{n_i} \exp(\lambda^{(\hat\tau_i)}_{y_{ik}} + \beta_{s_{\hat\tau_i,y_{ik}}})} 
\end{align}
The ``linear predictor'' $\lambda^{(t)}_j + \beta_{s_{t,j}}$ captures the ``status'' of bishop $j$ in year $t$ just as the vector $Z^{(t)}_j$ did in \sec~\ref{sec:latent-variable-param}. The covariates enter in the same way as before (\eq~\ref{eq:additive-model-for-Z}) and so $\lambda^{(t)}_j$ is the part of status not determined by seniority. 

\subsubsection{Plackett-Luce Priors}
The likelihood in \eq~\ref{eq:Plackett-Luce-likelihood} is invariant under $\lambda\to \lambda+\delta$ and $\beta\to\beta+\delta'$ for any $\delta,\delta'\in\R$.
As noted above in \sec~\ref{sec:non-identifiable-UbetaPO}, the $U,\beta$-latent variable model for posets has the same feature. We now remove this non-identifiability in our parameterisation.

Consider a projection matrix
\[
Q_{M} = I_M - \frac1M \mathbf{1}_M \mathbf{1}_M^T,
\] 
with $I_M$ the $M\times M$ identity matrix and $\mathbf{1}_M$ a column of $M$ ones. Let $Q_{\Sb}$ be defined similarly. The matrix $Q_M$ projects vectors in $\R^M$ to centred vectors $\{v\in\R^M: \sum_{j\in\M} v_j=0_M\}$. 
Under a Plackett-Luce model with covariates for seniority, the generative model for authority is a centred $AR(1)$ process,
\begin{equation*}
    \lambda^{(t)} = \theta \lambda^{(t-1)} + Q_M \epsilon^{(t)},
\end{equation*}
where $\epsilon^{(t)}\sim N(\mathbf{0}_M,\sigma^2I_M)$ iid for $t=B+1,...,E$.
If $\lambda^{(t-1)}$ is centred, then $\lambda^{(t)}$ is also centred. The process is initialised in its prior equilibrium distribution,
\begin{align}
\lambda^{(B)}&=Q_M\epsilon^{(b)},\nonumber\\
\intertext{where}
\epsilon^{(b)}&\sim N\left( 0_M, \frac{\sigma^2}{1 - \theta^2} I_M\right).\label{eq:PL-lambda-AR1init}
\end{align}
Since 
\begin{align}
    \lambda^{(t)}_1=-\sum_{j=2}^{M} \lambda^{(t)}_j,
    \quad 
    \beta_1=-\sum_{r=2}^S \beta_r, \label{eq:PL-centre-lambda-beta}
\end{align} 
we work with a non-singular prior density with parameters $\lambda^{(t)}_{2:M},\ t\in[B,E]$. Denote by $Q_{2:M}$ the $M-1\times M$ matrix obtained by dropping the first row from $Q_M$ (with $Q_{2:\Sb}$ defined similarly). The densities determined by these generative models are
\begin{align*}
    \pi(\lambda^{(t)}_{2:M}|\lambda^{(t-1)}_{2:M},\theta,\sigma)&=N\left(\lambda^{(t)}_{2:M};\theta \lambda^{(t-1)}_{2:M},\sigma^2Q_{2:M}Q_{2:M}^T\right),\\
    \pi(\lambda^{(B)}_{2:M}|\theta,\sigma) &= N \left(\lambda^{(B)}_{2:M}; \mathbf{0}_{2:M}, \frac{\sigma^2}{1 - \theta^2}Q_{2:M}Q_{2:M}^T \right),\\
    \intertext{and}
    \pi(\lambda_{2:M}|\theta,\sigma)&= \pi(\lambda^{(B)}_{2:M}|\theta,\sigma)\prod_{t=B+1}^{E} \pi(\lambda^{(t)}_{2:M}|\lambda^{(t-1)}_{2:M},\theta,\sigma).
\end{align*}
The prior on $\beta$ is similarly
\begin{equation}
    \pi(\beta) = N(\beta_{2:\Sb} ; \mathbf{0}_{\Sb-1}, Q_{2:\Sb}Q_{2:\Sb}^T).
\end{equation}
There is nothing ``special'' about $\lambda_1$ or $\beta_1$ in these priors. They are determined by 
\eqs~\ref{eq:PL-centre-lambda-beta}, but they are still exchangeable with the other parameters in the prior.
We take a $\mbox{Unif}(0, 1)$ prior for $\theta$ and a $\mbox{Gamma}(2, 2)$ for $\sigma$ as capturing reasonable variation.

\subsubsection{Posterior distribution}
The Plackett-Luce posterior is then
\begin{equation}
    \pi_{PL}(\lambda_{2:M}, \beta_{2:\Sb}, \sigma, \theta|y) \propto p_{PL}(y | \lambda, \beta) \pi(\lambda_{2:M} | \sigma, \theta) \pi(\beta_{2:\Sb},\sigma,\theta),
\end{equation}
with $\lambda_1$ and $\beta_1$ in $\lambda$ and $\beta$ appearing in $p_{PL}(y | \lambda, \beta)$ given by \eqs~\ref{eq:PL-centre-lambda-beta}.

We extended the $\lambda_j$-process for each bishop $j\in\M$ outside the range $[b_j,e_j]$.
This made coding easier, but this means a large set of extra parameters must be brought into equilibrium in the target. However, by the stationarity of an $AR(1)$ process, the posterior marginal (integrating over $\lambda^{(t)}_j$ for $t\not\in [b_j,e_j]$) of the fitted model is identical to the model which ``starts'' the $\lambda_j$-process at $t=b_j$ with the $AR(1)$ equilibrium in \eq~\ref{eq:PL-lambda-AR1init} and terminates it at $t=e_j$.

\subsubsection{Results for Plackett-Luce Time-series model}

We use a simple random-walk MCMC algorithm to target $\pi_{PL}(\lambda_{2:M}, \beta_{2:\Sb}, \sigma, \theta|y)$, on the full data of \secs~\ref{sec:poHB1aRS6} and \ref{sec:poHB2aRS6}. We took unconstrained $\beta\in \B_0$. Denote by $\lambda^{(t,l)}_j,\ t\in[B,E],\ j\in\M_t, \theta^{(l)}, \sigma^{(l)},\ l=1,...,L$ the realised MCMC samples after thinning and burn-in. 
The Effective Sample Sizes for
the parameters were $\theta (41)$, $\sigma (23)$, $\beta (21 \mbox{ parameters, } 538 - 2016)$ $\lambda (5092  \mbox{ parameters, } 55 - 6779)$.
Mixing for the key parameter vectors $\beta$ and $\lambda$ was fair, and the agreement between results from the two analyses in this section and \secs~\ref{sec:poHB1aRS6} and \ref{sec:poHB2aRS6} supports our conclusion that the samples are representative.

Results for quantities of historical interest are very similar to those obtained using the partial-order posterior in \eq~\ref{eq:posterior}. The posterior means for $\theta$ and $\sigma$ were $0.93 (0.02)$ and $0.69 (0.085)$ (posterior standard deviation in parenthesis). The posterior for $\theta$ is similar to its poset counterpart in \fig~\ref{fig:rho-th-p-poHB1aRS6}. Posterior effect $\beta$ distributions are plotted in \fig~\ref{fig:beta-poHB1aRS6-PL-version}.
\begin{figure}
    \centering
    \includegraphics[width=2.5in,trim=2.55in 3.05in 2.5in 3.6in,clip]{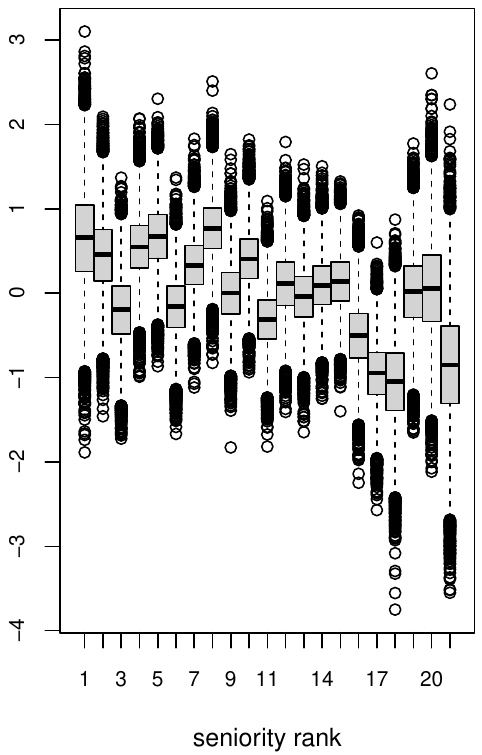}
    \caption{(Left) Marginal posterior distributions of seniority effect parameters from the unconstrained seniority effects analysis in \sec~\ref{sec:poHB1aRS6} with likelihood $p_{(D)}$ in \eq~\ref{eq:lkd-noisy-down}. (Centre) As left with likelihood $p_{(U)}$ in \eq~\ref{eq:lkd-noisy-up}. (Right) Corresponding distributions estimated using the Plackett-Luce time-series model in \app~\ref{app:PL-time-series-model}.}
    \label{fig:beta-poHB1aRS6-PL-version}
\end{figure}
They are qualitatively similar to those for the poset-analyses in \fig~\ref{fig:beta-poHB1aRS6}. However, seniority effects in \fig~\ref{fig:beta-poHB1aRS6-PL-version} do not decline as smoothly as we know they should, and clearly not as well in line with historical expectations as the corresponding poset estimates in \fig~\ref{fig:beta-poHB1aRS6}. This suggests PL is using the degrees of freedom in the seniority effects to accommodate misspecification in the observation model.
In \fig~\ref{fig:PL-lambda-posterior-mean-curves} we plot the curves
\[
\hat\lambda^{(t)}_j=\frac{1}{L}\sum_{l=1}^L \lambda^{(t,l)}_j 
\]
over $t\in[b_j,e_j]$ for each $j$ arranged by diocese as we did for $\bar U^{(t)}_j$ defined in \sec~\ref{sec:non-identifiable-UbetaPO}. The resulting authority-measures show similar dependence on time. This shows that our results for these features are robust.
\begin{figure}
    \centering
    \hspace*{-0.5in}\includegraphics[width=6.75in,trim=0in 0.25in 0.25in 0.25in, clip]{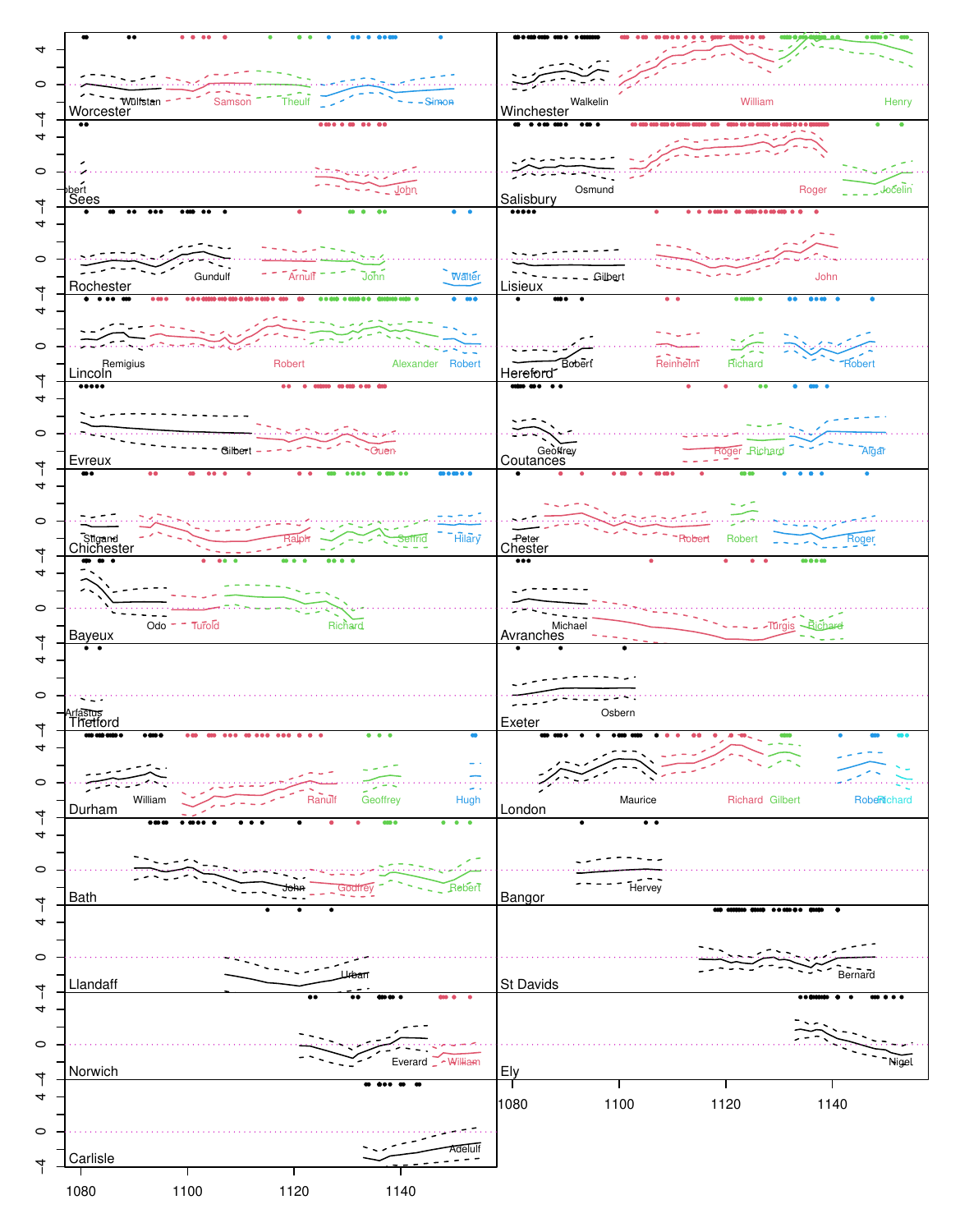}
    \caption{Authority-curves $\hat\lambda^{(t)}_j$
($y$-axis values) over $t\in[b_j,e_j]$ ($x$-axis, the year) for each bishop $j\in\M$ arranged by diocese, for the Plackett-Luce time-series model in \app~\ref{app:PL-time-series-model}. See \fig~\ref{fig:Uhat-poHB2aRS6} for details and comparison. }
    \label{fig:PL-lambda-posterior-mean-curves}
\end{figure}

\subsection{Comparison with a Plackett-Luce mixture model}\label{app:pl-mix-elpd}

\subsubsection{Parameters and Model}

In this section we define a Plackett-Luce mixture model for list data taken from an interval of time $[t_1,t_2]$ which is short enough to allow us to drop time dependence and covariates. The Leave-One-Out cross validation we use below is time consuming and we do this for ease of computation. We fit this model and a poset model without covariates and carry out model comparison between the two models. Dropping time-varying covariates is natural when working in a short time window, as the effects due to bishop and seniority cannot be separated. Denote by $\I_{t_1,t_2}$ the indices of the list data covered by the window (at least half of the years $[\tau_i^-,\tau_i^+]$ fall within $[t_1,t_2]$) and let $N_{t_1,t_2}=|\I_{t_1,t_2}|$ be the number of lists in this window. Let $y^{(t_1,t_2)}=(y_i)_{i\in \I_{t_1,t_2}}$ be the windowed data.
A finite mixture of Plackett-Luce models was proposed as a robust model for ranked data with incomplete lists (``partially ranked data'', \citet{Mollica14,Mollica20}. A mixture of $D$ Plackett-Luce models for ranking a set $\M_{t_1,t_2}=\cup_{t\in [t_1,t_2]} \M_t$ containing $m_{t_1,t_2}=|\M_{t_1,t_2}|$ actors assumes that lists are sampled from a heterogeneous population composed of $D$ sub-populations. Each sub-population is modeled by one of the mixture components with Plackett-Luce status vector $\lambda^{(d)}\in\R^{m_{t_1,t_2}},\ d\in [D]$. 
Let $\lambda=(\lambda^{(d)})_{d\in [D]}$ be the $m_{t_1,t_2}\times D$ matrix of status parameters.

The likelihood for a mixture of $D$ Plackett-Luce models is then
\begin{equation}
    p_{mix\,PL}(y^{(t_1,t_2)}|\lambda,\omega) = \prod_{i\in \I_{t_1,t_2}} \sum_{d=1}^D\omega_d \prod_{j=1}^{n_i} \frac{ \exp(\lambda^{(d)}_{y_{i,j}})}{\sum_{k=j}^{n_i} \exp(\lambda^{(d)}_{y_{i,k}})},
\end{equation}
where $\omega_1,...,\omega_D$ are weights of mixture components and $\sum_{d=1}^D \omega_d=1$. Noninformative priors suggested by \cite{Mollica20} are assigned with $\exp(\lambda^{(d)}_{j}) \sim \mbox{Gamma}(1,0,001)$ for $j\in \M_{t_1,t_2}$ and $\omega_1,...,\omega_D\sim Dir(1,...,1)$. 

The posterior is
\[
\pi_{mix\,PL}(\lambda,\omega)\propto \pi(\lambda,\omega)p_{mix\,PL}(y^{(t_1,t_2)}|\lambda,\omega).
\]
We used the MCMC sampler available the R-package \texttt{PLMIX} provided by \cite{Mollica20} to target $\pi_{mix\,PL}$. This sampler uses a data augmentation scheme due to \citet{Caron12mcmc}.

In order to fit the model we must select the number of mixture components $D$. This is related to Bayesian model comparison. The \texttt{PLMIX} offers several model selection criteria to choose the number of mixture components. We used the Deviation Information Criterion. The selected $D$-values are given in the third column of Table~\ref{tab:elpd}.

\subsubsection{Model comparison}

The expected log pointwise predictive density (ELPD, \cite{Vehtari17}) measures the posterior predictive accuracy of a model. It is relatively straightforward to estimate and one natural measure of goodness of fit. Let $Y\in \P_{O}$ be a single generic list with membership $O\subseteq \M_{t_1,t_2}$. Suppose the true generative model for a list is $Y\sim f_O$. Denote by $p_{A,O}(Y|y_{y_1,y_2})$ the posterior predictive probability for $Y$ in model $A$ when the membership of list $Y$ is $O$. The ELPD for model $A\in \{\mbox{Plackett-Luce, Partial order}\}$ conditioned on the observed data is 
\begin{align*}
\elpd(A|y^{(t_1,t_2)}) &= \sum^{m_{t_1,t_2}}_{i=1} E_{Y\sim f_{o_i}}(\log p_{A,o_i}(Y|y^{(t_1,t_2)})), \\
&=\sum^{m_{t_1,t_2}}_{i=1} \left[\sum_{Y\in \P_{o_i}}f_{o_i}(Y)\,\log p_{A,o_i}(Y|y^{(t_1,t_2)}) \right]\, .
\end{align*}
The expectation is (up to a constant) the negative of the KL divergence between the true generative model for a list with membership $o_i$ and the model-$A$ posterior predictive for that list. We favour models with larger values of the $\elpd$.

We do not know $f_{O}$ but we do know the data are sampled from this distribution, so we follow \cite{Vehtari17} and use leave one-out cross validation (LOO-CV) to estimate $\elpd(A|y^{(t_1,t_2)})$. Our estimator is
\begin{equation}
    \widehat{\elpd}(A|y^{(t_1,t_2)}) = \sum^{m_{t_1,t_2}}_{i=1} \log p(y_i|y^{(t_1,t_2)}_{(-i)}) 
\end{equation}
where $y^{(t_1,t_2)}_{(-i)}=y^{(t_1,t_2)}\setminus \{y_i\}$ so the $i$-th observation is left out. 

\subsubsection{Test results}
In Table \ref{tab:elpd} the twelve short periods $[t_1,t_2]$ we chose are given in the first column. The same periods are used to for numerical comparison with VSP and Bucket Orders in Table~\ref{tab:vsp-bucket-BF}. The number of lists $m_{t_1,t_2}$ covered by each interval is in the second column. We estimate the posterior predictive probability for each omitted list in a period by averaging its likelihood over posterior samples, so the number of MCMC runs for each entry in columns four and five is equal to $m_{t_1,t_2}$. This impacted our choices for the locations of time intervals as we had to avoid periods with too few lists or too many long lists.

\begin{center}
\begin{table}[H]
\caption{ELPD estimates and standard errors for twelve periods. \label{tab:elpd}  }
\begin{tabular}{|c|c| c c| c|}
\multicolumn{5}{c}{}\\ \hline
Period & $\#$ lists &\multicolumn{2}{|c|}{Plackett-Luce Mixture} &Partial-order \\
$[t_1,t_2]$ & $m_{[t_1,t_2]}$ & components $D$ & $\widehat{\elpd}$ & $\widehat{\elpd}$ \\
\hline 
1080-1084 & 26 & 2 & -162.52 (17.55) & -54.35 (11.80) \\
1086-1090 & 10 & 1 & -61.43 (14.43) & -23.33 (6.26) \\
1092-1096 & 13 & 1 & -49.46 (8.55) & -23.03 (7.98) \\
1104-1108 & 23 & 1 & -95.56 (10.15) & -41.08 (9.69) \\ 
1110-1114 & 21 & 1 & -59.93 (8.16) & -33.55 (9.62) \\
1118-1122 & 20 & 2 & -89.01 (12.36) &  -21.07 (6.99) \\ 
1126-1130 & 25 & 2 & -113.03 (11.77) &  -36.13 (9.21) \\
1128-1132 & 33 & 2 & -190.34 (16.37) & -73.97 (13.30) \\
1132-1134 & 19 & 2 & -110.49 (22.07) & -35.66 (9.91) \\
1138-1142 & 32 & 1 & -186.09 (16.77) &  -88.09 (25.58) \\ 1144-1148 & 7 & 1 & -42.15 (4.26) & -17.56 (8.33) \\
1150-1154 & 11 & 1 & -60.05 (11.27) & -16.53 (4.20) \\
\hline 
\end{tabular} 
\end{table} \end{center}

The $\elpd$'s in column five are uniformly larger than those in column four, so the poset model is clearly favoured over the Plackett-Luce mixture model for all these short periods. This may be because the poset is able to restrict the set of probable lists more tightly around the lists in the data, while the fitted Plackett-Luce mixture is relatively more dispersed over lists. The reported uncertainty estimate is just the sample variation of the mean of the $m_{t_1,t_2}$-samples used to form the estimate. It is likely to be an underestimate because the same lists appear in test and training positions in different terms in the sum, so its terms are correlated \citep{Bengio2004,Sivula2022}.

\newpage\section{Checks on synthetic data}

\subsection{Checks on synthetic data: replicate main analysis}
\label{app:synth-poHB22aRS4}

We simulated synthetic data with the same structure as the real data and checked we could recover true parameter values. We took true values of $\rho,\theta,\beta,U,\tau,p$ sampled from the unconstrained posterior with $K=11$ discussed in \sec~\ref{sec:poHB1aRS6} and simulated data $y'_i\sim p_{(U)}(\cdot|h^{(\tau_i)}[o_i],p),\ i\in \I$. 

MCMC output targeting the posterior for the synthetic data $y'$ is shown in \fig~\ref{fig:synth-mcmc-poHB22aRS4}. The true $\rho,\theta$ and $p$ values fall within the support of their respective posteriors in \fig~\ref{fig:synth-rho-th-p-poHB22aRS4}. Credible intervals for $\bar Z$ in \fig~\ref{fig:synth-Zhat-poHB22aRS4} and $\beta$ in \fig~\ref{fig:synth-beta-poHB22aRS4} cover the truth well.  

In order to check reconstruction of relations in the fitted social hierarchy we report in \fig~\ref{fig:false-pos-neg-poHB22aRS4} the proportion of true relations which we would reject if we reject when the posterior odds for the relation are smaller than one third (the threshold for ``positive'' evidence in \cite{kass95}). This is equivalent to rejecting the relation $i\succ_{h^{(t)}}\!j$ for $i,j\in\M_t,\ t\in[B,E]$ if the estimated posterior probability $\hat p_{i,j}^{(t)}$ for the relation is less than one quarter.
This runs at around 10\% incorrect rejections at this level of evidence. We considered comparing the consensus partial order with the true partial order in each year and counting true positives and false negatives at a threshold $\xi=0.5$. However, one has then to choose the threshold and, as in our discussion in \app~\ref{app:K-equals-18-posterior}, a relation with  support just below a threshold of $\xi=0.5$ is still the favored relation of the three possibilities (correct relation, relation but wrong order and no relation) so that would be misleading.

\begin{figure}
    \centering
    \hspace*{-0.75in}\includegraphics[width=6.5in]{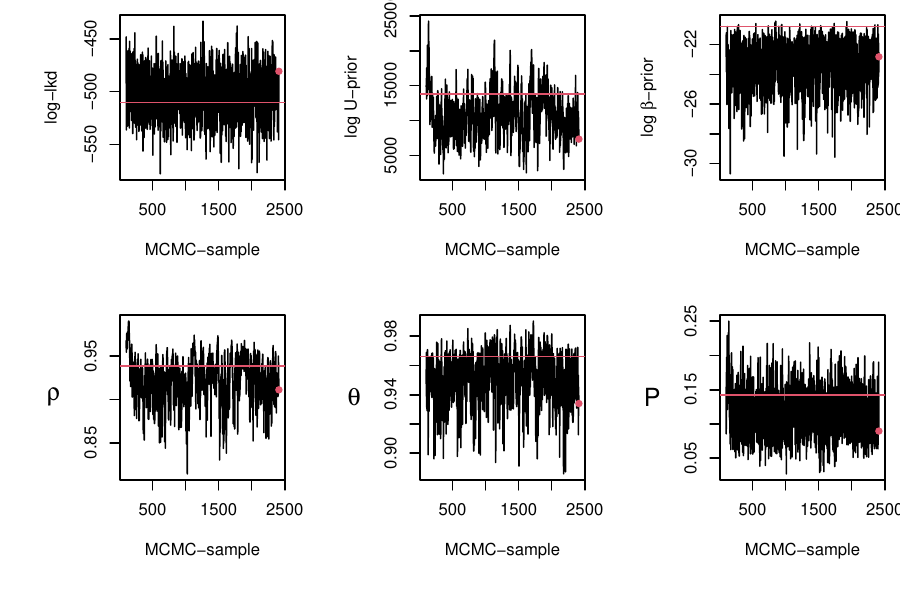}
    \caption{Synthetic Data. Selected MCMC traces from the unconstrained seniority-effect analysis in \app~\ref{app:synth-poHB22aRS4}. Horizontal line shows values at truth.}
    \label{fig:synth-mcmc-poHB22aRS4}
\end{figure}

\begin{figure}
    \centering
    \hspace*{0.0in}\includegraphics[width=5.5in,trim=0.2in 0.3in 0.05in 0in,clip]{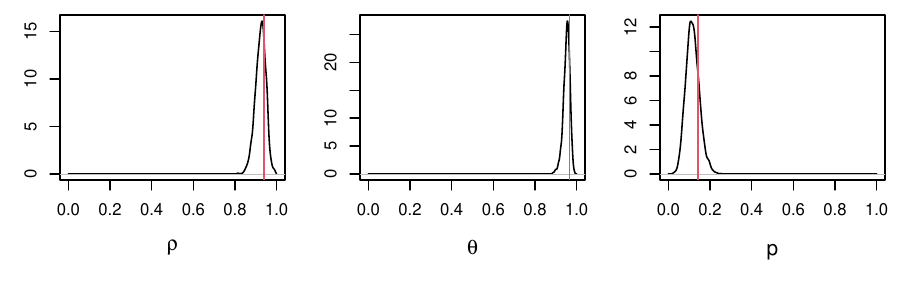}
    \caption{Synthetic Data. Posterior parameter densities for $\rho,\theta$ and $p$ from the unconstrained seniority effects analysis in \app~\ref{app:synth-poHB22aRS4}. Vertical line shows true value.}
    \label{fig:synth-rho-th-p-poHB22aRS4}
\end{figure}

\begin{figure}
    \centering
    \hspace*{-0.5in}\includegraphics[width=5in,trim=1in 3in 1.5in 3.65in, clip]{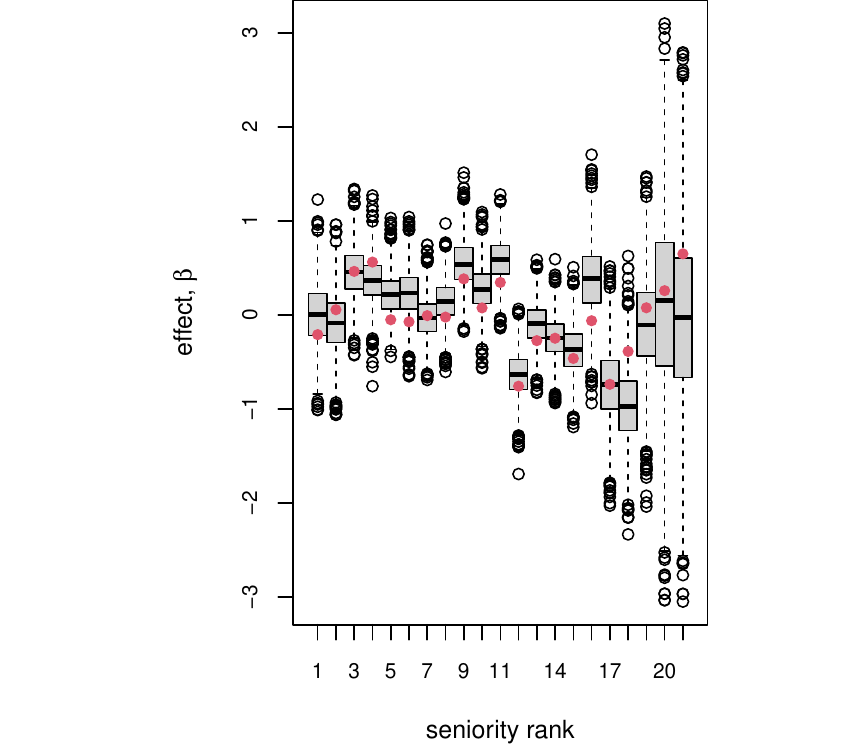}
    \caption{Synthetic Data. Marginal posterior distributions of seniority effect parameters from the unconstrained seniority effects analysis in \app~\ref{app:synth-poHB22aRS4}. Solid dots over boxes show true values.}
    \label{fig:synth-beta-poHB22aRS4}
\end{figure}

\begin{figure}
    \centering
    \hspace*{-0.8in}\includegraphics[width=6.75in,trim=0in 0.25in 0.25in 0.25in, clip]{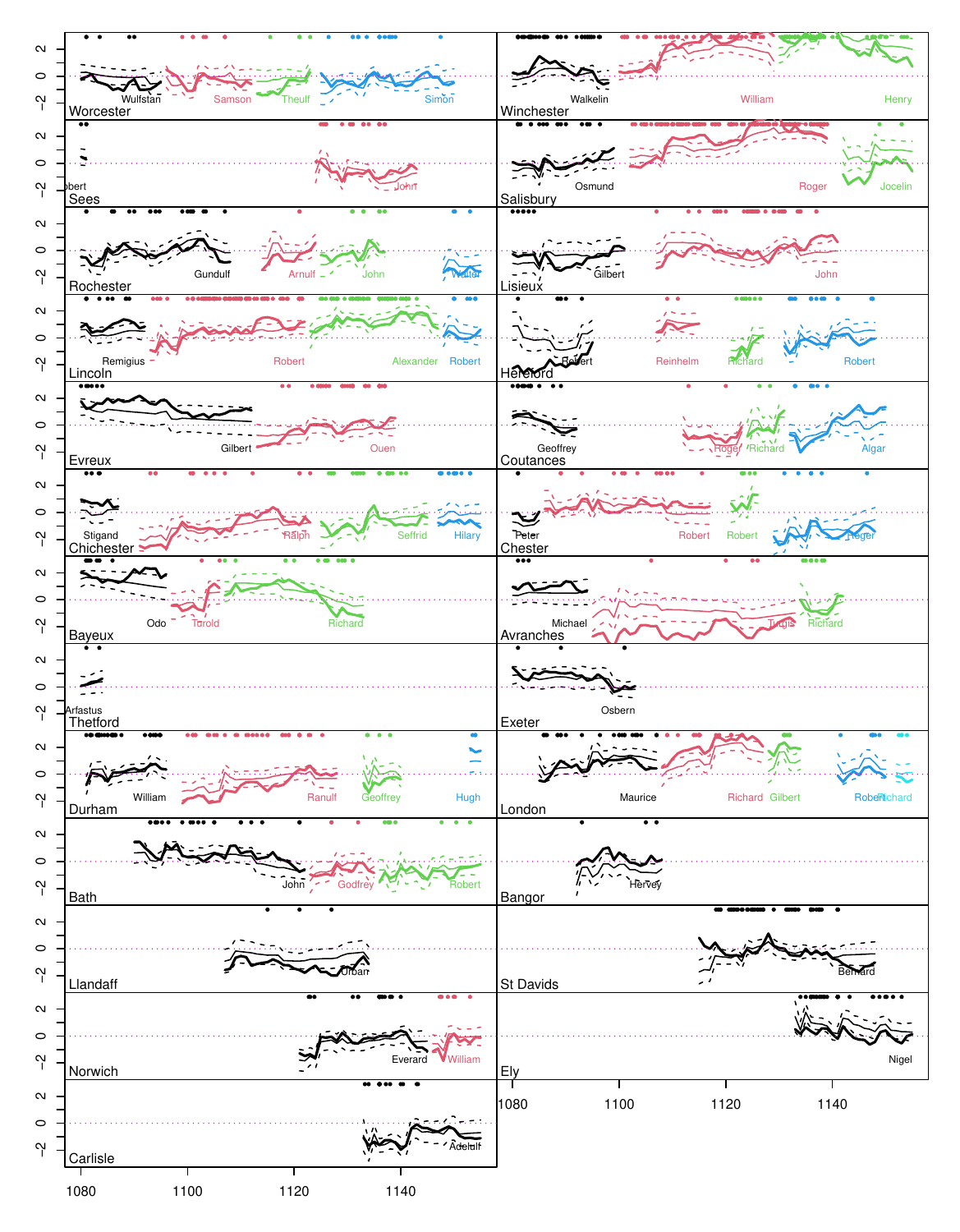}
    \caption{Synthetic Data. Bishop-status curves $\hat Z^{(t)}_j$ ($y$-axis, thin solid curves) plotted for each bishop $j\in \M$ as a function of time ($x$-axis, the year) from $b_j$ to $e_j$ with standard errors at one-sigma. The dots along the top of each graph show the times of the lists in which the color-matched bishop below appeared. This is output from the constrained seniority effects analysis in \sec~\ref{app:synth-poHB22aRS4}. The thick solid lines show the true values.}
    \label{fig:synth-Zhat-poHB22aRS4}
\end{figure}

\begin{figure}
    \centering
    \hspace*{-0.5in}\includegraphics[width=4.5in,trim=0in 0.2in 0.4in 0.7in, clip
    ]{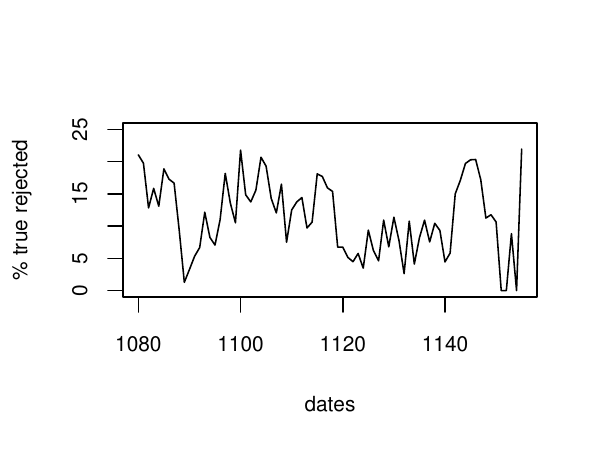}
    \caption{Synthetic data analysis discussed in \app~\ref{app:synth-poHB22aRS4}. The fraction of true relations rejected in the unconstrained seniority effects analysis of a synthetic data set simulated with $K=22$ as a function of year. }
    \label{fig:false-pos-neg-poHB22aRS4}
\end{figure}

\subsection{Checks on synthetic data: when the true orders are total orders}
\label{app:synth-total-orders}
In some feedback from historians a concern was expressed that the relatively shallow posets we see reflect differences in how often bishops witness and that the ``true'' orders were deeper. We tested this by simulating synthetic data in which the true posets all had depth close to or equal the maximum $m_t$ in each year $t\in [B,E]$, but using the same list memberships $o_i,\ i\in \I$ as the real data. For the purpose of this analysis we slightly reduced the data set, dropping some dioceses
(Sherborne, Exeter,  Bath,  Thetford, Wells, Le Mans, Bangor, St Asaph, Tusculum, 
the Orkneys and  Llandaff) that appear very infrequently in witness lists. This left $M=59$ bishops and the same number $N=371$ of lists.
  
We set the ``true'' noise at level $p^*=0.1$ typical of that seen in our fit. We used the same time intervals for the lists as in the real data so the synthetic data has the same structure as the real data, but different list orders. We simulated a ``true sequence'' $h^*={h^*}^{(t)},\ t\in [B,E]$ of high-depth orders (by setting $\rho^*=0.99999$ close to one) and true observation times $\tau^*$ from the prior, then simulated synthetic-data lists according to the observation model $y'_i\sim p_{(U)}(\cdot|{h^*}^{(\tau^*_i)}[o_i],p^*),\ i\in \I$. We then simulated the posterior $\pi(\rho,\theta,U,\beta,\tau,p|y')$ and checked consensus partial orders and checked we recovered the true depth. The posterior depth distribution is shown in \fig~\ref{fig:prior-post-depth-boxes-poHB22aRS4-poHB34aRS2} at right. The posterior shows good overlap with the truth despite it being on the boundary of the space.

\begin{figure}
    \centering
    \begin{tabular}{cc}
    \hspace*{-0.7in}\includegraphics[width=3.3in, trim=0.0in 0.2in 0.4in 0.7in, clip]{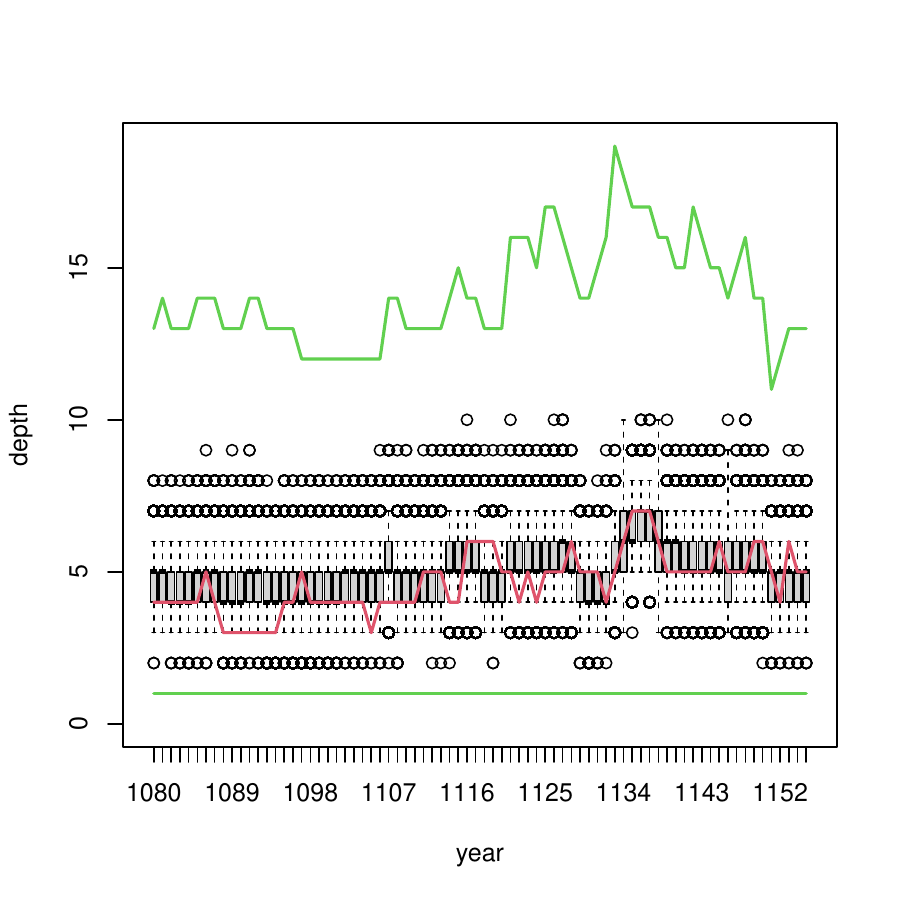}
    &
    \includegraphics[width=3.3in, trim=0.4in 0.2in 0.0in 0.7in, clip]{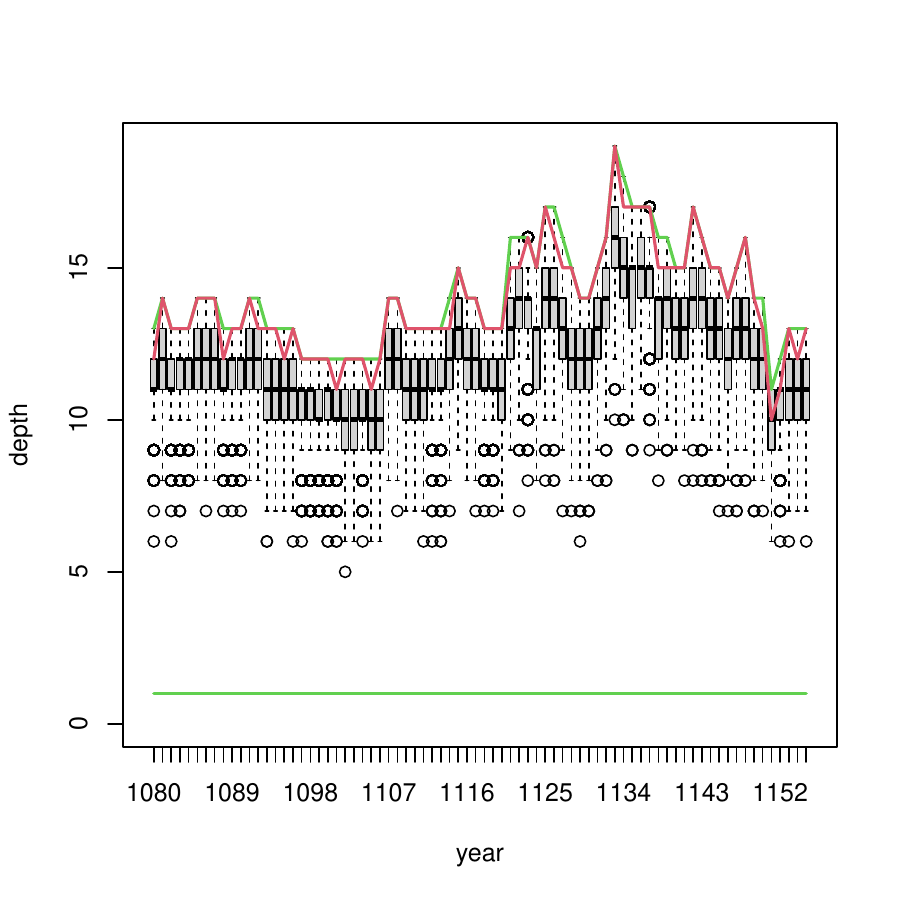}
    \end{tabular}
    \caption{Synthetic data taking the true order sequence $h^*$ to be a sample from the \sec~\ref{sec:poHB2aRS6}-posterior (left, with the slightly reduced data set used in this appendix) or a sequence of deep orders sampled from the prior (right). Fitting with prior settings as \sec~\ref{sec:poHB2aRS6} shows good overlap between truth and posterior. This suggests the prior is not biasing. The green lines above and below show the minimum and maximum possible poset depth in each year and the red line shows the depth of the true poset used to generate the synthetic data.}
    \label{fig:prior-post-depth-boxes-poHB22aRS4-poHB34aRS2}
\end{figure}

In summary, we reconstructed the true total orders well with good depth. If the true orders were total orders we would see this in our analysis. We do not, so we conclude that the uneven list memberships, dates and list-lengths are not obscuring some hidden deeper order.

It is of interest to consider what happens when the true poset has depth 1, so the poset is the empty order and there are no constraints at all.
 The data $y_i$ would then all be uniform random permutations of the list membership vectors $o_i,\ i\in \I$ and there is then a non-identifiability. The model can fit uniform random permutations by making $p\simeq 1$, so taking high noise, and any poset we like (as we get back uniform random permutations as $p\to 1$ irrespective of the underlying partial order) or by making the poset empty (so depth 1) and letting $p$ be any value. In some experiments which we do not report on synthetic data with depth 1 we found that the model favoured the $p\to 1$ fit (as the prior weights against depth 1).
Again, the fit for our data favours small noise, $p\simeq 0.1$ so by the same reasoning as for total orders, this non-identifiability is not an issue for the data we actually have. This has been called the ``principle of sufficient reason'': the prior is well behaved over the range of parameter values actually supported by the data.


\newpage\section{Fixed-time model}\label{app:fixed-time-posterior}

A simpler version of the model summarised in \sec~\ref{sec:time-series-posterior} outlined in \citet{nicholls11}, without any time-series structure, may be of interest so we give the simplest model of this sort. Covariates depending on time (and their effects $\beta$) have been dropped. They might be replaced with covariates which vary across lists. However we do not pursue this.

In this fixed-time setting we have $N$ lists with labels $i\in \I$, $M$ bishops with labels $j\in \M$,
$Z$ is an $M\times K$ feature matrix (not a time series of matrices) with rows $Z_j\in \R^K$ giving the $K$ status features for each bishop, and $H=h(Z)$ is a poset on $[M]$. All the lists have the same equal time, so $\tau$ is dropped. The generative model is
\begin{align*}
    \rho&\sim \mbox{Beta}(\gamma)\\
    Z_j&\sim N(0,\Sigma^{(\rho)}),\ \mbox{iid for}\ j\in\M\\
    p&\sim \mbox{Beta}(1,\delta),
\intertext{then $H=h(Z)$ using \eq~\ref{eq:z-to-h-mapping} and the data are realised}
    y_i&\sim p(\cdot|H[o_i],p),\quad \mbox{independently for $i\in \I$}.
\end{align*}
The joint posterior distribution is 
\begin{align}\label{eq:posterior-simple}
\pi(\rho,Z,p|y) \propto \pi(\rho,p)  \prod_{j\in\M} N(Z_j; 0,\Sigma^{(\rho)}) \prod_{i=1}^{N} p(y_i|H[o_i],p) \,. 
\end{align} 


\end{document}